\documentclass[11pt]{amsart}
\usepackage{amssymb,mathrsfs,graphicx}
\usepackage{hyperref,color}
\usepackage{amsmath}
\usepackage{amsfonts}
\usepackage{upref}
\usepackage{float}

\newtheorem{theorem}{Theorem}[section]

\newtheorem{corollary}[theorem]{Corollary}
\newtheorem{proposition}[theorem]{Proposition}

\newtheorem*{mydefinition}{Definition}
\newtheorem*{myremark}{Remark}

\def\llbracket{[}
\def\rrbracket{]}
\def\bp{{\mathbf p}}
\def\bq{{\mathbf q}}
\def\bx{{\mathbf x}}
\def\by{{\mathbf y}}
\def\bu{{\mathbf u}}
\def\bv{{\mathbf v}}
\def\bw{{\mathbf w}}
\def\bz{{\mathbf z}}
\def\br{{\mathbf r}}

\newcommand{\dm}[1]{{\boldsymbol [}\hspace{0.01cm}#1\hspace{0.01cm}{\boldsymbol ]}}  
\newcommand{\Dm}[1]{\begin{bmatrix}\displaystyle #1\end{bmatrix}} 
\newcommand{\conv}[1]{\overline{#1}}   					
\newcommand{\ave}[1]{{\langle}\hspace{0.01cm}#1\hspace{0.01cm}{\rangle}} 
\newcommand{\con}[1]{{#1}^\infty}  

\def\hf{\frac{1}{2}}
\def\D{{R}}	      
\def\Diam{{D_\infty}} 
\def\C{{\mathcal C}}  
\def\deg{{\sigma}}         
\def\P{{\mathcal P}}      
   
\def\Supp{Supp\{\phi(\cdot)\}}

\def\nei{q}           
\newcommand{\ceta}[1]{\eta_{{}_{#1}}}  
\newcommand{\Ceta}{\Theta}           
\newcommand{\veep}[1]{\vee_{#1}} 

\usepackage{fancyhdr}
 \fancypagestyle{plain}{
  \fancyhead{}
  
}

\def\paragraph#1{{\bf #1\ }}

\newcommand{\expo}{\mathrm{e}}

\usepackage[utf8]{inputenc}

\author[Sebastien Motsch]{Sebastien Motsch}
\address[Sebastien Motsch]{\newline
	School of Mathematical \& Statistical Sciences\newline
	Arizona State University\newline	
	Tempe, AZ 85287 USA}
\email[]{Sebastien.Motsch@asu.edu}
\urladdr{http://www.seb-motsch.com}
\author[Eitan Tadmor]{Eitan Tadmor}
\address[Eitan Tadmor]{\newline
  Center for Scientific Computation And Mathematical Modeling (CSCAMM)\newline
   and \newline
  Department of Mathematics,  Institute for Physical Science and Technology\newline
  University of Maryland, 
  College Park, MD 20742 USA}
\email[]{tadmor@cscamm.umd.edu}
\urladdr{http://www.cscamm.umd.edu/tadmor}

\begin{document}

\date{\today}

\subjclass{92D25,74A25,76N10}
\keywords{Agent-based models, self-alignment, heterophilious dynamics, clusters, consensus, flocking, active sets, connectivity of graphs, mean-field limits, kinetic equations, hydrodynamics}

\thanks{\textbf{Acknowledgments.} S.M. would like to thank the support of the Center for Scientific Computation And Mathematical Modeling (CSCAMM) where this research was performed. 
We thank the anonymous referee who brought to our  attention several references which helped improving an earlier version of the paper.
The work is supported by NSF grants DMS10-08397, RNMS11-07444 (KI-Net) and ONR grant N00014-1210318.}

\title[Heterophilious dynamics enhances consensus]{Heterophilious dynamics enhances consensus}

\setcounter{page}{1}

\begin{abstract}
We review a general class of models for self-organized dynamics based on alignment. 
The  dynamics of  such systems is   governed solely by interactions among individuals or ``agents'', with the tendency to adjust to their `environmental averages'. This,  in turn, leads to the formation of clusters, e.g., colonies of ants, flocks of birds, parties of people, rendezvous in mobile networks, etc.  A natural question which arises in this context is to understand when and how clusters emerge through the  self-alignment of agents, and what type of ``rules of engagement'' influence the formation of such clusters. Of particular interest to us are cases in which  the self-organized behavior tends to concentrate into \emph{one} cluster, reflecting  a \emph{consensus} of opinions, \emph{flocking} of birds, fish or cells, \emph{rendezvous} of mobile agents, and in general, concentration of other traits intrinsic to the dynamics.\newline
Many standard models for self-organized dynamics in social, biological and physical science assume that the intensity of alignment increases as agents get closer,  reflecting a common tendency  to align with those who think or act alike. Moreover, ``Similarity breeds connection,'' reflects our intuition  that  increasing the intensity of  alignment  as the difference  of positions  decreases, is more likely to lead to a consensus. We argue here that the converse is true: when the dynamics is driven by local interactions, it  is more likely to approach a consensus when the interactions among  agents \emph{increase} as a function of their difference in position. \emph{Heterophily} --- the tendency to bond more with those who are different rather than with those who are similar, plays a decisive r\^{o}le in the process of clustering.  We point out that the number of clusters in heterophilious dynamics \emph{decreases} as the heterophily dependence among agents increases. 
In particular, sufficiently strong heterophilious interactions enhance consensus. 
\end{abstract}

\maketitle

\centerline{\date}

\setcounter{tocdepth}{1}

\vspace*{-1.4cm}
{\small \tableofcontents}
\vspace*{-0.1cm}


\section{Introduction}
\setcounter{equation}{0}
\setcounter{figure}{0}

Nature and human societies offer many examples of self-organized behavior. Ants form colonies to coordinate the construction of  a new nest, birds form  flocks which fly in the same direction, mobile networks are sought to form a coordinated rendezvous and human crowd form  parties to reach a consensus when choosing a leader. The self-organized aspect of such systems is their dynamics, governed solely by interactions among its individuals or ``agents'', which tend to cluster into colonies, flocks, parties, etc.  A natural question which arises in this context is to understand when and how clusters emerge through the  self-interactions of agents, and what type of ``rules of engagement'' influence the formation of such clusters. Of particular interest to us are cases in which  the self-organized behavior tends to concentrate into \emph{one} cluster, reflecting  a \emph{consensus} of opinions, \emph{flocking} of birds, fish or cells, \emph{rendezvous} of mobile networks, and in general, \emph{concentration} around other positions intrinsic to the self-organized dynamics. Generically, we will refer to this process as concentration around an emerging consensus.

Many models have been introduced to appraise the emergence of consensus. Representative examples can be found in  \cite{blondel_krauses_2009,couzin_effective_2005, cucker_emergent_2007,duering_boltzmann_2009,hegselmann_opinion_2002,motsch_tadmor_new_2011,vicsek_novel_1995,ZEP_2011}, and we refer the reader to a  more comprehensive list of references surveyed in section \ref{sec:further} below. The starting point for our discussion is a general framework which embed several types of models describing self-organized dynamics. We consider the evolution of $N$ agents, each of which is identified by its ``position'' ${\bf p}_i(t)\in \mathbb{R}^d$. The position ${\bf p}_i(t)$ may account for  opinion, velocity, or other attributes of agent ``$i$'' at time $t$. Each agent adjusts its position according to the position of his neighbors:
\begin{equation}
  \label{eq:framework}
  \frac{d}{dt}{\bf p}_i = \alpha\,\sum_{j\neq i} a_{ij} ({\bf p}_j-{\bf p}_i), \qquad  a_{ij}\geq0.
\end{equation}
This provides a rather general description for  processes of \emph{alignment}. Here, $\alpha>0$ is a scaling parameter and the coefficients $a_{ij}$ quantify the strength of influence between  agents $i$ and $j$: the larger $a_{ij}$ is, the more weight is given to agent $j$ to align itself with agent $i$, based on the  difference of their positions ${\bf p}_i-{\bf p}_j$. 
The underlying fundamental assumption here is that agents do not react to the \emph{position} of others but  to their differences relative to other agents. In particular, the $a_{ij}$'s themselves are allowed to depend on the relative differences, ${\bf p}_i-{\bf p}_j$.  Indeed, we consider  nonlinear models (\ref{eq:framework}) where
\[
a_{ij}=a_{ij}(\P(t)), \qquad \P(t):=\{{\bf p}_{k}(t)\}_k.
\]
We emphasize the nonlinear aspect of the  alignment models (\ref{eq:framework}): the intricate aspect of
such models is  the nonlinear dependence of the influence matrix on the dynamics, $a_{ij}=a_{ij}(\P(t))$.  
 We ignore two other important processes involved in self-organized dynamics as advocated in the pioneering work of Reynolds,  \cite{reynolds_flocks_1987},  namely,  the short-range repulsion (or avoidance) and the long-range cohesion (or attraction), and we refer to  recent works  driven by the balance of these two processes in  e.g., \cite{BCL2009,DCBC_2006,EK_2001,Li_2008,MOA_2010,T-B_2004}. Our purpose here is to shed light on the role of mid-range alignment which covers the important zone ``trapped'' between the short-range attraction and long-range repulsion.

We distinguish between two main classes of self-alignment models. In the global case, the rules of engagement are such that every agent is influenced by every other agent, $a_{ij} > \eta >0$. The dynamics in this case is driven by \emph{global} interactions. We have a fairly good understanding of the large time dynamics of  such  models; an incomplete list of recent works in this direction includes \cite{blondel_convergence_2005,canizo_collective_2009,cucker_emergent_2007,DLBGL_2010,ha_simple_2009,ha_particle_2008,hegselmann_opinion_2002,jackson_2010,krause_discrete_2000,motsch_tadmor_new_2011}, and the references therein.  Global interactions which are sufficiently strong lead to \emph{unconditional consensus} in the sense that \emph{all} initial configurations of agents concentrate around an emerging limit state, the ``consensus'' $\con{\bp}$,
\[
{\bf p}_i(t) \stackrel{t \rightarrow \infty}{\longrightarrow} \con{\bp}.
\]
The first part of the paper, section \ref{sec:global}, contains an overview of the concentration dynamics in such global models, from the perspective of the  general framework of (\ref{eq:framework}).  

In more realistic models, however, interactions between agents are limited to their local neighbors, \cite{Aoki_1982,Ballerini_2008,couzin_2002,JE_2007,reynolds_flocks_1987}. The behavior of   \emph{local} models where some of the $a_{ij}$ may vanish,  requires a more intricate analysis. In the general scenario for such local models, discussed in section \ref{sec:clustering}, agents  tend to concentrate into one or more separate \emph{clusters}. The particular case in which agents concentrate into one cluster, that is the emergence of a consensus or a flock, depends on the propagation of \emph{uniform connectivity} of the underling (weighted) graph associated with the adjacency matrix, $\{a_{ij}\}$. This issue is explored in section \ref{sec:concon} where we show that connectivity implies consensus. Thus, the question of consensus for local models is turned into the question of persistence of connectivity over time. Note that even if the initial configuration is assumed connected, then there is still a possibility of losing connectivity as the $a_{ij}$'s may vary in time together with the positions $\P(t)$. The  open  question of tracing the propagation of connectivity in time for \emph{general} class of local models (\ref{eq:framework}) plays an important role in many applications, beyond the implication of emerging consensus. As an example we mention engineering applications to  sensor-based networks, from automatic traffic control and wireless communication to production systems and  mobile robot networks, e.g., \cite{He_2001,JE_2007,OS_2006,OSM_2004,Ring_2012,ZP_2007,ZEP_2011} and the references therein.

Many standard models for self-organized dynamics in social, biological and physical science assume that the dependence of $a_{ij}$ decreases as a function of  $|{\bf p}_i-{\bf p}_j|$, where $|\,\cdot\,|$ is a problem-dependent proper metric to measure a difference of positions, opinions, etc. 
The statement that ``Birds of feature flock together'' reflects a common tendency  to align with those who think or act alike, \cite{jackson_2010,merton_1954,McPherson_2001}. Moreover, ``Similarity breeds connection,'' reflects the intuitive scenarios in which   the  influence coefficients $a_{ij}$  increase as the difference  of positions $|{\bf p}_i-{\bf p}_j|$ decreases: the more the $a_{ij}$'s increase, the more likely it is to lead to a consensus. But in fact, we argue here that the converse is true: for a   self-organized dynamics driven by local interactions, it  is more likely to approach a consensus when the interaction among  agents \emph{increases} as a function of their difference $|{\bf p}_i-{\bf p}_j|$. \emph{Heterophily} --- the tendency to bond more with the different rather than with those who are similar, plays a decisive r\^{o}le in the clustering of (\ref{eq:framework}). The consensus in heterophilious dynamics  is explored in the second part of the paper, in terms of local interactions of the form, $a_{ij}=\phi(|\bp_i-\bp_j|)$, where $\phi(\cdot)$ is a compactly supported influence function which \emph{is increasing} over its support. In section \ref{sec:heterophily} we report our  extensive numerical simulations  which confirm the counter-intuitive phenomenon, where the number of clusters \emph{decreases} as the heterophilious dependence increases; in particular, if $\phi$ is increasing fast enough then the corresponding dynamics concentrate into one cluster, that is, heterophilious dynamics enhances consensus. 
We mention in passing the scenario of ``extreme heterophily" advocated in \cite{JE_2007,ZP_2007,ZEP_2011}, where distributed coordination is governed by local influence function $\phi(\cdot)$ which 
grows to infinity as it approaches the right edge of its support, in order to create an energy barrier which enforces connectivity and hence consensus. 
We are not unaware that this phenomenon 
of enhanced consensus in the presence of heterophilious interactions, may have intriguing  consequences in different areas other than social networks,  e.g., global bonding in atomic scales, avoiding materials' fractures in mesoscopic scales, or ``cloud'' formations  in macroscopic scales.  
 
 In the rest of the paper, we address a few important extensions of the self-alignment models outlined above. These extensions are still work in progress and we by no means try to be comprehensive.
In  section \ref{sec:nearest_neighbor} we turn our attention to nearest neighbor dynamics. Careful 3D observations made by   
 the StarFlag project, \cite{cavagna_et_al_2008i,cavagna_et_al_2008ii, cavagna_et_al_2010}, showed that interactions of birds are driven by \emph{topological} neighborhoods, involving a fixed number of nearby birds, instead of geometric neighborhoods involving a fixed radius of interaction. Here we prove that in the simplest case of two nearest neighbor dynamics, connectivity propagates in time and consensus follows for  influence functions which are non-decreasing on their compact support.
In section \ref{sec:discrete_dynamics} we turn our attention to \emph{fully discrete} models for self-alignment. The large time evolution in discrete time-steps, e.g., the opinion of dynamics in  \cite{blondel_convergence_2005,blondel_krauses_2009,krause_discrete_2000}, may depend on the time-step $\Delta t$. Here, we show that the  semi-discrete framework for global and local self-alignment outlined in sections \ref{sec:global}--\ref{sec:heterophily} can be extended, mutatis mutandis, to the fully-discrete case. In particular, we recover a decreasing Lyapunov functional, a fully-discrete analogue of the semi-discrete clustering analysis in section \ref{sec:local_non_symm_od}.
Finally, in section \ref{sec:hydro} we discuss the passage  from the agent-based description to mean-field limits as the number of agents, or ``particles" tends to be large enough. There is a growing literature on kinetic descriptions  of such models, \cite{carrillo_milling_2009, carrillo_asymptotic_2010a, carrillo_models_2010b,duering_boltzmann_2009,eftimie_2012,ha_particle_2008} and the references therein. Here we focus our attention on the hydrodynamic descriptions of self-organized opinion dynamics and flocking.  The closing section \ref{sec:further} is devoted to a more detailed discussion on the broader subject of self-organized dynamics.  Since a comprehensive review of this multidisciplinary  subject is beyond the scope of this paper, in particular, we  include a  selection of references, classified into several complementary categories of different disciplines, models, scales, approaches and patterns.

\subsection{Examples of opinion dynamics and flocking}

Models for self-organized dynamics (\ref{eq:framework}) have appeared in a large variety of different contexts, including load balancing in computer networks, evolution of languages, gossiping, algorithms for sensor networks, emergence of flocks, herds, schools and other biological ``clustering'', pedestrian dynamics, ecological models, peridynamic elasticity, multi-agent robots, models for opinion dynamics, economic networks and more; a detailed list of references is surveyed in section \ref{sec:further} below. 

To demonstrate the general framework for self-alignment dynamics (\ref{eq:framework}), we shall work with two main concrete examples. The first models  \emph{opinions dynamics}. In these models, $N$ agents, each  with vector of opinions quantified by  ${\bf p}_i \leadsto {\bx}_i \in \mathbb{R}^d$, interact with each other according to the first-order  system,
\begin{subequations}\label{eqs:opinion}
\begin{equation}
  \label{eq:opinion_formationa}
  \frac{d}{dt}{\bx}_i = \alpha\sum_{j\neq i} a_{ij} ({\bx}_j-{\bx}_i) \quad , \quad  a_{ij}  = \frac{\phi (|\bx_j-\bx_i|)}{N}.
\end{equation}
Here, $0<\phi<1 $ is the scaled \emph{influence function} which acts on  the ``difference of opinions", $|\bx_i-\bx_j|$. The metric $|\cdot|$ needs to be properly interpreted, adapted to the specific context of the problem at hand.
Another model for interaction of opinions is 
\begin{equation}
  \label{eq:opinion_formationb}
  \frac{d}{dt}{\bx}_i = \alpha\sum_{j\neq i} a_{ij} ({\bx}_j-{\bx}_i) \quad , \quad  a_{ij}  = \frac{\phi_{ij}}{\sum_k\phi_{ik}}, \quad \phi_{ij}:=\phi (|\bx_j-\bx_i|).
\end{equation}
\end{subequations}
The classical Krause model for opinion dynamics \cite{krause_discrete_2000,blondel_krauses_2009} is a time-discretization of (\ref{eq:opinion_formationb}), which will be discussed in section  \ref{sec:discrete_dynamics} below.
Observe that the  adjacency matrix $\{a_{ij}\}$ in the first model (\ref{eq:opinion_formationa}) is  symmetric while in the second model, (\ref{eq:opinion_formationb}), it is not.

Another branch of models have been proposed to describe \emph{flocking}. These are second-order models where the 
observed property is the velocity of birds,  ${\bf p}_i \mapsto {\bf v}_i \in \mathbb{R}^d$, which are coupled to their location ${\bf x}_i \in \mathbb{R}^d$.
The flocking model of Cucker and Smale (C-S) has received a considerable attention in recent years, \cite{cucker_emergent_2007,cucker_mathematics_2007,ha_particle_2008,canizo_collective_2009,ha_simple_2009}, 
\begin{subequations}\label{eqs:flocking}
\begin{equation}
  \label{eq:CS_model}
  \frac{d}{dt}{\bf v}_i = \alpha\,\sum_{j\neq i} a_{ij} ({\bf v}_j-{\bf v}_i)   \quad , \quad a_{ij} = \frac{\phi (|{\bf x}_j-{\bf x}_i|)}{N} \quad   \text{where} \quad  \frac{d}{dt}{\bf x}_i = {\bf v}_i.
\end{equation}
In C-S model, alignment is carried out by isotropic averaging. In \cite{motsch_tadmor_new_2011} we advocated a more realistic alignment-based model for flocking, where  alignment is based on the relative influence, similar to  (\ref{eq:opinion_formationb}), 
\begin{equation}
  \label{eq:MT_model}
  \frac{d}{dt}{\bf v}_i = \alpha\,\sum_{j\neq i} a_{ij} ({\bf v}_j-{\bf v}_i) \quad , \quad a_{ij} = \frac{\phi_{ij}}{\sum_k \phi_{ik}} \quad \text{with} \quad \phi_{ij} := \phi(|{\bf x}_j-{\bf x}_i|).
\end{equation}
\end{subequations}
Again, C-S model is based on a symmetric adjacency matrix, $\{a_{ij}\}$, while symmetry is lost in (\ref{eq:MT_model}),  i.e. $a_{ij} \neq a_{ji}$.

\begin{figure}[ht]
  \centering
  \includegraphics[scale=.38]{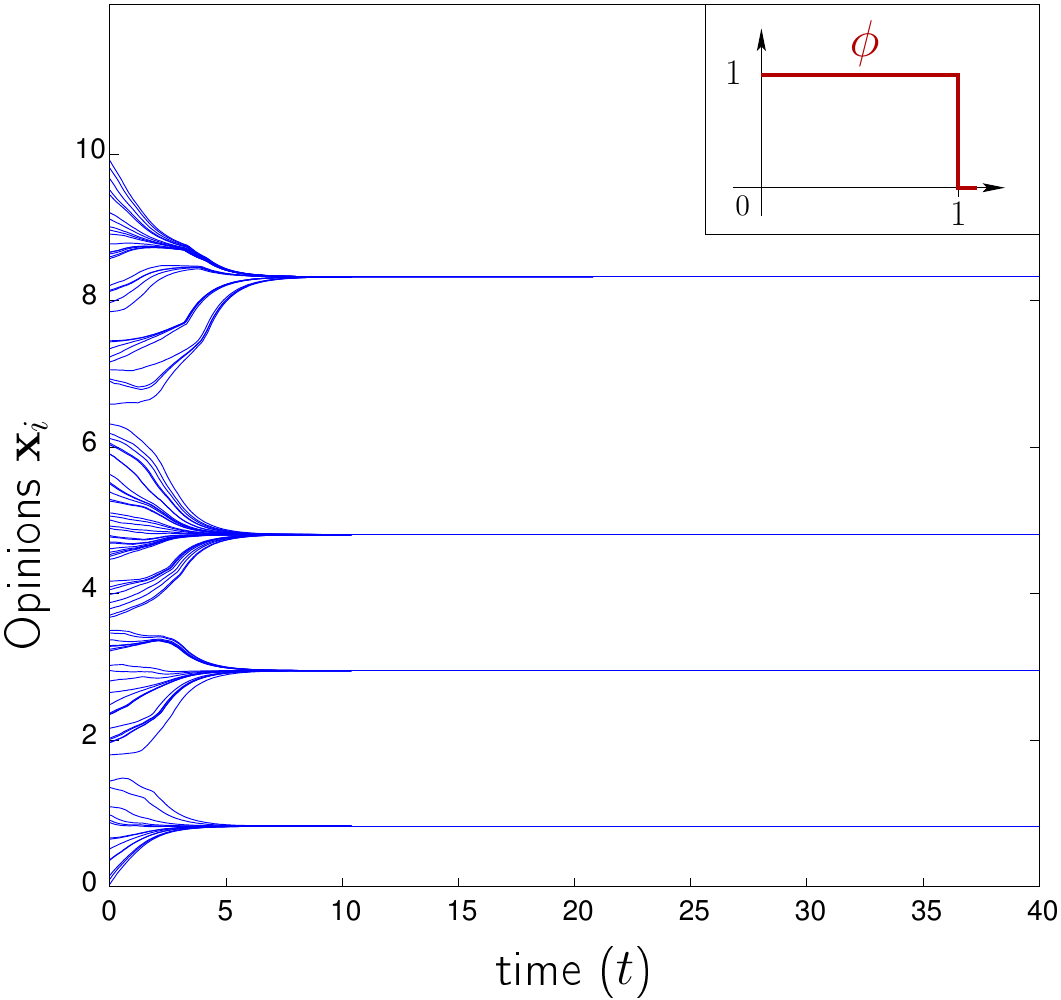}
  \includegraphics[scale=.38]{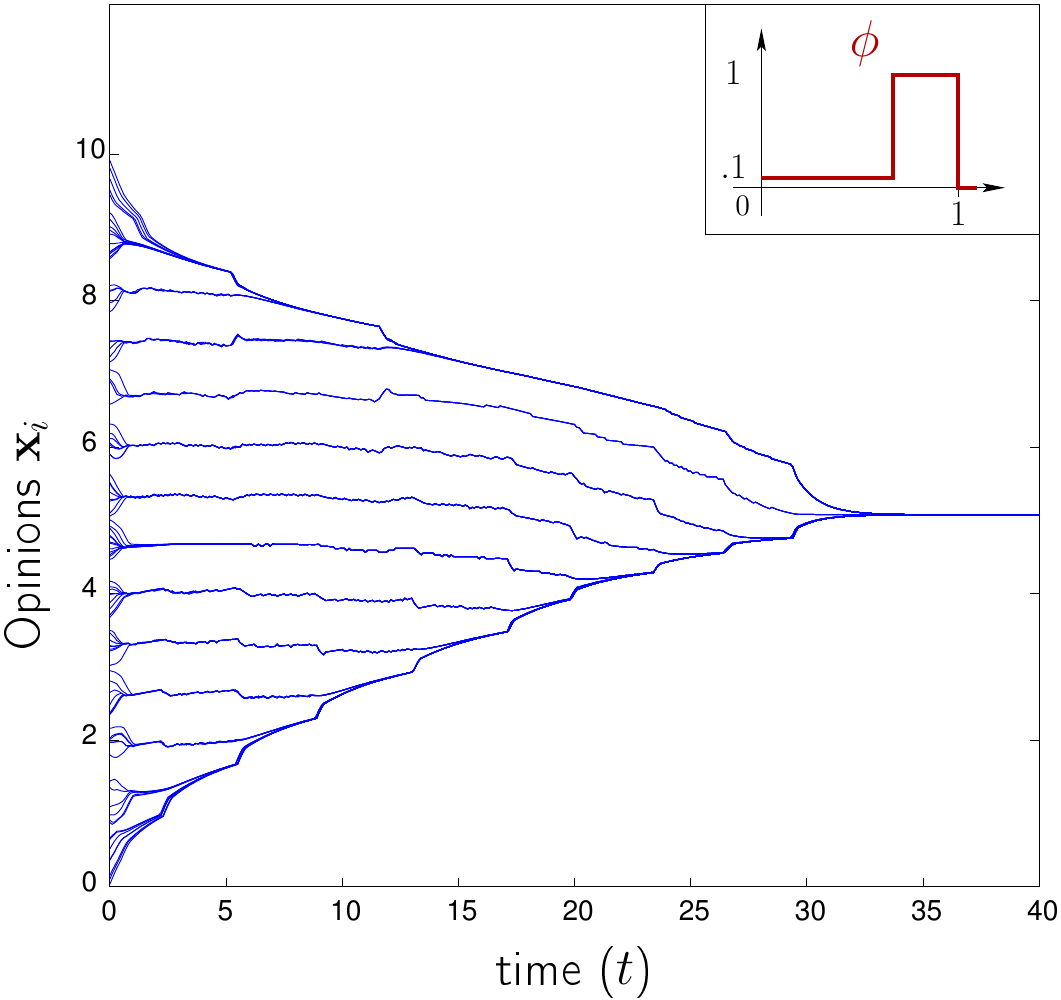}
  \caption{{\small Evolution in time of the consensus model for two different influence functions $\phi$ ({\bf Left} figure: $\phi(r)=\chi_{[0,1]}$, {\bf Right} figure: $\phi(r)= .1 \chi_{[0,1/\sqrt{2}]}+\chi_{[1/\sqrt{2},1]}$). By diminishing the influence of close neighbors ({\bf Right} figure), we enhance the emergence of a consensus. Simulations are started with the same initial condition ($100$ agents uniformly distributed on $[0,10]$).}}
  \label{fig:evolution_X_farVscloseNeighbor}
\end{figure}

The models for opinion and flocking dynamics (\ref{eqs:opinion}), and respectively, (\ref{eqs:flocking}), can be written in the unified form
\begin{equation}\label{eq:unified}
\frac{d}{dt} \bp_i = \alpha \sum_{j=1}^N a_{ij} (\bp_j-\bp_i), \qquad a_{ij}=\frac{1}{\deg_i}\phi(|\bx_i-\bx_j|).
\end{equation}
In the opinion dynamics, $\bp\mapsto \bx$; in the flocking dynamics, $\bp \mapsto \dot{\bx}$. The degree $\deg_i=N$
in the symmetric models, or $\deg_i=\sum_{j\neq i} \phi(|\bx_i-\bx_j|)$ in the non-symmetric models.
The local vs. global behavior of these models  hinges on the behavior of the  influence function, $\phi$. If the support of $\phi$ is large enough to cover the convex hull of $\P(0)=\{{\mathbf p}_k(0)\}_k$, then global interactions will yield unconditional consensus or flocking. On the other hand, if $\phi$ is locally supported, then the group dynamics in (\ref{eq:unified}) depends on the connectivity of the underlying graph, $\{a_{ij}\}$. In particular, if the overall connectivity is lost over time, then each connected component may lead to a separate cluster.   Heterophilious self-organized dynamics is characterized by a locally supported influence function, $\phi$, which  is \emph{increasing} as a function of the  mutual differences, $\phi_{ij}=\phi(|{\mathbf x}_i-{\mathbf x}_j|)$. The more heterophilious the dynamics is, in the sense that its influence function has a steeper increase over its compact support, the more it tends to concentrate in the sense of approaching  a smaller number of clusters.  In particular, heterophilious dynamics is more likely to lead to a consensus as  demonstrated for example, in figure \ref{fig:evolution_X_farVscloseNeighbor} (and is further documented in figures  \ref{fig:ratio011210} and  \ref{fig:S_depending_bDiva_2D} below). Observe that the \emph{only} difference between the two models depicted in figure \ref{fig:evolution_X_farVscloseNeighbor} is  that  the influence in the immediate neighborhood (of radius $r\leq 1/\sqrt{2}$)  was decreased, from $\phi=1\chi_{[0,1]}$ (on the left) into 
$\phi=0.1\chi_{[0,1/\sqrt{2}]}+ \chi_{[1/\sqrt{2},1]}$ (on the right): this was sufficient to  enhance the four-party clustering on the left to turn into a consensus shown on the right.

\section{Global interactions and unconditional emergence of consensus}\label{sec:global}
\setcounter{equation}{0}
\setcounter{figure}{0}

In this section we  derive explicit conditions for global self-organized dynamics (\ref{eq:framework}) to concentrate around an emerging   consensus. 
Our starting point is a  \emph{convexity argument} which is  valid for any adjacency  matrix $A=\{a_{ij}\}$, whether symmetric or not. We begin by noting  without loss of generality, that $A$ may be assumed to be row-stochastic,
\begin{equation}\label{eq:stoc}
\sum_j a_{ij} =1, \qquad i=1,\ldots, N. 
\end{equation}
Indeed, by rescaling $\alpha$ if necessary we have $\sum_{j\neq i} a_{ij}\leq 1$, and  (\ref{eq:stoc}) holds when we set   $a_{ii}:=1-\sum_{j\neq i} a_{ij} \geq 0$.
We rewrite (\ref{eq:framework}) in the form
\begin{equation}
  \label{eq:framework2}
  \frac{d}{dt}{\bp}_i = \alpha\,(\conv{\bp}_i-{\bp}_i) \;, \qquad \conv{\bp}_i := \sum_{j=1}^N a_{ij} {\bp}_j.
\end{equation}
Thus, if we let  $\Omega(t)$ denote the  convex hull of the properties $\{{\bp}_k\}_k$, then according to (\ref{eq:framework2}), ${\bp}_i$ is relaxing to the average value $\conv{\bp}_i\in \Omega(t)$, while the boundary of $\Omega$ is a \emph{barrier} for the dynamics. 
It follows that the  positions in the general self-organized model (\ref{eq:framework}) remain bounded. 

\begin{proposition}\label{ppo:Omega_decreasing_ct}
  The convex hull of ${\bf p}(t)$ is decreasing in time in the sense that the convex hull, $\Omega(t) := \text{Conv}\left(\{{\bp}_i(t)\}_{i\in\llbracket1,N\rrbracket}\right)$, satisfies
\begin{equation}\label{eq:Omega}
   \Omega(t_2)\,\subset\,\Omega(t_1), \quad t_2>t_1\geq 0.
\end{equation}
Moreover, we have
\begin{equation}\label{eq:decay}
\max_i|\bp_i(t)| \leq  \max_i|\bp_i(0)|.
\end{equation}
  \end{proposition}
\begin{proof}
We verify (\ref{eq:decay}) for a general vector norm $|\cdot|$ which we characterize in terms of its dual $|\bw|_*=\sup_{\bw\neq 0} \langle \bw,\bz\rangle/|\bz|$ so that $|\bp|= \sup \langle \bp,\bw\rangle/|\bw|_*$. Let $\bw=\bw(t)$ denotes the maximal dual vector of $\bp_i(t)$, so that $\langle \bp_i,\bw\rangle= |\bp_i|$, then
\[
\langle \dot{\bp}_i,\bw\rangle = \alpha\left(\langle \conv{\bp}_i,\bw\rangle - \langle \bp_i,\bw\rangle\right) \leq \alpha (|\conv{\bp}_i|-|\bp_i|).
\] 
Since $\langle \bp_i, \dot{\bw}\rangle \leq 0$ we have
\[
\frac{d}{dt}|\bp_i(t)| = \langle \dot{\bp}_i,\bw\rangle + \langle \bp_i, \dot{\bw}\rangle \leq \alpha (|\conv{\bp}_i(t)|-|\bp_i(t)|),
\]
and finally, $|\conv{\bp}_i(t)| \leq \max_i |\bp_i(t)|$ yields (\ref{eq:decay}).
\end{proof}

\begin{figure}[ht]
  \centering
  \includegraphics[scale=1]{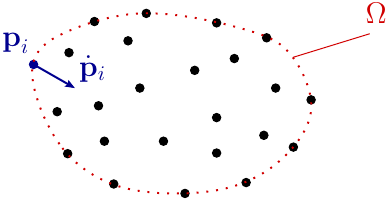}
  \caption{The convex hull $\Omega$ of the positions ${\bp}_i$.}
  \label{fig:convex_hull}
\end{figure}

\begin{myremark}
Since the models of opinion dynamics and flocking dynamics \eqref{eq:unified} are  translation invariant in the sense of admitting the family of solutions $\{\bp_i-{\mathbf c}\}$, then for any fixed state ${\mathbf c}$, proposition \ref{ppo:Omega_decreasing_ct} implies 
\[
\max_i|\bp_i(t)-{\mathbf c}| \leq  \max_i|\bp_i(0)-{\mathbf c}|.
\]
\end{myremark}

Consensus and flocking are achieved when the decreasing $\Omega(t)$  shrinks to a limit point $\Omega(t) \stackrel{t\rightarrow \infty}{\longrightarrow} \{\con{\bp}\}$,
\[
\max_i|\bp_i(t)-\con{\bp}| \stackrel{t\rightarrow \infty}{\longrightarrow} 0.
\]
There are  various approaches, not unrelated, to derive conditions which ensure unconditional consensus or flocking. We shall mention two: an $L^\infty$ contraction argument and an $L^2$ energy method based on spectral analysis.

\subsection{An $L^\infty$ approach: contraction of diameters}\label{sec:emer}

Proposition  \ref{ppo:Omega_decreasing_ct} tells us that  $\{{\bp}_i(t)\}_{i}$ remain uniformly bounded and the diameter, 
$\max_{ij} | {\bp}_i(t)-{\bp}_j(t)|$, is non-increasing in time.
In order to have concentration, however, we need to verify that the  diameter of ${\bp(t)}$ decays to zero. The next proposition quantifies this decay rate.

\begin{theorem}\label{thm:contract}
Consider  the self-organized model \eqref{eq:framework} with a raw stochastic adjacency matrix  $A$,  \eqref{eq:stoc}. 
Let 
\[
\dm{\bp}:=\max_{ij}|\bp_i-\bp_j|
\]
denote the diameter of the position vector $\bp$.
Then the diameter  satisfies the concentration estimate
\begin{equation}\label{eq:contract}
\frac{d}{dt}\dm{\bp(t)} \leq -\alpha\ceta{{A(\P(t))}} \dm{\bp(t)}, \qquad \ceta{A}:=\min_{ij}\sum_k \min\{a_{ik},a_{jk}\}.
\end{equation}
In particular, if there is a slow decay of the  concentration factor  so that 
$\displaystyle \int^\infty \!\!\!\!\!\ceta{{A(\P(s))}}\,ds = \infty$, then the agents concentrate in the sense that
\begin{subequations}\label{eqs:pconsen}
\begin{equation}\label{eq:qconsen}
\Ceta(t):= \int^t \!\!\!\ceta{{A(\P(s))}}\,ds \stackrel{t \rightarrow \infty}{\longrightarrow}  \infty \ \ \quad \leadsto \ \ \quad \lim_{t\rightarrow \infty}\max_{i,j}|\bp_i(t)-\bp_j(t)|=0.
\end{equation}
Moreover, if the decay of the concentration factor is slow enough in the sense that $\int^\infty \exp(-\alpha \Ceta(s))ds<\infty$, then there is an emerging consensus $\con{\bp} \in \Omega(0)$,
\begin{equation}\label{eq:pconsen}
\int^\infty \!\!e^{-\alpha \Ceta(t)}dt < \infty \ \ \quad \leadsto \ \ \quad |\bp_i(t)-\con{\bp}| \lesssim e^{-\alpha \Ceta(t)}\dm{\bp(0)} \ \ \quad   \text{for all} \ \ i=1,\ldots,N.
\end{equation}
\end{subequations}
\end{theorem}

\begin{myremark}
We note that theorem \ref{thm:contract} applies to any vector norm $|\cdot|$.
\end{myremark}

\begin{proof}
We begin with the following  estimate which quantifies the contractivity of the row stochastic $A$ in the induced vector \emph{semi}-norm $\dm{\cdot}$ (since this bound  is solely due to the convexity  of the row stochastic $A$,  we suppress the time-dependence of $\bp$ and $\conv{\bp}=A\bp$),

\begin{equation}\label{eq:dpcontract}
\dm{A\bp} \leq (1-\ceta{_A})\dm{\bp}, \qquad \dm{\bp}=\max_{ij}|\bp_i-\bp_j|, \quad 1-\ceta{_A}=\frac{1}{2}\sum_k|a_{ik}-a_{jk}|.
\end{equation}

The estimate (\ref{eq:dpcontract}) in its $\ell^1$-dual form for column stochastic matrices goes back to Dobrushin \cite{Dob_1956},  and his so-called coefficient of ergodicity, $\ceta{_A}$,  was later used to quantify the relative entropy in discrete Markov processes \cite{CDZ_1993,CIRRSZ_1993} and the contractivity in models of opinion dynamics \cite{krause_discrete_2000}.
For completeness, we proceed with the proof for general vector norms $|\cdot|$. Fix any $i$ and $j$ which are to be chosen later, and set $\eta_k := \min\{a_{ik},a_{jk}\}$ so that $a_{ik}-\eta_k$ and $a_{jk}-\eta_k$ are non-negative. Then, for arbitrary $\bw\in{\mathbb R}^d$ we have, 
\begin{eqnarray*}
\langle \conv{\bp}_i-\conv{\bp}_j, \bw\rangle  & = & \sum_k a_{ik}\langle {\bp}_k, \bw\rangle- \sum_k a_{jk}\langle {\bp}_k,\bw\rangle \\
&= &  \sum_k (a_{ik}-\eta_k)\langle {\bp}_k,\bw\rangle- \sum_k (a_{jk}-\eta_k)\langle {\bp}_k, \bw\rangle \\
& \leq & \sum_k (a_{ik}-\eta_k)  \max_k \langle {\bp}_k,\bw\rangle- \sum_k (a_{jk}-\eta_k)\min_k \langle {\bp}_k, \bw\rangle \\
& = &
(1-\ceta{_A})\left(\max_k \langle {\bp}_k,\bw\rangle - \min_k \langle {\bp}_k,\bw\rangle\right ) \\
& \leq & (1-\ceta{_A})\max_{k\ell}\langle \bp_k-\bp_\ell,\bw\rangle 
 \leq  (1-\ceta{_A})\max_{k,
\ell}|\bp_k-\bp_\ell| |\bw|_*.
\end{eqnarray*}  
In the last step, we characterize the norm $|\cdot|$ by its dual $|\bw|_*=\sup_{\bw\neq 0} \langle \bw,\bz\rangle/|\bz|$ so that 
$\langle \bz,\bw\rangle \leq |\bz| |\bw|_*$. Now, choose $i$ and $j$ as a maximal pair such that
$\dm{\conv{\bp}}=|\conv{\bp}_i-\conv{\bp}_j|$; we then  have
\[
\dm{A\bp}\equiv \dm{\conv{\bp}} =|\conv{\bp}_i-\conv{\bp}_j|=\sup_{\bw\neq 0}  \frac{\langle \conv{\bp}_i-\conv{\bp}_j, \bw\rangle}{|\bw|_*} \leq (1-\ceta{_A})\max_{k,\ell}|\bp_k-\bp_\ell|
\]
and (\ref{eq:dpcontract}) now follows.\newline 

\noindent
Next, we consider the discrete time-marching system associated with (\ref{eq:framework}),
\[
\frac{\bp(t+\Delta t)-\bp(t)}{\Delta t} = \alpha\left(A\bp(t)-\bp(t)\right).
\]
Using (\ref{eq:dpcontract}) we obtain 
\[
\dm{\bp(t+\Delta t)} = \dm{(1-\alpha\Delta t)\bp(t) +  \alpha\Delta t\, A\bp(t)}\leq 
(1-\alpha\Delta t)\dm{\bp(t)} + \alpha\Delta t(1-\ceta{_A})\dm{\bp(t)},
\]
or after rearrangement,
\[
\frac{\dm{\bp(t+\Delta t)} -\dm{\bp(t)}}{\Delta t}
\leq -\alpha\ceta{_A}\dm{\bp(t)},
\]
and the desired bound (\ref{eq:contract}) follows by letting $\Delta t \rightarrow 0$.
\ifx
Next, we rewrite (\ref{eq:framework}) in the form,
\[
\frac{d}{dt} e^{\alpha t} \bp(t) = \alpha e^{\alpha t} A\bp(t).
\]
Using (\ref{eq:dpcontract}) we obtain 
\[
\frac{d}{dt} e^{\alpha t}\dm{\bp(t)} \leq \Dm{\frac{d}{dt} e^{\alpha t}\bp(t)} = \alpha e^{\alpha t} \dm{A\bp(t)} \leq \alpha e^{\alpha t}(1-\ceta{_A})\dm{\bp(t)},
\]
 which proves the desired bound (\ref{eq:contract}). 
\fi
In particular, we have
\begin{equation}\label{eq:xgd}
\max_{ij}|\bp_i(t)-\bp_j(t)| \leq \exp\,\left(-\alpha \int_0^t\!\!\ceta{{A(\P(s))}}ds\right) \dm{\bp(0)} \stackrel{t \rightarrow \infty}{\longrightarrow} 0,
\end{equation}
which proves (\ref{eq:qconsen}). Moreover,  
\begin{eqnarray*}
|\bp_i(t_2)-\bp_i(t_1)| &= &\left|\int_{\tau=t_1}^{t_2} \dot{\bp}_i(\tau)\,d\tau\right| \leq \alpha \max_{ij}\int_{\tau=t_1}^{t_2} |\bp_i(\tau)-\bp_j(\tau)|\,ds \\
& \leq & \alpha   \int^{t_2}_{\tau=t_1}\exp\,\left(-\alpha \Ceta(\tau)\right) d\tau \,\dm{\bp(0)}, \qquad \Ceta(\tau)=\int_0^\tau \ceta{A(\P(s)}ds,
\end{eqnarray*}
which tends to zero, $ |\bp_i(t_2)-\bp_i(t_1)| \rightarrow 0$ for $t_2>t_1 \gg1$, thanks to our assumption (\ref{eq:pconsen}).
It follows that the limit  $\bp_i(t) \stackrel{t\rightarrow \infty}{\longrightarrow} \con{\bp}_i$ exists, and hence all agents concentrate around  the same limit position,  an emerging consensus $\con{\bp} \in \Omega(0)$. The concentration rate  estimate  (\ref{eq:pconsen}) follows from  (\ref{eq:xgd}).
\end{proof}

Theorem \ref{thm:contract} relates  the emergence of consensus or flocking of $\dot{\bp}=A\bp-\bp$ to the behavior of $\int^t \ceta{{A(\P(s))}}ds \uparrow \infty$, and to this end we seek lower-bounds on the ``concentration factor'' $\ceta{A}$, which are easily checkable in terms of the entries of $A$. This brings us to the following definition.

\begin{mydefinition}[Active sets \cite{motsch_tadmor_new_2011}] Fix $\theta>0$. The  active set, $\Lambda(\theta)$, is the set of agents which influence every other agent ``more'' than $\theta$,
  \begin{equation}
    \label{eq:active_set_p}
    \Lambda(\theta) := \{ j \ \big| \ a_{ij} \geq \theta \ \text{for any } \ i\}.
  \end{equation}
  Observe that since $a_{ij}$ changes in time, $a_{ij}=a_{ij}(\P(t))$,  the number of agents in the active set $\Lambda(\theta)$ is a time dependent quantity, denoted $\lambda(\theta) =\lambda(\theta,t):= |\Lambda(\theta,t)|$. 
\end{mydefinition}

\noindent
The straightforward lower bound, $\ceta{_A} \geq \max_\theta \theta\cdot\lambda(\theta)$  yields

\begin{corollary}\label{cor:act}
The diameter  of the self-organized model \eqref{eq:framework} with a stochastic adjacency matrix  $A$, \eqref{eq:stoc}
satisfies the concentration estimate 
\begin{equation}\label{eq:act}
\frac{d}{dt}\dm{\bp(t)} \leq -\alpha (\max_\theta \theta\cdot\lambda(\theta,t))\, \dm{\bp(t)}.
\end{equation}
In particular, the lower  bound $\ceta{_A}\geq N\min_{ij}a_{ij}$, corresponding to $\theta=\min_{ij}a_{ij}$ with $\lambda(\theta,t)=N$,  yields \cite{ha_particle_2008}
\begin{equation}\label{eq:Na}
|\bp(t)-\con{\bp}| \lesssim \exp\left(-\alpha N \int^t_0 m(s)ds\right)\dm{\bp(0)}, \qquad m(s):= \min_{ij}a_{ij}(s).
\end{equation}
\end{corollary}

\begin{myremark}
The  bound \eqref{eq:act} is an improvement of the ``flocking'' estimate \cite[Lemma 3.1]{motsch_tadmor_new_2011}
\[
\frac{d}{dt}\dm{\bp(t)} \leq -\alpha (\max_\theta \theta\cdot\lambda(\theta,t))^2\, \dm{\bp(t)}.
\]
\end{myremark}

Corollary \ref{cor:act} is a useful tool to verify consensus and flocking behavior for general adjacency matrices $A=\{a_{ij}\}$, whether symmetric or not. We demonstrate its application with the following sufficient condition for the emergence of a consensus in the opinion models (\ref{eqs:opinion}). In either the symmetric or non-symmetric case,
\[
    a_{ij} = \left\{\begin{array}{c}\displaystyle \frac{\phi_{ij}}{N}\\ \\ \displaystyle \frac{\phi_{ij}}{\deg_i}\end{array}\right\}  \geq \frac{\phi(\dm{\bx(t)})}{N}, \qquad \deg_i=\sum_k\phi_{ik}\leq N. 
\]
By proposition \ref{ppo:Omega_decreasing_ct},  the diameter $\dm{\bx(t)}$  is non-increasing, yielding the lower bound
\[
N a_{ij}(\P(t)) = \frac{N}{\deg_i}\phi(|\bx_i(t)-\bx_j(t)|) \geq \min_{r\leq \dm{\bx(t)}}\phi(r) \geq \min_{r\leq \dm{\bx(0)}}\phi(r),
\]
which in turns implies the following exponentially fast convergence towards a consensus $\con{\bx}$.

\begin{proposition}[Unconditional consensus]\label{prop:CV_opinions}
Consider the models for opinion dynamics \eqref{eqs:opinion} with an influence function $\phi(r)\leq 1$, and assume that  
\begin{equation}\label{eq:mispos}
m:=\min_{r\leq \dm{\bx(0)}} \phi(r)>0.
\end{equation}
Then, there is an exponentially fast convergence towards an emerging consensus $\con{\bx}$, 
\begin{equation}
    \label{eq:x_i_t_cv}
  |\bx_i(t)-\con{\bx}| \lesssim e^{-\alpha m t} \dm{\bx(0)}.
\end{equation}
\end{proposition}

Similar arguments apply for the flocking models (\ref{eqs:flocking}):
 since $\dm{\bv(t)}$ is non-increasing
then $\dm{\bx(t)} \leq \dm{\bx(0)} + t\dm{\bv(0)}$ and hence
\[
N a_{ij}(\P(t)) = \frac{N}{\deg_i}\phi(|\bx_i(t)-\bx_j(t)|) \geq \min_{r\leq \dm{\bx(t)}}\phi(r) \geq \min_{r\leq \dm{\bx(0)} + t\dm{\bv(0)}}\phi(r);
\]
if $\phi(\cdot)$ is decreasing then we can set $m(t)=\phi(\dm{\bx(0)} + t\dm{\bv(0)})$ and unconditional flocking follows from for
corollary \ref{cor:act} for   sufficiently strong interaction so that $\int^\infty \phi(s)ds=\infty$.
In fact, a more precise statement of flocking is summarized in the following.

\begin{proposition}[Unconditional flocking]\label{prop:CV_flocking}
Consider the flocking dynamics \eqref{eqs:flocking} with a decreasing influence function $\phi(r)\leq \phi(0)\leq 1$, and assume that  
\begin{equation}\label{eq:non_int}
\int^\infty \!\!\phi(s)ds=\infty;
\end{equation}
Then, the diameter of positions  remains uniformly bounded, $\dm{\bx(t)} \leq \Diam <\infty$, and there is an exponentially fast concentration of velocities around a flocking state $\con{\bv}$, 
\begin{equation}
    \label{eq:v_i_t_cv}
|{\bf v}_i(t)-\con{\bv}| \lesssim e^{-\alpha m t}\dm{\bv(0)}, \qquad m=\phi(\Diam).
\end{equation}
\end{proposition}

\begin{proof}
Unlike the first-order models for consensus, the diameter  in second-order flocking models, $\dm{\bx(t)}$, may increase in time. The bound  $\Diam$ stated in \eqref{eq:v_i_t_cv} places a uniform bound on the maximal active diameter. To derive such a  bound observe that in the second order flocking models,   the evolution of the diameter of velocities satisfies, 
\[
    \frac{d}{dt}\dm{\bv(t)}  \leq  -\alpha \phi(\dm{\bx(t)})\dm{\bv(t)},
\]
 and is coupled with the evolution of positions $\dm{\bx(t)}$:  since $\dot{\bx} = {\bv}$, we have 
\[
\frac{d}{dt}\dm{\bx(t)} \leq \dm{\bv(t)}.
\]
The last two inequalities imply that the following energy functional introduced by Ha and Liu \cite{ha_simple_2009}, 
\[
\mathcal{E}(t) := \dm{\bv(t)} + \alpha\int_0^{\dm{\bx(t)}} \phi(s)ds,
\]
 is decreasing in time, 
\begin{equation}\label{eq:hl}
\alpha  \int_{\dm{\bx(0)}}^{\dm{\bx(t)}} \phi(s)\,ds \leq \dm{\bv(0)}-\dm{\bv(t)}\leq \dm{\bv(0)}.
\end{equation}
This, together with our assumption \eqref{eq:non_int} yield  the existence of  a finite $\Diam> \dm{\bx(0)}$ such that
\begin{equation}\label{eq:non_intii}
\alpha  \int_{\dm{\bx(0)}}^{\dm{\bx(t)}} \phi(s)\,ds \leq \dm{\bv(0)} \leq \alpha \int_{\dm{\bx(0)}}^{\Diam} \phi(s)\,ds.
\end{equation}
Thus, the active diameter of positions does not exceed $\dm{\bx(t)}\leq \Diam$, and since $\phi$ is assumed decreasing, the minimal interaction is $Na_{ij}\geq \phi(\dm{\bx(t)}) \geq \phi(\Diam)$ which yields 
\[
\frac{d}{dt}\dm{\bv(t)}\leq -\alpha \phi(\Diam)\dm{\bv(t)}. 
\]
This concludes the proof of \eqref{eq:v_i_t_cv}.
\end{proof}

\begin{myremark}[Global interactions]
Proposition \ref{prop:CV_opinions} derives an unconditional consensus under the assumption of \emph{global interaction}, namely, according to  \eqref{eq:mispos}  every agent interacts with every other agent as 
\[
a_{ij}\geq \frac{1}{N}\phi(|\bx_i-\bx_j|)\geq \frac{m}{N}>0.
\] 
Similarly, the unconditional flocking stated in proposition \ref{prop:CV_flocking} requires  global interactions, in the sense of having an influence function \eqref{eq:non_int} which is supported over the entire flock. Indeed, if the influence function $\phi$ is compactly supported, $supp\{\phi\}=[0,\D]$, then  assumption \eqref{eq:non_int}  tells us that
\[
 \dm{\bv(0)} \leq \alpha \int_{\dm{\bx(0)}}^{\D} \phi(s)\,ds;
\]
but according to \eqref{eq:hl}, $\displaystyle \alpha  \int_{\dm{\bx(0)}}^{\dm{\bx(t)}} \phi(s)\,ds \leq \dm{\bv(0)}$ and hence the support of $\phi$ remains larger than the diameter of positions, $\D\geq \dm{\bx(t)}$.
\end{myremark}

Proposition \ref{prop:CV_flocking} recovers the unconditional flocking results for the  C-S model, $\phi(r)\propto (1+r)^{-2\beta}, \beta>1/2$, obtained earlier using  spectral analysis, $\ell_1$-, $\ell_2$- and $\ell_\infty$-based estimates  \cite{cucker_emergent_2007,ha_particle_2008,canizo_collective_2009,ha_simple_2009,carrillo_asymptotic_2010a}. 
The derivations are different, yet they all required the symmetry of the C-S influence matrix, $a_{ij}=\phi_{ij}/N$. 
Here, we unify and generalize the results, covering both the symmetric and non-symmetric scenarios. In particular,  we  improve here the unconditional flocking result in the non-symmetric model obtained in  \cite[theorem 4.1]{motsch_tadmor_new_2011}. Although the tools are different --- notably, lack of conservation of momentum $\frac{1}{N} \sum_i {\bv}_i(t)$ in the non-symmetric case, we nevertheless end up with same condition (\ref{eq:non_int}) for unconditional flocking.  

\subsection{Spectral analysis of symmetric models}\label{sec:symemer}

A more precise description of the concentration phenomena is available for models governed by   \emph{symmetric} influence matrices,  $a_{ij}=a_{ji}$, such as (\ref{eq:opinion_formationa}) and (\ref{eq:CS_model}). 
Set $\bq_i=\bp_i-\ave{\bp}$ where $\ave{\bp}:=1/N\sum_i \bp_i$  is  the average (total momentum), which thanks to symmetry is conserved in time,
$\dot{\ave{\bp}}(t) \propto \sum_{ij} a_{ij}(\bp_i-\bp_j)=0$, and hence the symmetric system (\ref{eq:framework}) reads
\[
\frac{d}{dt} \bq_i(t)= \alpha \sum_{j=1}^N a_{ij}(\bq_j-\bq_i), \qquad \bq_i:=\bp_i-\ave{\bp}.
\]
Let $L_A:=I-A$ denote the \emph{Laplacian matrix} associated with $A$, with ordered eigenvalues $0=\lambda_1(L_A)\leq \lambda_2(L_A) \leq \ldots \lambda_N(L_A)$. The following estimate is at the heart of matter (here $|\cdot|$ denotes the usual Euclidean norm on ${\mathbb R}^d$),
\begin{equation}\label{eq:Vbound}
\hspace*{1.3cm} \, \frac{1}{2}\frac{d}{dt} \sum_i|\bq_i(t)|^2 = \alpha \sum_{i,j} a_{ij}\langle \bq_j-\bq_i, \bq_i\rangle
= -\frac{\alpha}{2}\sum_{ij}a_{ij}|\bq_j-\bq_i|^2  
 \leq \! -\alpha \lambda_2(L_A) \sum_{i}|\bq_i(t)|^2. 
\end{equation}
The second equality is a straightforward consequence of $A$ being symmetric; the following inequality follows from the Courant-Fischer characterization of the second eigenvalue of $L_A$ in terms of vectors ${\bq}$ orthogonal to the first eigenvector  ${\bf 1}=(1,1,\ldots,1)^\top$,
\begin{equation}\label{eq:fidratio}
\lambda_2(L_A)= \min_{\sum \bq_k=0}\frac{\langle L_A\bq,\bq\rangle}{\langle\bq,\bq\rangle}\leq  
\frac{(1/2)\sum_{ij} a_{ij}|\bq_i-\bq_j|^2}{\sum_{i}|\bq_i|^2}.
\end{equation}
We end up with the following sufficient condition for the emergence of unconditional concentration.

\begin{theorem}[Unconditional concentration in the symmetric case]\label{thm:symm}
Consider  the self-organized model \eqref{eq:framework},\eqref{eq:stoc} with a symmetric adjacency matrix  $A$. 
Then the following concentration estimate holds
\begin{equation}\label{eq:symm_contract}
\veep{\bp(t)} \leq \exp\left(-{\alpha}\int^t \lambda_2(L_{A(\P(s))})ds\right) \veep{\bp(0)}, \qquad \vee^2_{\bp(t)}:= \frac{1}{N}\sum_{i}|\bp_i(t)-\ave{\bp}(0)|^2.
\end{equation}
In particular, if the interactions remain  ``sufficiently strong'' so that $\int^\infty \lambda_2(L_{A(\P(s))})ds  =\infty$, then there is convergence towards consensus $\bp_i(t) \rightarrow \con{\bp} = \ave{\bp}(0)$.
\end{theorem}

\medskip
To apply theorem \ref{thm:symm},   we need to trace effective lower bounds on $\lambda_2(L_A)$; here are two examples which recover our previous results in section \ref{sec:emer}.

\smallskip\noindent
{\bf Example \#1} (revisiting theorem \ref{thm:contract}). If ${\mathbf r}$ is the Fiedler eigenvector associated with $\lambda_{N-1}(A)$ with ${\mathbf r}\perp {\boldsymbol 1}$, then (\ref{eq:dpcontract}) implies 
\[
 \lambda_{N-1}(A) = \frac{\dm{A\mathbf r}}{\dm{\mathbf r}} \leq \sup_{\bp\perp {\mathbf 1}} \frac{\dm{A\bp}}{\dm{\bp}} \leq 1-\ceta{_A}.
\]
We end up with the following lower bound for the Fiedler number  
\[
\lambda_2(L_A)=1-\lambda_{N-1}(A)\geq 1-(1-\ceta{_A}) \geq \ceta{_A}.
\]
Thus, theorem \ref{thm:contract}  is recovered here as a special case of the sharp bound (\ref{eq:symm_contract}) in theorem \ref{thm:symm}. The former has the advantage that it applies to non-symmetric models, but as  remarked earlier, is limited to models with global interactions; the latter can address the consensus of local, connected  models, consult section \ref{sec:nearest_neighbor} below. 

We remark in passing that  while theorem \ref{thm:contract} employs the $\ell^\infty$-based diameter,
$\dm{\bp}=\dm{\bp}_\infty=\max_{ij}|\bp_i-\bp_j|$, then theorem \ref{thm:symm} is in fact the corresponding $\ell^2$-based diameter, $\dm{\bp}^2_2:=\sum_{ij}|\bp_i-\bp_j|^2/(2N)=\veep{\bp}$.

\smallskip\noindent
{\bf Example \#2} (revisiting propositions \ref{prop:CV_opinions} and \ref{prop:CV_flocking}). A straightforward lower bound   $\lambda_2(L_A)\geq N\min a_{ij}$  recovers corollary \ref{cor:act},

\begin{equation}\label{eq:symop}
\veep{\bp(t)} \leq \exp\left(-{\alpha}\int^t m(s)ds\right) \veep{\bp(0)}, \qquad m(t):=\min_{ij}\phi(|\bx_i(t)-\bx_j(t)|),
\end{equation}

The characterization of concentration in theorem \ref{thm:symm} is sharp in the sense that the estimate (\ref{eq:Vbound}) is.
Indeed, it is well known that positivity of the Fiedler number, $\lambda_2(L_A)>0$, characterizes the \emph{algebraic connectivity} of the  graph associated with the adjacency matrix $A$, \cite{fiedler1973,mohar1991,chung_1997}. Theorem \ref{thm:symm} places a minimal requirement on the amount of \emph{connectivity as a necessary condition} for consensus\footnote{We ignore  possible cases in which  the self-organized dynamics may regain connectivity under ``cluster dynamics'', namely, agents separated into disconnected clusters  and  merging into each other at a later stage.}. 
There are  many characterizations for the algebraic connectivity of \emph{static} graphs \cite{chung_1997,Demmel_1996,fiedler1973,fiedler1989,GR_2001,Mer_1994,mohar1991,Scha_2007}. 
In the present context of self-organized dynamics (\ref{eq:framework}), however, the dynamics of $\dot{\bp}=\alpha(A\bp-\bp)$ dictates  the connectivity of $A=A(\P(t))$, which in turn, determines the clustering behavior of the dynamics, due to the  nonlinear dependence, $A=A(\P(t))$.  Thus, the intricate aspect of
the self-organized dynamics (\ref{eq:framework}) is   tracing its algebraic connectivity over time through the self-propelled  mechanism in which the nonlinear dynamics and  algebraic connectivity are tied together.  This issue will be explored in the next sections, dealing with clustering driven by \emph{local} interactions.
  
\section{Local interactions and  clustering}\label{sec:clustering}
\setcounter{equation}{0}
\setcounter{figure}{0}

In this section we consider the self-organized dynamics (\ref{eq:framework}) of a ``crowd'' of $N$ agents, 
$\P=\{\bp_i\}_{i=1}^N$ which does not interact globally: entries in their adjacency matrix may vanish, $a_{ij}\geq 0$.
The dynamics is dictated by local interactions and its large time behavior leads to the formation of one or more \emph{clusters}.

\subsection{The formation of clusters}\label{sec:dissipation_clustering}
A cluster $\C$ is a connected subset  of agents, $\{\bp_i\}_{i\in \C}$, which is separated from all other agents outside $\C$, namely

\medskip
$\#1. \quad  a_{ij}\neq 0 \quad \text{for all} \quad i,j\in \C$; \quad and \quad 
$\#2. \quad a_{ij}=0 \quad  \text{whenever} \quad  i\in\C \ \text{and} \ j\notin\C$.

\medskip\noindent
The important feature of such clusters is  their  self-contained dynamics in the sense that
\[
\frac{d}{dt}\bp_i = \alpha\sum_{j\in\C} a_{ij}(\bp_j-\bp_i),  \quad \sum_{j\in\C} a_{ij}=1, \qquad i\in \C.
\]
The dynamics of such self-contained clusters is covered by the concentration statements of global dynamics in section \ref{sec:global}. In particular, if cluster $\C(t)$ remains connected and isolated for sufficiently long time, then its agents will tend to concentrate around a  local consensus, 
\[
\bp_i(t) \stackrel{t \rightarrow \infty}{\longrightarrow}\con{\bp}_{\C}, \ \ \text{for all} \ \ i\in\C.
\] 
The intricate  aspect, however, is the last \emph{if} statement: the evolution of agents in a cluster $\C$ may become influenced by non-$\C$ agents, and in particular, different clusters may merge over time. 

\medskip
\noindent
In the following, we fix our attention on the particular models for opinion and flocking dynamics, expressed in the  unified framework  (\ref{eq:unified}),
\begin{subequations}\label{eqs:frame}
\begin{equation}\label{eq:frame}
\frac{d}{dt} \bp_i = \alpha\sum_{j=1}^N a_{ij} (\bp_j-\bp_i), \qquad a_{ij}=a_{ij}(\bx)=\frac{1}{\deg_i}\phi(|\bx_i-\bx_j|);
\end{equation}
Recall that  $\bp\mapsto \bx$ in opinion dynamics,  $\bp \mapsto \dot{\bx}$ in flocking dynamics, and $\deg_i$ is the degree,
\begin{equation}
\qquad
\left\{\begin{array}{ll} 
\deg_i=N, & \text{symmetric model},\\ \\
\deg_i=\sum_{j\neq i} \phi(|\bx_i-\bx_j|), & \text{non symmetric model}.\end{array}\right.
\end{equation}
\end{subequations}
We assume that the influence function $\phi$ is  compactly supported 
\begin{equation}
  \label{eq:phi_compactly_supported}
  \Supp = [0,\D].
\end{equation}

A cluster $\C=\C(t)\subset \{1,2,\ldots, N\}$ is dictated by the finite diameter of the influence function $\phi$ such that the following two properties hold:   

\medskip
$\#1. \quad  \displaystyle \max_{i,j\in\C(t)}|\bx_{i}(t)-\bx_{j}(t)| \leq  \D$; \quad and \quad $\#2. \quad  \displaystyle \min_{i\in \C(t), j\notin \C(t)}|\bx_{i}(t)-\bx_{j}(t)|> \D.$

When the  dynamics is global, $\D \gg \dm{\bx(0)}$, then the whole crowd of agents can be considered as one connected cluster.
Here we consider the \emph{local} dynamics when $\D$ is small enough relative to the active diameter of the global dynamics:  $\D<\dm{\bx(0)}$ in the opinion dynamics (\ref{eqs:opinion}), or $\D<\Diam$ in the flocking dynamics (\ref{eqs:flocking}). The statements of global concentration towards a consensus state  asserted in propositions \ref{prop:CV_opinions} and \ref{prop:CV_flocking} do not  apply. Instead, the local dynamics of agents leads them to concentrate in one or \emph{several} clusters --- consult for example, figures \ref{fig:evolution_X_farVscloseNeighbor} and \ref{fig:simu_2D}, \ref{fig:ratio011210} below. Our primary interest is in the large time behavior of such clusters. The generic scenario is a crowd of agents which is partitioned into a collection of clusters, ${\C}_k, \ k=1,\ldots K$, such that 
\[
\left\{\begin{array}{lll} \text{either} &  |\bp_i(t)-\bp_j(t)| \stackrel{t \rightarrow \infty}{\longrightarrow} 0, & \text{if} \ \  \ i,j\in \C_k \leftrightarrow  |\bx_i(t)-\bx_j(t)|\leq  \D  \ \\ \\
\text{or} & |\bx_i(t)-\bx_j(t)| > \D, & \text{if} \ \ \ i\in \C_k, j\in \C_\ell, \ \ k\neq \ell.
\end{array}\right.
\]
In this context, we raise the following two fundamental questions.

\medskip\noindent
{\bf Question \#1}.  Identify the class of initial configurations, $\P(0)$, which evolve into finitely many clusters, ${\C}_k, \ k=1,\ldots K$. In particular, characterize  the number of such clusters $K$ for $t\gg1$. 

\smallskip\noindent
{\bf Question \#2}. Assume that the initial configuration $\P(0))$ is connected. Characterize  the initial configuration $\P(0)$ which evolve into one cluster, $K(t)=1$ for $t\gg 1$, namely, the question of emerging of consensus in the local dynamics.

A complete answer to these questions should provide an extremely interesting insight into  local processes of self-organized dynamics, with many applications. In the next two sections we provide partial answers to these questions. 
We begin with the first result which shows that if the solution of (\ref{eqs:frame}) has bounded time-variation then it must be partitioned into a collection of clusters.

\begin{proposition}[Formation of clusters]\label{ppo:limit_set_ct}
  Let $\P(t)=\{\bp_k(t)\}_k$ be the solution of the opinion or flocking   models \eqref{eqs:frame} with compactly supported influence function $\Supp=[0,\D)$, and assume it has a bounded time-variation
\begin{equation}\label{eq:fast_en}
\int^\infty |\dot{\bp}_i(s)|ds < \infty. 
\end{equation}
Then $\P(t)$ approaches a stationary state, $\con{\bp}$,
which  is partitioned into $K$ clusters, $\{{\mathcal C}_k\}_{k=1}^K$,  such that  $\{1,2,\ldots,N\}=\cup_{k=1}^K \C_k$ and
  \begin{equation}\label{eq:pq_cluster} 
   \left\{
 \begin{array}{lcll}
   \text{either} &  \!\!\!\!\!\!\!\bp_i(t) \longrightarrow  \con{\bp}_{\C_k} & \!\!\!\!\!\!\text{as} \ t\rightarrow \infty, & \qquad \text{for all} \ \ i\in \C_k, \\ \\
\text{or} & \!\!\!\!|{\bx}_i(t) - {\bx}_j(t)| > \D & \text{for} \ t\gg 1, & \qquad \text{if} \ \ i\in \C_k, j\in \C_\ell, \ \ k\neq \ell.
\end{array}\right.
  \end{equation}
\end{proposition}

\begin{myremark}
Observe that if the solution decays fast enough --- in particular, if $\bp(t)$ decays exponentially fast,
$|\bp_i(t)-\bp_i^\infty| \lesssim e^{-C (t-t_0)}, \quad t\geq t_0>0$ $($as in the unconditional consensus and flocking of global interactions discussed in section \ref{sec:global}$)$, 
then it has a bounded time-variation.
\end{myremark}

\begin{proof} 
  Assumption  (\ref{eq:fast_en}) implies
\[
|\bp_i(t_2)-\bp_i(t_1)|\leq \int_{t_1}^{t_2} |\dot{\bp}_i(s)|ds \ll 1 \ \ \text{for} \ \ t_2>t_1\gg1,
\]
hence each agent approaches  its own stationary state, $\bp_i(t) \stackrel{t\rightarrow \infty}{\longrightarrow} \con{\bp}_i$. 
We claim that 
$\dot{\bp}_i(t)\stackrel{t\rightarrow \infty}{\longrightarrow} 0$.
 To this end, we distinguish  between the two cases of first-order opinion dynamics and second-order flocking dynamics. In opinion dynamics, $\bp\mapsto \bx$: since the expression of the right of (\ref{eqs:frame}),
\begin{equation}\label{eq:bpdot}
\dot{\bp}_i(t) = \frac{\alpha}{\deg_i(t)}\sum_j \phi(|\bx_i(t)-\bx_j(t)|)(\bp_i(t)-\bp_j(t)), \quad \deg_i(t)= \sum_j \phi(|\bx_i(t)-\bx_j(t)|), 
\end{equation}
has a limit (involving $\con{\bp}_i=\con{\bx}_i$), it follows that $\lim_{t\rightarrow\infty} \dot{\bp}_i(t)$ exists and by (\ref{eq:fast_en}) it must be zero, $\dot{\bp}_i(t) \rightarrow 0$.
In the case of flocking dynamics, $\bp \mapsto \dot{\bx}$, and there are two types of pairs of agents $(i,j)$: either they have the same limiting ``velocity'', $\con{\bp}_i-\con{\bp}_j=0$, and since $\phi$ is bounded,  
\[
\phi(|\bx_i(t)-\bx_j(t)|)(\bp_i(t)-\bp_j(t))\stackrel{t\rightarrow \infty}{\longrightarrow} 0;
\]
or --- if $\con{\bp}_i-\con{\bp}_j \neq 0$ then, 
\begin{equation}\label{eq:ndtype}
|\con{\bx}_i-\con{\bx}_j| \gtrsim  |\con{\bp}_i-\con{\bp}_j|t > \D, \qquad t\gg 1,
\end{equation}
and hence
\[
\phi(|\bx_i(t)-\bx_j(t)|)(\bp_i(t)-\bp_j(t)) =0, \qquad t\gg 1.
\]
In either case, the expression on the right of (\ref{eq:bpdot}) vanishes as $t\rightarrow \infty$.
\newline
Now, take the scalar product of (\ref{eq:bpdot}) against ${\bp}_i$ and sum,
\begin{equation}\label{eq:abound}
    \sum_i \deg_i \langle\dot{\bp}_i,{\bp}_i\rangle = \alpha \sum_{ij} \phi_{ij} \langle{\bp}_j-{\bp}_i,{\bp}_i\rangle \equiv
    -\frac{\alpha}{2} \sum_{ij} \phi_{ij} |{\bp}_j-{\bp}_i|^2.
\end{equation}
  Since ${\bp}_i\in \Omega(0)$, $\deg_i \leq N$ are uniformly bounded and $\dot{\bp}_i(t) \rightarrow 0$ on the left, it follows that the expression on the right tends to zero.  In opinion dynamics ($\bp\mapsto \bx$) we can pass to the limit on the expression on the right which yields
  \begin{equation}
    \label{eq:convergence_sum_ct}
     \phi(|\bx_i^\infty-\bx_j^\infty|) |{\bp}_i^\infty-{\bp}_j^\infty|^2= 0, \quad   \   \text{for all} \ \ i,j \leq N.
  \end{equation}
Thus, if $|\con{\bx}_i-\con{\bx}_j|> \D$, then agents $i$ and $j$ are in separate clusters. Otherwise, when they are in the same cluster, say $i,j\in \C_k$ so that $|\con{\bx}_i-\con{\bx}_j| <\D$, then $\phi(|\bx_i^\infty-\bx_j^\infty|)>0$ and by (\ref{eq:convergence_sum_ct}) they must share the same stationary state, $\con{\bp}_i=\con{\bp}_j=:\con{\bp}_{\C_k}$, that is,  (\ref{eq:pq_cluster}) holds. In the case of flocking dynamic, $\bp \mapsto \dot{\bx}$, we either have one type of pairs, $|\bp_i(t)-\bp_j(t)| \stackrel{t\rightarrow \infty}{\longrightarrow}0$ or a second type of pairs, (\ref{eq:ndtype}),
namely,  (\ref{eq:pq_cluster}) holds.
\end{proof}

\medskip\noindent
We now turn our attention to the number of clusters, $K$.

\subsection{How many clusters?}\label{sec:countK}
Note that if  $\bp^\infty=(\bp_1^\infty,\ldots,\bp_N^\infty)^\top$ be a stationary state of (\ref{eqs:frame}) then $\bp^\infty$ is an eigenvector associated with the nonlinear eigenvalue problem, 
\[
A(\bx^\infty)\bp^\infty=\bp^\infty, 
\]
corresponding to the eigenvalue $\lambda_N(A(\bx^\infty))=1$.   
 Actually, the number of stationary clusters can be directly computed from the multiplicity of leading spectral eigenvalues of $\lambda_N(A(\con{\bx}))$.
\begin{proposition}  \label{ppo:cluster_eigenvalueA}
  Assume that the crowd of $N$ agents $\{\bp_i(t)\}_{i=1}^N$ is partitioned into $K$ clusters, $\{1,2,\ldots,N\}=\cup_{k=1}^{K(t)}\C_k$. Then, the number of clusters, $K=K(t)$, equals the geometric multiplicity of $\lambda_N(A(\bx(t))=1$, 
  \begin{equation}
    \label{eq:S_eigenvalue1}
    K(t) = \left\{\#\lambda_N(A(\bx(t)) \ | \ \lambda_N(A(\bx(t))=1  \right\}.
  \end{equation}
\end{proposition}

\begin{proof} We include the rather standard argument for completeness. 
  Suppose that the dynamics of (\ref{eqs:frame}) at time $t$ consists of  $K=K(t)$ clusters, $\displaystyle \cup_{k=1}^{K(t)}{\mathcal C}_k$. Define the vector ${\br}^k=(r^k_1,\ldots,r^k_N)^\top$ such that:
  \begin{displaymath}
    r^k_j = \left\{
      \begin{array}{ll}
        1 &  \text{if } j\in \mathcal{C}_k \\
        0 &  \text{otherwise}
      \end{array}
      \right. .
  \end{displaymath}
  We obtain
  \begin{displaymath}
    \left(A {\br}^k\right)_i = \sum_j a_{ij} r^k_j = \sum_{j\in \mathcal{C}_k} a_{ij}.
  \end{displaymath}
  Using the fact that $A$ is a stochastic matrix and that $a_{ij}=0$ if  ${\bx}_i$ and ${\bx}_j$ are not in the same cluster, we deduce
  \begin{displaymath}
    \sum_{j\in \mathcal{C}_k} a_{ij} = \left\{
      \begin{array}{ll}
        1 &  \text{if } i\in \mathcal{C}_k \\
        0 &  \text{otherwise}
      \end{array}
      \right\}=\br^k_i,
  \end{displaymath}
  and therefore $A{\br}^k = {\br}^k$. Thus, associated with each cluster $\mathcal{C}_k$, there is an eigenvector ${\br}^k$ corresponding to $\lambda_N(A)=1$. To conclude the proof, we have to show  that there are no other vectors ${\br}$ satisfying $A{\br}={\br}$.
Indeed, assume that  $A{\br}={\br}$,  
  \begin{displaymath}
    \sum_j a_{ij} r_j = r_i \qquad \text{ for any } i.,
  \end{displaymath}
  Fix a cluster $\mathcal{C}_k$. Then for any  $p\in {\mathcal C}_k$  we have
  \begin{displaymath}
    \sum_{j\in\mathcal{C}_k} a_{pj} r_j = r_p \qquad \text{ for any } p\in\mathcal{C}_k.
  \end{displaymath}
  Denote by $r_q$  the maximal  entry  of $r_j$'s on the left, corresponding to  some $q\in {\mathcal C}_k$: since $\sum_{j\in\mathcal{C}_k} a_{pj}=1$ with $a_{pj}>0$, we deduce that for any $p\in {\mathcal C}_k$ we have $r_p =\sum a_{pj}r_j \leq \sum a_{pj}r_q= r_q$. Thus,  the entries of ${\br}$ are constant on the cluster $\mathcal{C}_k$, so that $\br \propto{\br}^k$.
\end{proof}

\subsection{Numerical simulations with local dynamics}

We illustrate  the emergence of clusters with one- and two-dimensional simulations of the opinions dynamics model (\ref{eq:opinion_formationb}),
\begin{equation}  \label{eq:matrix_A}
  \frac{d}{dt}\bx_i = \sum_j \frac{\phi_{ij}}{\sum_k \phi_{ik}}(\bx_j-\bx_i), \qquad \bx_i(t)\in {\mathbb R}^d.
\end{equation}
The influence  function, $\phi$, was taken as the characteristic function of the interval $[0,1]$: $\phi(r) = \chi_{[0,1]}$, and we use the Runge-Kutta method of order $4$ with a time step of $\Delta t=.05$, for the  time discretization of the system of ODEs (\ref{eq:matrix_A}).

\begin{figure}[ht]
  \centering
  \includegraphics[scale=.30]{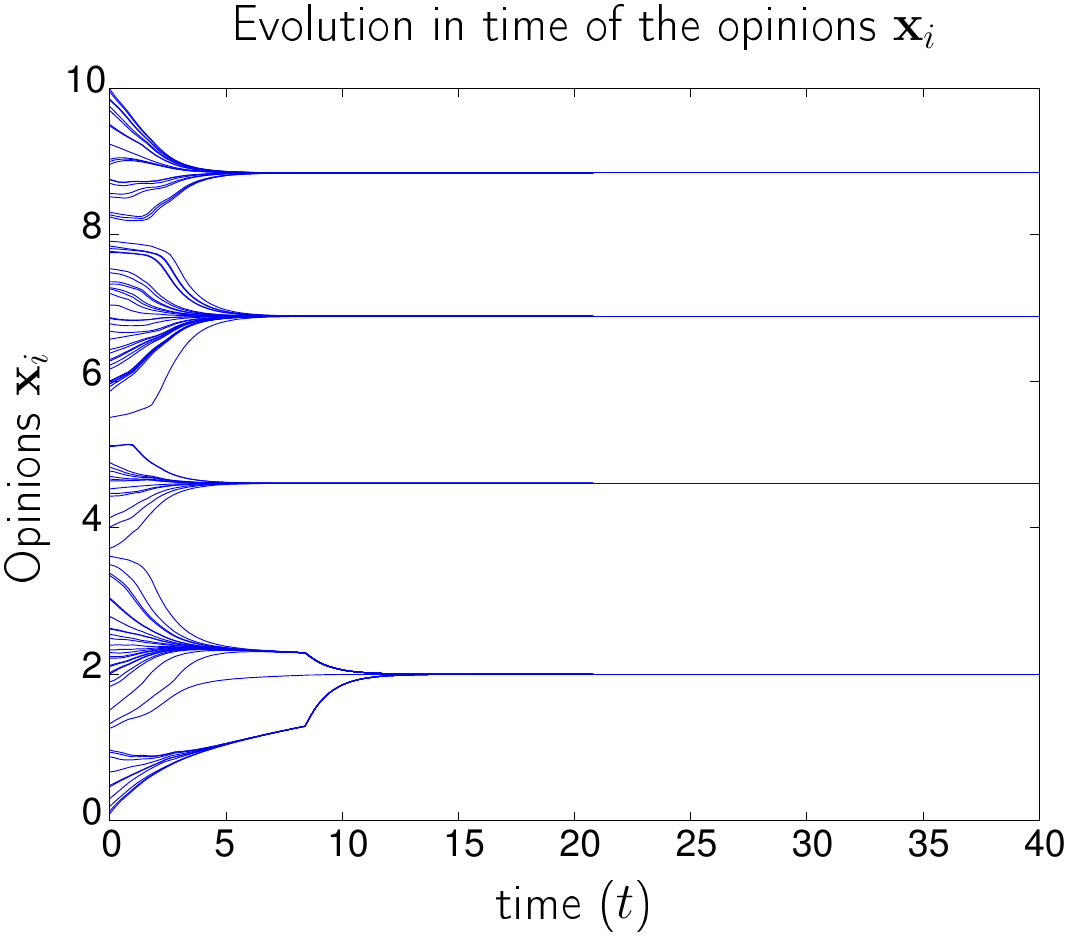} \quad
  \includegraphics[scale=.30]{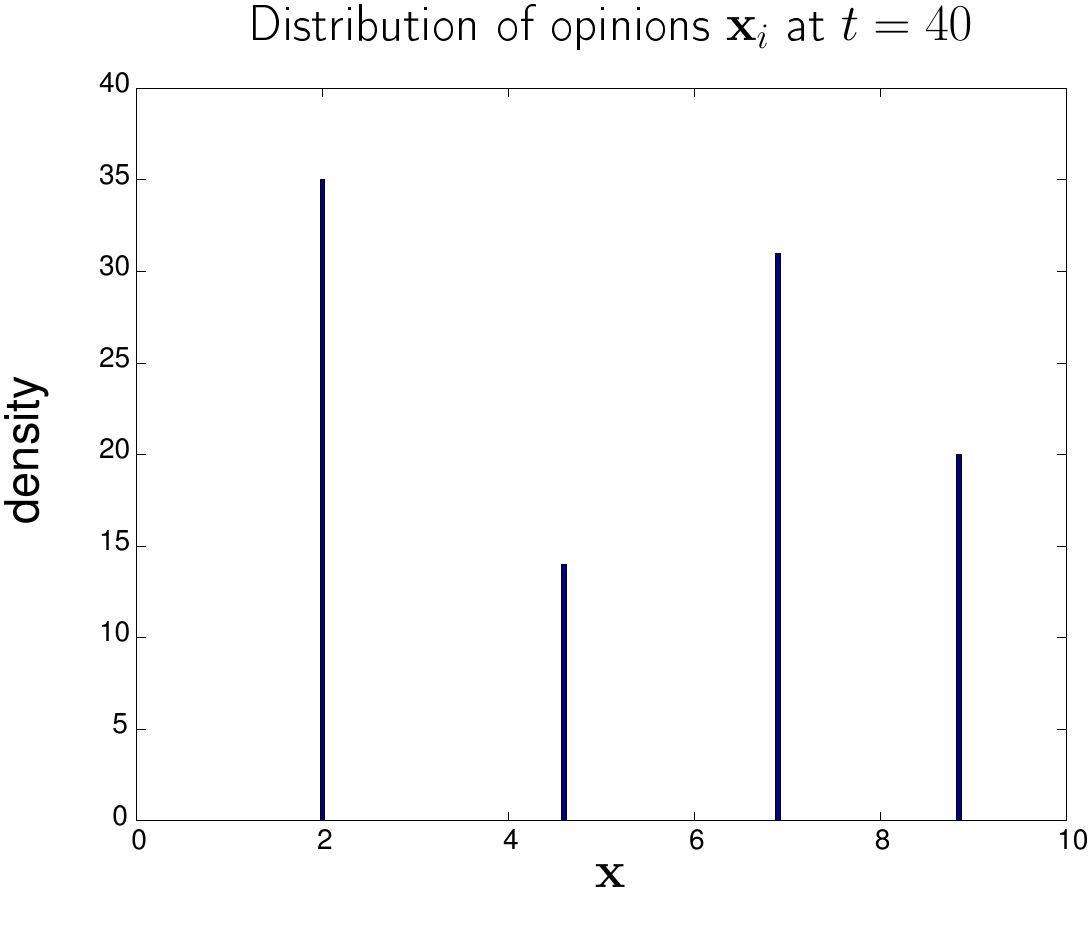}
  \caption{{\small The opinion model (\ref{eq:opinion_formationb}) with $M=100$ agents and $\phi=\chi_{[0,1]}$ ({\bf Left} figure) and the histogram of the distribution of ${\bf x}_i$ at $t=40$ unit time ({\bf Right} figure). We observe the formation of $4$ clusters separated by a distance greater than $1$.}}
  \label{fig:evolutionXratio1_bis}
\end{figure}

As a first example, we run a simulation of the one-dimensional opinion model, $d=1$, subject to initial configuration of $N=100$ agents uniformly distributed on the interval $[0,10]$. In the figure \ref{fig:evolutionXratio1_bis} (Left), we plot the evolution of the opinions ${\bx}_i(t)$ in time. We observe the formation of 4 clusters after $15$ unit time. The histogram of the distribution of agents at the final time $t=40$ (figure \ref{fig:evolutionXratio1_bis} right) shows that the distance between the clusters is greater than $1$ as predicted by proposition \ref{ppo:limit_set_ct}. We also observe that the number of opinions contained in each cluster differs (respectively 35, 14, 31 and 20 agents). Indeed, the larger cluster at $x\approx2$ with $35$ opinions is a merge between $3$ {\it branches} (figure \ref{fig:evolutionXratio1_bis}) with one branch in the middle connecting the two external branches. When the two external branches finally connect at $t\approx8.5$ (their distance is less than $1$), we observe an abrupt change in the dynamics following by a merge of the $3$ branches into a single cluster.

To analyze the cluster formation, we also look at the evolution of the eigenvalues of the matrix of interaction $A(\bx(t))$ in (\ref{eq:matrix_A}), $a_{ij}=\phi_{ij}/\sum_k \phi_{ik}$. In the figure \ref{fig:eigenvalues_ratio1_bis}, we represent the evolution of the $8$ first eigenvalues of the matrix $A$.  From $t=0$ to $t\approx3.5$, we observe that the 4 first eigenvalues converge to $1$ which counts for the fact that only $4$ clusters remains at this time. Then the matrix $A(\bx(t))$ remains constant in time from $t\approx3.5$ to $t\approx8.5$. At $t\approx8.5$, two {\it branches} (see figure \ref{fig:evolutionXratio1_bis}) re-connect, and the two eigenvalues $\lambda_5$ and $\lambda_6$ equal zero.
This confirms proposition  \ref{ppo:cluster_eigenvalueA} where the additional multiplicity of the spectral eigenvalue $\lambda_N(A(\bx(t))=1$, indicates the formation of a new cluster.

\begin{figure}[ht]
  \centering
  \includegraphics[scale=.40]{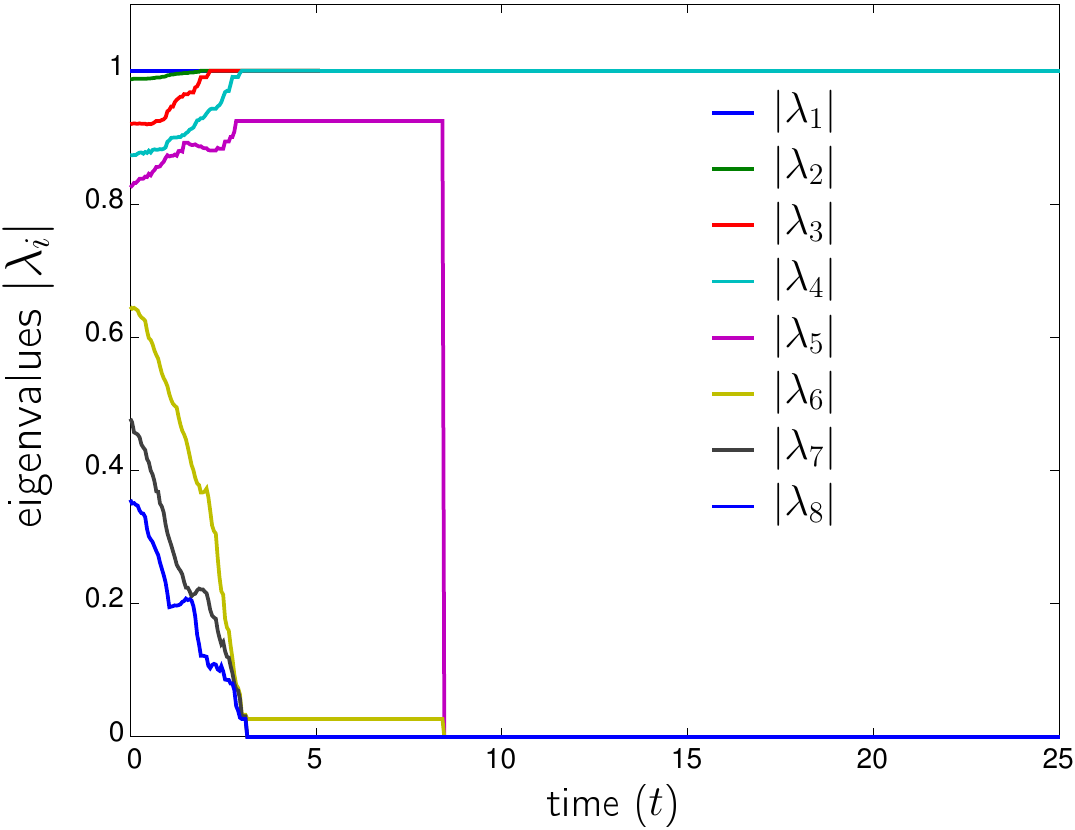}
  \caption{{\small Absolute values of the eigenvalues of the matrix $A$ (\ref{eq:matrix_A}) during the simulation given in figure \ref{fig:evolutionXratio1_bis}. The number of eigenvalues equal to $1$ corresponds to the number of clusters.}}
  \label{fig:eigenvalues_ratio1_bis}
\end{figure}

Next  we turn  to illustrate the dynamics of the two-dimensional, $d=2$, opinion model (\ref{eq:opinion_formationb}). With this aim, we run the model  starting with an initial condition of $N=1000$ agents distributed uniformly on the square  $[0,10]\times[0,10]$.
We present, in figure \ref{fig:simu_2D}, several snapshots of the simulations at different time ($t=0,\,2,\,4,\,6,\,12 \text{ and } 30$ unit time). As in the 1D case, we first observe a fast transition to a cluster formation (from $t=0$ to $t=6$). However at time $t=12$, the dynamics does not have yet converged to a stationary state, we observe at the upper-left that three branches are at distance less than $1$. This scenario is similar to the one observed in figure \ref{fig:evolutionXratio1_bis} with the apparition of $3$ {\it branches}. At $t=30$, the three clusters at the upper left have finally merged and the system has reached a stationary state: each cluster it at distance greater than $1$ from each other.

\begin{figure}[p]
  \centering
  \includegraphics[scale=.45]{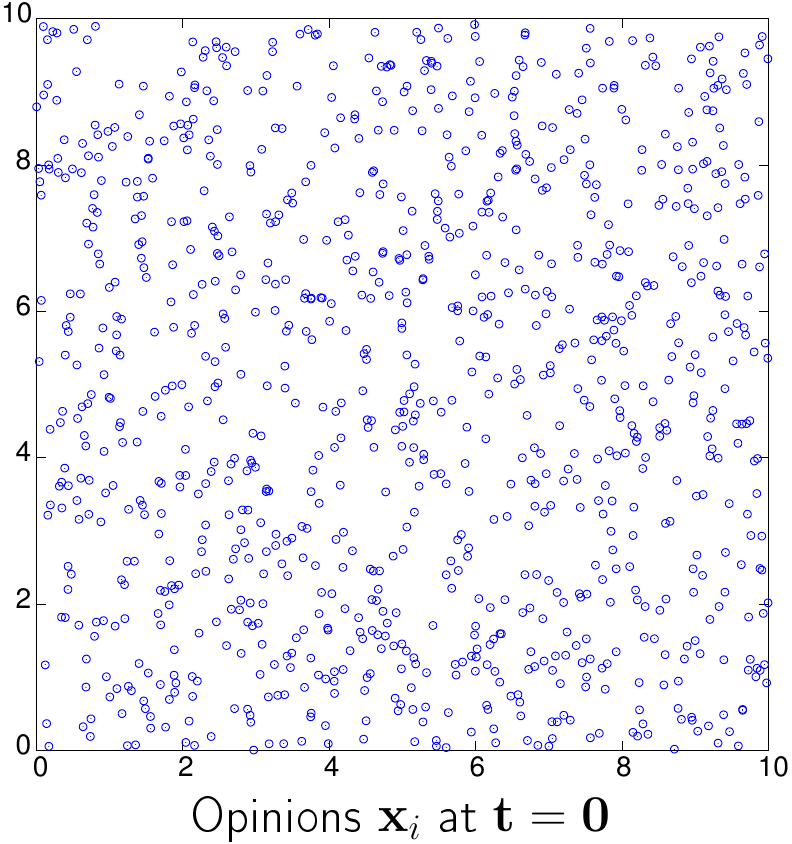} \quad
  \includegraphics[scale=.45]{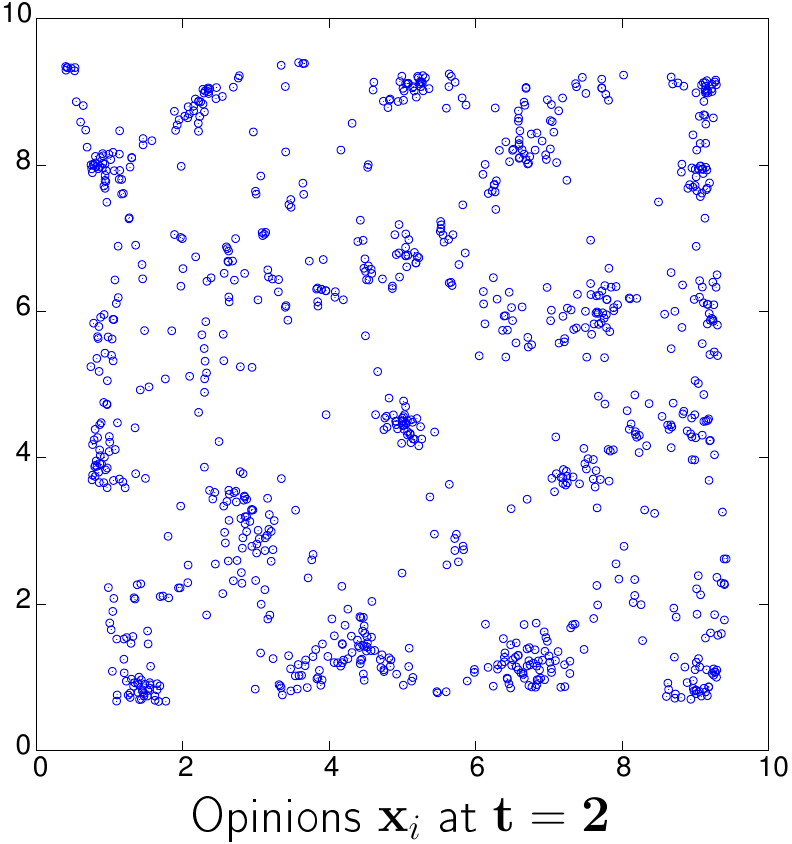}

  \bigskip

  \includegraphics[scale=.45]{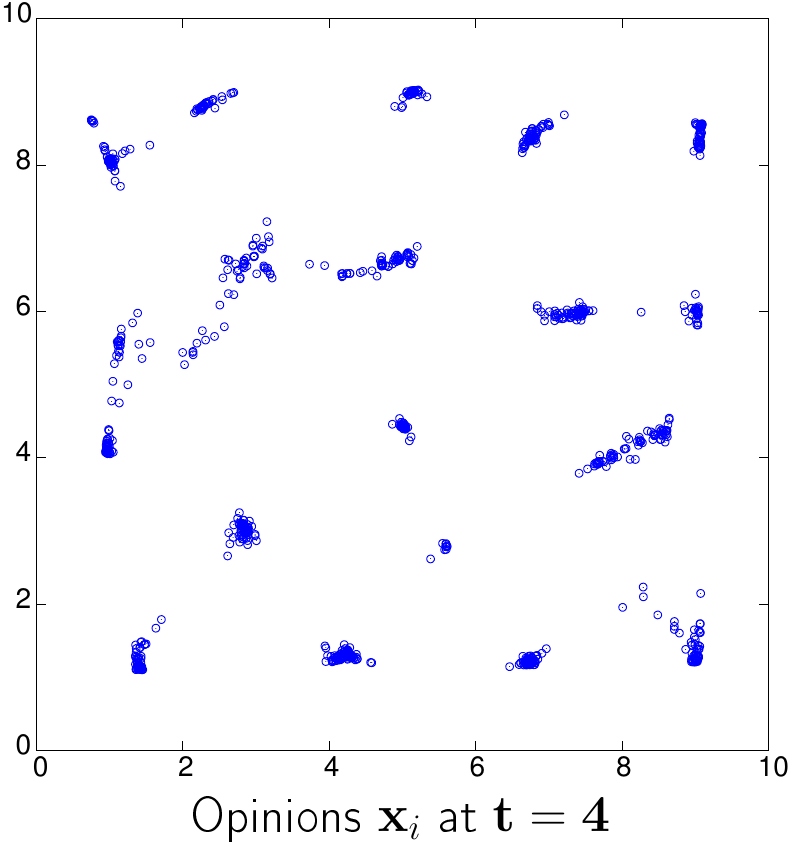} \quad
  \includegraphics[scale=.45]{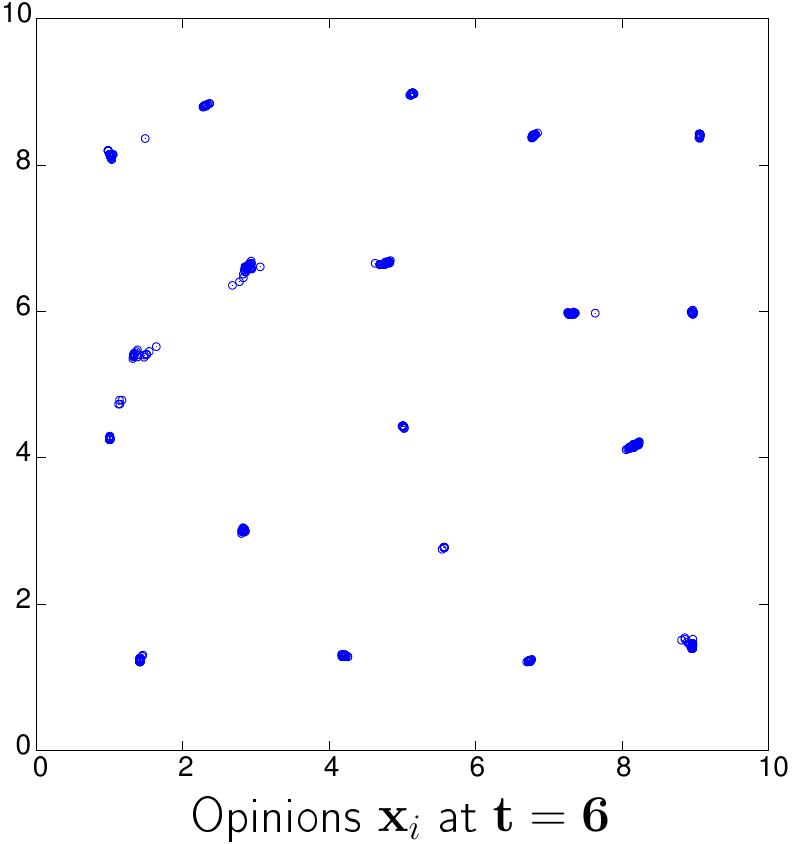}

  \bigskip
  
  \includegraphics[scale=.45]{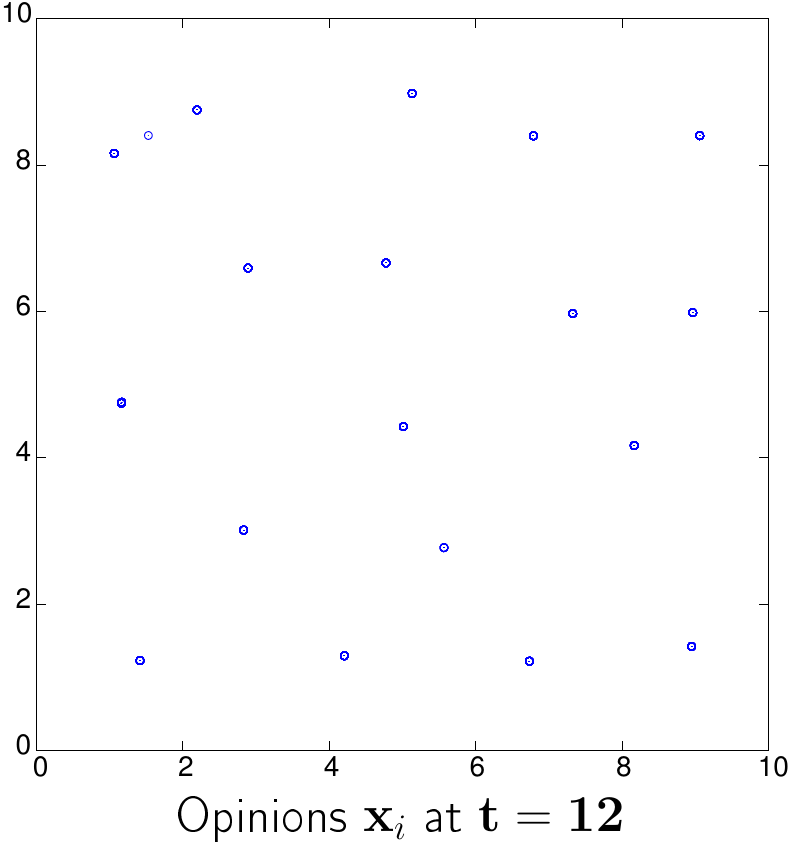} \quad
  \includegraphics[scale=.45]{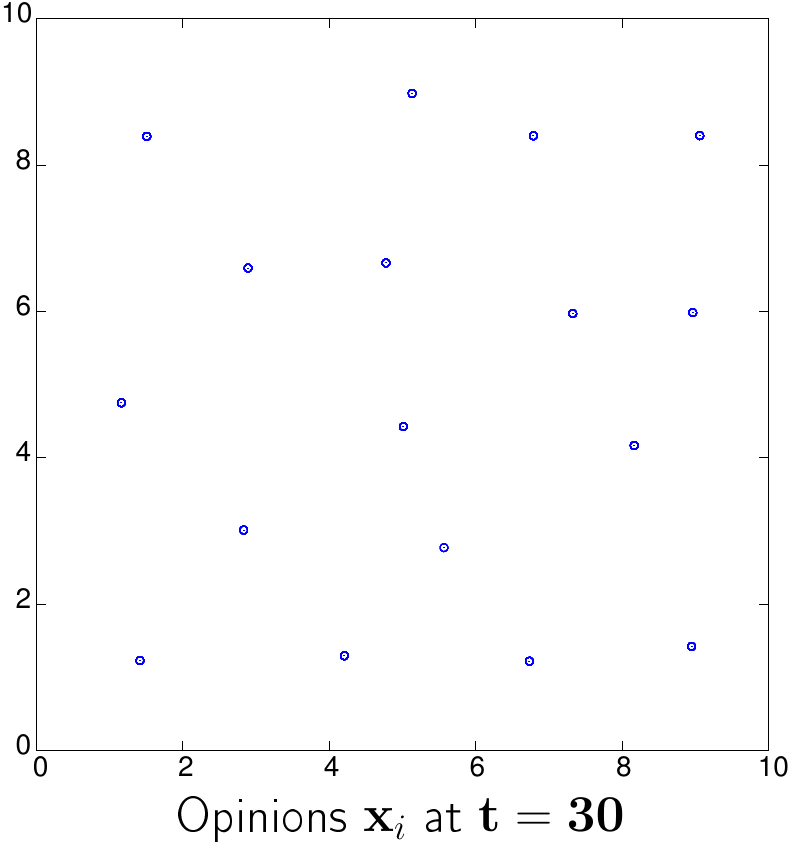}
  \caption{{\small Simulation of the opinion model (\ref{eq:opinion_formationb}) in 2D with $M=1000$ agents and $\phi=\chi_{[0,1]}$. The dynamics converges to a cluster formation (17 clusters) with each cluster separated by a distance greater than $1$.}}
  \label{fig:simu_2D}
\end{figure}

\section{$K=1$: uniform connectivity implies consensus}\label{sec:concon}
\setcounter{equation}{0}
\setcounter{figure}{0}
The emergence of a consensus in the opinion or flocking models  (\ref{eq:unified}) implies that the underlying graph associated with the dynamics must remain connected, namely $|\bx_i(t)-\bx_i(t)| \ll  \D$ at least for $t\gg 1$. In this section we discuss the converse statement, namely, that uniform connectivity implies consensus.
 The implication of consensus in the symmetric case is based on a straightforward  application of \emph{algebraic connectivity} and is outlined in section \ref{sec:loc_symm_models}. The corresponding question of consensus in  
non-symmetric connected models is carried out in section \ref{sec:local_non_symm_od} using an \emph{energy method}. 
We emphasize that consensus in both cases depend on the time-dependent behavior of  intensity of connectivity, beyond the mere graph connectivity. 
Recall that the graph associated with (\ref{eq:framework}), ${\mathcal G}_A:=(\P,A(\P))$, is connected if every two agents $\bp_i(t)$ and $\bp_j(t)$ are connected through a path $\Gamma_{ij}:=\{k_1=i<k_2<\ldots <k_r=j\}$ of length $r_{ij}\leq N$.
We measure the \emph{uniform connectivity} by its ``weakest link''.
\begin{mydefinition}[Uniform connectivity]
The self-organized dynamics \eqref{eq:framework} is connected if there exists $\mu(t)>0$ such that for all paths $\Gamma_{ij}$, 
\begin{equation}\label{eq:ucon}
\min_{k_\ell \in \Gamma_{ij}} a_{k_\ell,k_{\ell+1}}(\P(t)) \geq \mu(t)>0, \qquad \text{for all} \ i,j.
\end{equation} 
In particular, if $\mu(t)\geq \mu>0$ then we stay that $\P(t)$ is uniformly connected.
\end{mydefinition}
Alternatively, uniform connectivity of (\ref{eq:framework}) requires the existence of $\mu=\mu_A>0$ independent of time, such that
\[
\left(A^N(\P(t))\right)_{ij} \geq \mu^N>0.
\]

\subsection{Consensus in local dynamics -- symmetric models}\label{sec:loc_symm_models}
We consider the \emph{symmetric} dynamics  (\ref{eq:framework}) with associated graph ${\mathcal G}_A:=(\P,A(\P))$.
Fix the positions of any two agents $\bp_i(t)$ and $\bp_j(t)$ and their (shortest) connecting path $\Gamma_{ij}$ of length $r_{ij}$. Thus, $r_{ij}$ measures the degree of separation between agents $(i,j)$, and if we let the maximal degree of separation denote the diameter of the graph, $diam({\mathcal G}_A):=\max_{ij}r_{ij}$, then 
\[
|\bp_i-\bp_j|^2 \leq diam({\mathcal G}_A)\sum_{k_\ell \in \Gamma_{ij}} |\bp_{k_{\ell+1}}-\bp_{k_\ell}|^2, \qquad  diam({\mathcal G}_A) \leq N.
\]
By uniform connectivity $\mu \leq a_{k_{\ell+1},k_\ell}$ along each path and hence
\begin{equation}\label{eq:uni}
\frac{\mu}{diam({\mathcal G}_A)}|\bp_i-\bp_j|^2 \leq  \sum_{k_\ell \in \Gamma_{ij}} a_{k_{\ell+1},k_\ell}|\bp_{k_{\ell+1}}-\bp_{k_\ell}|^2 
\leq \sum_{ij} a_{ij}|\bp_i-\bp_j|^2, 
\end{equation}
and summation over all pairs yields
\[
\frac{\mu}{diam({\mathcal G}_A)}\sum_{ij}|\bp_i-\bp_j|^2 \leq  N^2\sum_{ij} a_{ij}|\bp_i-\bp_j|^2.
\]
Now we recall our notation $\bq_i:=\bp_i-\ave{\bp}$: invoking (\ref{eq:fidratio}) we find,
\begin{equation}\label{eq:fidlow}
\lambda_2(L_A)  = \min_{\sum \bq_k=0} \frac{\langle L_A\bq,\bq\rangle}{\langle \bq,\bq\rangle}
=\min_{\bp} \frac{(1/2)\sum_{ij} a_{ij}|\bp_i-\bp_j|^2}{(1/2N)\sum_{ij}|\bp_i-\bp_j|^2}  \geq  \frac{\mu}{N diam({\mathcal G}_A)}.
\end{equation}
Thus,  the  scaled connectivity factor $\mu/(N diam({\mathcal G}_A)) \geq \mu/N^2$ serves as a lower bound for the Fiedler number associated with the symmetric dynamics of (\ref{eq:framework}) (counting the number of ``maximal" edges, yields the  slightly sharper lower bound $\lambda_2 \geq 4\mu/N^2$, \cite{mohar1991}).\newline
Using theorem \ref{thm:symm} we conclude the following.

\begin{theorem}[Connectivity implies consensus: the symmetric case]\label{thm:con_symmetric}
Let $\P(t)=\{\bp_k(t)\}_k$ be the solution of a symmetric  self-organized dynamics
\[
\frac{d}{dt}\bp_i(t)=\alpha\sum_{j\neq i} a_{ij}(\P(t))(\bp_j(t)-\bp_i(t)), \qquad a_{ij}=a_{ji}.
\] 
If $\P(t)$ remains  connected in time with  ``sufficiently strong'' connectivity $\mu_{A(\P(s))}>0$, then it approaches the  consensus $\ave{\bp}(0)$, namely,
\[
\veep{\bp(t)} \lesssim \exp\left(-\frac{\alpha}{N^2} \int_0^t\mu_{A(\P(s))}ds\right)\veep{\bp(0)}, \ \ \quad   \veep{\bp(t)}^2:=\frac{1}{N}\sum |\bp_i(t)-\ave{\bp}(0)|^2.
\]
In particular, if $\P(t)$ remains uniformly connected in time, \eqref{eq:ucon}, then it approaches an emerging consensus, $\bp_i(t) \rightarrow \con{\bp}=\ave{\bp}(0)$ with a convergence rate,
\begin{equation}\label{eq:erate}
 \veep{\bp(t)} \lesssim e^{\displaystyle -\alpha \frac{\mu}{N^2} t} \veep{\bp(0)}. 
\end{equation} 
\end{theorem}

It is important to notice that theorem \ref{thm:con_symmetric} requires the intensity of connectivity to be  sufficiently strong: connectivity alone,
with a rapidly decaying $\mu(t)$,  is not  sufficient for consensus as illustrated by the following.

\noindent
{\bf Counterexample}.  Consider the symmetric dynamics (\ref{eq:opinion_formationa}) with $5$ agents, $x_1,\dots,x_5$, subject to  initial configuration
\begin{equation}
  \label{eq:cex_symmetry}
  \bx_1(0)=-\bx_5(0), \quad \bx_2(0)=-\bx_4(0), \quad \bx_3(0)=0, 
\end{equation}
with   $(\bx_4(0),\bx_5(0))$ to be specified below inside the box ${\mathcal D}:=\{\frac{1}{2}< \bx_4 < 1 <\bx_5<\frac{3}{2}\}$. 
  We fix the influence function $\phi(r)=(1+r)^2(1-r)^2 \chi_{[0,1]}$, compactly supported on $[0,1]$; note that $\phi'(0)=\phi'(1)=0$. By symmetry, the initial ordering in (\ref{eq:cex_symmetry}) is preserved in time.  In particular,  $\bx_3(t)\equiv 0$, and  $(\bx_4(t),\bx_5(t))\mapsto (x(t),y(t))$ preserve the original  ordering, $\frac{1}{2}< x(t) < 1 <y(t)$, the symmetric opinion dynamics (\ref{eq:opinion_formationa}) (with $\alpha=5$ for simplicity), is reduced to 
\begin{equation}
  \label{eq:system_cexample}
  \begin{array}{lcl}
    \dot{x} &=& -\phi(|x|) x + \phi(|y-x|) (y-x) \\
    \dot{y} &=& \phi(|x-y|) (x-y)
  \end{array}
\end{equation}
 An equilibrium for the system is given by $x=y=1$. The eigenvalues of the linearized system at $(1,1)$ are $\lambda_1=0$ and $\lambda_2=-2$, therefore the equilibrium is unstable. 
We would like to prove that there exists an initial condition $(x(0),y(0))$ close to $(1,1)$ which converges toward this unstable equilibrium. We use for that a variant of the antifunnel theorem \cite{Hubbard_West_1997}.

\begin{figure}[ht]
  \centering
  \includegraphics[scale=.3]{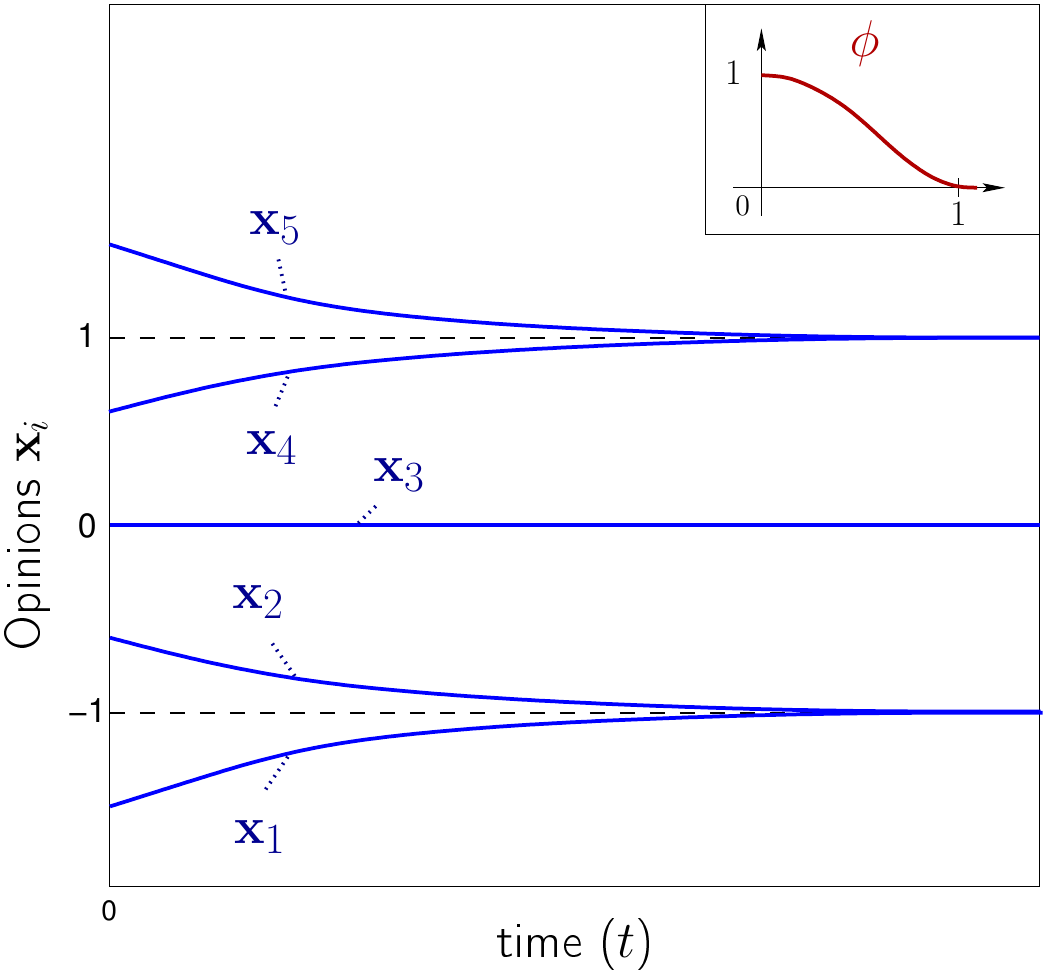}\qquad
  \includegraphics[scale=1]{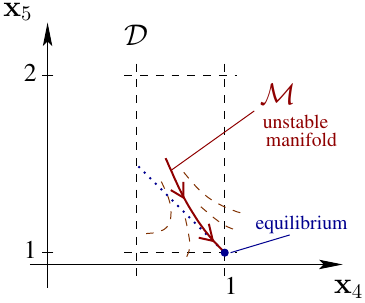}
  \caption{{\bf Left}: a solution of the symmetric model that stays connected but does not converge to a consensus. {\bf Right}: in phase space, the counter example is a solution that stays in the {\it antifunnel} formed by the curve $\alpha$ and $\beta$.}
  \label{fig:phase_portrait_cex2}
\end{figure}

We study the phase portrait of the dynamical system \eqref{eq:system_cexample} close to the unstable equilibrium $(1,1)$. Take $\varepsilon$ such that $0<\varepsilon<\frac12$ and consider the $3$ curves (see figure \ref{fig:phase_portrait_cex2}):
\begin{eqnarray*}
  \alpha(s) = (2-s,s),\quad \beta(s)=(1,s) && \text{ for } s\in(1,1+\varepsilon] \\
  \gamma(s) = (2-s,1+\varepsilon) \hspace{2.1cm} && \text{ for } s\in[1,1+\varepsilon]
\end{eqnarray*}
We denote by $\mathcal{D}_\varepsilon$ the domain enclosed by the $3$ curves:
\begin{displaymath}
  \mathcal{D}_\varepsilon=\{2-y\leq x\leq1\;,\;1<y\leq1+\varepsilon\}.
\end{displaymath}
Notice that on the domain $\mathcal{D}_\varepsilon$, we have $\dot{y}<0$. Thus, given a solution of \eqref{eq:system_cexample} starting on $\gamma$, there are $3$ possibilities: the solution exits the domain passing through the curves $\alpha$, or it exits passing through $\beta$ or it converges to the equilibrium $(1,1)$.

To prove the existence of solutions in the third category, we notice that the curves $\alpha$ and $\beta$ form an {\it antifunnel} for the dynamical system. Starting on the curve $\beta$, since $\dot{x}>0$, the solution exits the domain $\mathcal{D}_\varepsilon$ (see figure \ref{fig:phase_portrait_cex2}). Similarly, on the curve $\alpha$, since $\dot{x}<\dot{y}$, the solution exits the domain $\mathcal{D}_\varepsilon$ as well. 

We denote by $\gamma_\alpha$ the set of initial conditions contained in $\gamma$ such that the solution exits through $\alpha$. The set $\gamma_a$ is non-empty since $(1-\varepsilon,1+\varepsilon)\in\gamma_a$. Moreover, using the same arguments as in \cite{Hubbard_West_1997}, we find out that $\gamma_\alpha$ is open. Similarly, we denote  by $\gamma_\beta\subset\gamma$ the set of initial conditions such that the solution exits through $\beta$ and we deduce that $\gamma_\beta$ is open and non-empty. Since $\gamma_\alpha\cap\gamma_\beta=\emptyset$, by connectivity of the set $\gamma$, there exists $(x_*,y_*)$ which does not belong to $\gamma_\alpha\cup\gamma_\beta$. Thus, the solution $(x(t),y(t))$ starting from $(x_*,y_*)$ stays in between $\alpha$ and $\beta$:
\begin{displaymath}
  2-y(t)\leq x(t)\leq1 \quad , \quad 1\leq y(t) \leq 1+\varepsilon, \qquad \text{ for all } t\geq0.
\end{displaymath}
Since $y(t)$ is decreasing and lower bounded, $y(t)$ converges: $y(t)\stackrel{t \rightarrow \infty}{\longrightarrow} 
y_\infty$. Moreover, the solution $(x(t),y(t))$ is globally Lipschitz, thus $\dot{y}$ converges to zero. Then, combining \eqref{eq:system_cexample} with $\phi(|y(t)-x(t)|)\!\geq\!m\!>\!0$, we deduce that $x(t)\stackrel{t \rightarrow \infty}{\longrightarrow} y_\infty$. Since there is only one equilibrium in the domain $\overline{\mathcal{D}}_\varepsilon$, we necessarily have $y_\infty=1$, and therefore $(x(t),y(t)) \stackrel{t \rightarrow \infty}{\longrightarrow} (1,1)$.

\subsection{Consensus in local non-symmetric opinion dynamics}\label{sec:local_non_symm_od}
Next we turn to consider the question of consensus for the non-symmetric opinion model (\ref{eq:opinion_formationb}).
\begin{theorem}[Connectivity implies consensus: non-symmetric opinion dynamics]\label{thm:con_non-symmetric}
Let $\P(t)=\{\bp_k(t)\}_k$ be the solution of the non-symmetric   opinion  dynamics \eqref{eq:opinion_formationb} with compactly supported influence function, $\Supp=[0,\D)$,
\[
\deg_i\frac{d}{dt}\bx_i(t)= \alpha \sum_j \phi_{ij}(\bx_i(t)-\bx_j(t)), \qquad \deg_i=\sum_k \phi_{ik}.
\]
If $\P(t)$ remains uniformly connected in time in the sense that each pair of agents $(i,j)$ is connected through a path $\Gamma_{ij}$ such that\footnote{Observe that here we measure connectivity in terms of the influence function $\phi_{ij}$ rather than the adjacency matrix, $a_{ij}$ as  \eqref{eq:ucon}; the two are equivalent up to obvious scaling of the degree $\sigma_i$.} 
\[
\min_{k_\ell \in \Gamma_{ij}} \phi(|\bx_{k_\ell}-\bx_{k_{\ell+1}}|) \geq \mu>0, \qquad \text{for all} \ i,j,
\]
then it has bounded time-variation and consequently, 
 $\P(t)$ approaches an emerging consensus, $\bx_i(t) \rightarrow \con{\bx}$ with a convergence rate,
\begin{equation}\label{eq:erate_non-symmetric}
|\bx_i(t)-\con{\bx}| \lesssim e^{-\alpha m (t-t_0)} \dm{\bx(0)}, \qquad m=\min_{r\leq \D/2}\phi(r)>0.
\end{equation} 
\end{theorem}
\begin{proof}
We introduce the energy functional,
\begin{subequations}\label{eqs:energy} 
\begin{equation}\label{eq:energya}
  {\mathcal E}(t) :=  \alpha\sum_{i,j} \Phi(|{\bx}_j(t)-{\bx}_i(t)|), \qquad \Phi(r):=\int_{s=0}^r s\phi(s)ds,
\end{equation}
which is decreasing in time,  
\begin{eqnarray}
  \frac{d}{dt}{\mathcal E}(t) &=& \alpha\sum_{i,j} \phi_{ij} \langle\dot{\bx}_j-\dot{\bx}_i\,,\,{\bx}_j-{\bx}_i\rangle 
  = -2\alpha\sum_{i,j} \phi_{ij} \langle\dot{\bx}_i\,,\,{\bx}_j-{\bx}_i\rangle \label{eq:energyb}\\
  &=& -2 \sum_i \langle \dot{\bx}_i\,, \alpha\sum_{j\neq i} \phi_{ij}\,({\bx}_j-{\bx}_i)\rangle =   -2\sum_{i} \deg_i |\dot{\bx}_i|^2 \leq 0.
\nonumber
 \end{eqnarray}
\end{subequations}
To upperbound the expression on the right of (\ref{eq:energya}), sum (\ref{eq:opinion_formationb}) against $\bx_i$ to find,
\begin{eqnarray}\label{eq:veebound}
 \frac{\alpha}{2}\sum_{i,j}\phi_{ij}|\bx_i-\bx_j|^2 & = & -\alpha \sum_{i,j} \phi_{ij} \langle \bx_i-\bx_j,\bx_i\rangle = \sum \deg_i \langle \bx_i,\dot{\bx}_i\rangle \\
& \leq &  \sqrt{ \sum_i\deg_i |\bx_i|^2}\sqrt{\sum_i \deg_i |\dot{\bx}_i|^2}
\leq N\max_i|\bx_i(0)|\sqrt{\sum_i \deg_i |\dot{\bx}_i|^2}. \nonumber
\end{eqnarray}
We end up with the energy decay
\begin{equation}
  \label{eq:decay_energy_E}
  \frac{d}{dt}{\mathcal E}(t) \leq -\frac{1}{2}\alpha^2C_0^2 \left(\sum_{i,j}\phi_{ij}|\bx_i-\bx_j|^2\right)^2, \qquad C_0= \frac{1}{N \max_i |\bx_i(0)|}.  
\end{equation}
Hence, since 
\[
\int^\infty \left(\sum_{i,j} \phi_{ij}(t)|\bx_i(t)-\bx_j(t)|^2\right)^2 dt < \frac{2}{\alpha^2C_0^2} {\mathcal E}(0) <\infty,
\]
the sum $\sum_{i,j} \phi_{ij}(t)|\bx_i(t)-\bx_j(t)|^2$ must become arbitrarily small at some point of time, namely, there exists $t_0>0$ such that
\begin{equation}\label{eq:arrival_decay}
\sum_{i,j} \phi_{ij}(t_0)|\bx_i(t_0)-\bx_j(t_0)|^2 \leq \frac{\mu}{4N} \D^2,
\end{equation}
and by uniform connectivity, consult (\ref{eq:uni}),  
\begin{equation}\label{eq:unii}
\frac{\mu}{N} |\bx_i(t_0)-\bx_j(t_0)|^2 \leq   \sum_{k_\ell\in \Gamma_{ij}} \phi_{k_\ell,k_{\ell+1}}(t_0) |\bx_{k_\ell}(t_0)-\bx_{k_{\ell+1}}(t)|^2 \leq \frac{\mu}{4N} \D^2.
\end{equation}
Thus, the  dynamics at time $t_0$ concentrate so that its diameter, $\dm{\bx(t_0)} = \max_{i,j}|\bx_i(t_0)-\bx_j(t_0)| \leq \D/2$, and since $\dm{\bx(\cdot)}$ is non-increasing in time, $\dm{\bx(t)}\leq \D/2$ thereafter. Arguing along the lines of proposition \ref{prop:CV_opinions}, we conclude  that there is an exponential time decay,
\[
Na_{ij} \geq \phi(|\bx_i(t)-\bx_j(t)|) \geq \min_{r\leq \dm{\bx(t)}}\phi(r) \geq \min_{r\leq \D/2}\phi(r) = m, \qquad t>t_0,
\]
and consensus follows from corollary (\ref{cor:act}).
\end{proof}

The decreasing energy functional ${\mathcal E}(t)$ can be used to estimate the first ``arrival'' time of concentration $t_0$. To this end, observe that:
\begin{displaymath}
  \Phi(|\bx_j-\bx_i|) = \int^{|\bx_j-\bx_i|}_{s=0} s\phi(s)ds \leq M\int^{|\bx_j-\bx_i|}_{s=0} sds = M \frac{|\bx_j-\bx_i|^2}{2},
\quad M:=\max_r \phi(r).
\end{displaymath}
Using the assumption of uniform connectivity, there exists $\mu>0$ and a path $\Gamma_{ij}$ such that:
\begin{displaymath}
  |\bx_j-\bx_i|^2 \leq \frac{N}{\mu} \sum_{k_\ell \in \Gamma_{ij}} \phi_{k_\ell,k_{\ell+1}} |\bx_{k_{\ell+1}}-\bx_{k_\ell}|^2 \leq \frac{N}{\mu} \sum_{ij}\phi_{i,j}|\bx_j-\bx_i|^2.
\end{displaymath}
Combining the last two inequalities, we can upperbound the energy $\mathcal{E}$:
\begin{displaymath}
  \mathcal{E} = \sum_{ij} \Phi_{ij} \leq \frac{M N^3}{4\mu} \sum_{ij}\phi_{i,j}|\bx_j-\bx_i|^2.
\end{displaymath}
Hence, (\ref{eq:decay_energy_E}) implies the Riccati equation
\begin{displaymath}
  \frac{d}{dt}{\mathcal E}(t) \leq -\frac{1}{2}\alpha^2C_0^2 \left(\frac{4\mu}{MN^3}\mathcal{E}  \right)^2 = -\frac{C\mu^2}{N^6} \mathcal{E}^2, 
\end{displaymath}
which shows the energy decay
\begin{displaymath}
  \mathcal{E}(t) \lesssim \frac{1}{1+\displaystyle \frac{C\mu^2t}{N^6} }.
\end{displaymath}
Thus, the arrival time of concentration $t_0$ (\ref{eq:arrival_decay}) is at most of the order of ${\mathcal O}(N^7/\mu^3)$.
This bound on the first arrival time can be improved.\footnote{In fact, the  energy $\mathcal{E}(t)$ decays exponentially in time.}

We close this section by noting the lack of a consensus proof for our  non-symmetric model of flocking dynamics \eqref{eq:MT_model} is due to the lack of a proper decreasing energy functional.
  
\section{Heterophilious dynamics enhances consensus --- simulations}\label{sec:heterophily}
\setcounter{equation}{0}
\setcounter{figure}{0}
As we noted earlier, the large-time behavior of local models for self-organized dynamics depend on the details of the interactions, $\{a_{ij}\}$, and in the particular case of local models (\ref{eqs:frame}), on the profile of the compactly supported influence function $\phi$. 
Here  we explore how  the profile of  $\phi$ dictates  cluster formation in the opinion dynamics model (\ref{eq:opinion_formationb}). The numerical simulations presented in this section leads to the  main conclusion  that an \emph{increasing} profile of $\phi$ reduces the number of clusters $\{\C_k\}_{k=1}^K$.  In particular, if the profile of $\phi$ is increasing fast enough, then $K=1$; thus, heterophilious dynamics enhances the emergence of consensus.   

In the following, we employ a compactly supported influence function $\phi$ which is a simple step  function,
\begin{equation}
  \label{eq:phi_step}
  \phi(r) = \left\{
    \begin{array}{cl}
      a & \quad \text{for } r\leq \frac{1}{\sqrt{2}} \\
      b & \quad \text{for } \frac{1}{\sqrt{2}} < r\leq 1 \\
      0 & \quad \text{for } r>1.
    \end{array}
  \right.
\end{equation}
The essential quantity here is the ration $b/a$ which measures the balance between the influence of ``far'' and ``close'' neighbors (see figure \ref{fig:phi_ab}). We initiate the opinion dynamics (\ref{eq:opinion_formationb}) with  random initial configuration $\{{\bf x}_i(0)\}_i$.

\begin{figure}[ht]
  \centering
  \includegraphics[scale=.7]{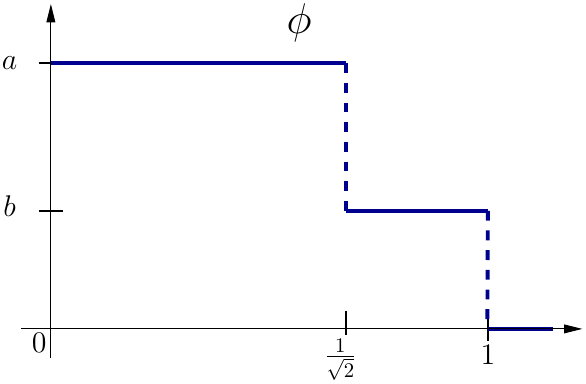}
  \caption{{\small Influence functions $\phi$ used in the simulations. The larger $b/a$ is, the more heterophilious is the dynamics.}}
  \label{fig:phi_ab}
\end{figure}

\subsection{1D simulations}
We begin with four simulations of the 1D opinion dynamics (\ref{eq:opinion_formationb})  subject to  $100$  opinions distributed uniformly on $[0,10]$, the same initial configuration as in  figure \ref{fig:evolutionXratio1_bis}. To explore the impact of the influence step function  (\ref{eq:phi_step}) on the dynamics, we used four different ratios of $b/a=.1, \, 1, \,2$ and $10$.  As $b/a$ increases, we reduce the influence of the closer neighbors and increase the influence of  neighbors further away; thus, increasing $b/a$ reflects  the tendency to ``bond with the other''.
As observed in figure \ref{fig:ratio011210}, the increase in the ratio $b/a=.1, \, 1, \,2$ and $10$, reduces the corresponding number of limit clusters to $K=6, \ 4, \ 2$, and for $b/a=10$, the dynamics  converged to a consensus, $K=1$. 
The simulations of figure \ref{fig:ratio011210} indicate that  reducing the influence of closer neighbors and hence increasing the weight for the  influence of neighbors further away,  will favor increased \emph{connectivity} and the emergence of consensus.

To make a systematic analysis of the cluster formation dependence on the ratio $b/a$, we made several simulations with random initial conditions for a given ratio $b/a$. Then we make an average of the number of clusters, denoted by $\langle S\rangle$, at the end of each simulation ($t=100$). To compute the number of clusters, we estimate the number of connected components of the matrix $A$ (\ref{eq:matrix_A}) using a depth-first search algorithm. As observed in figure \ref{fig:S_depending_bDiva}, the number of clusters $\langle S\rangle$ decreases as $b/a$ increases. Moreover, $\langle S\rangle$ approaches $1$ when $b/a$ approaches $10$, implying that a consensus is likely to occur when $b/a$ is large enough.

\begin{figure}[ht]
  \centering
  \includegraphics[scale=.35]{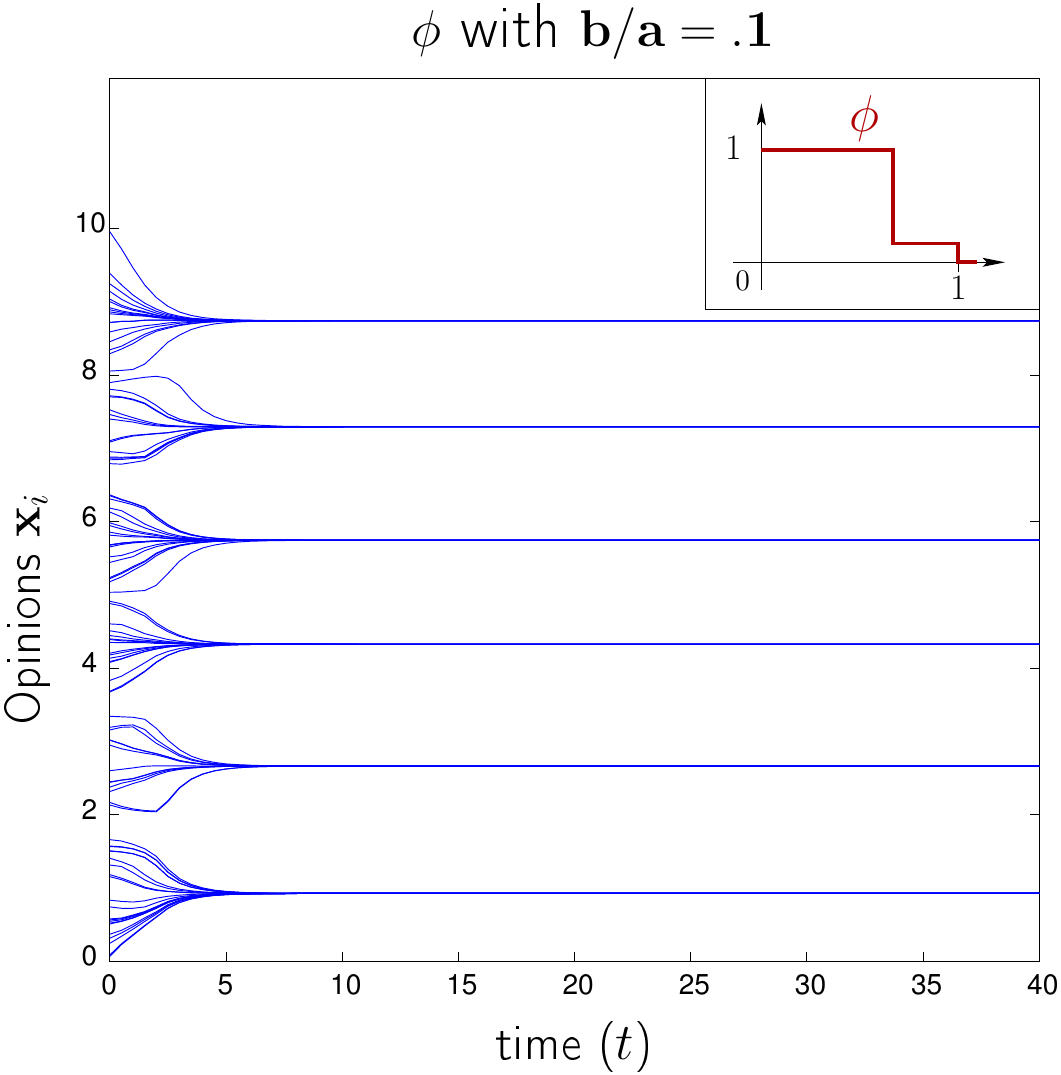} \quad
  \includegraphics[scale=.35]{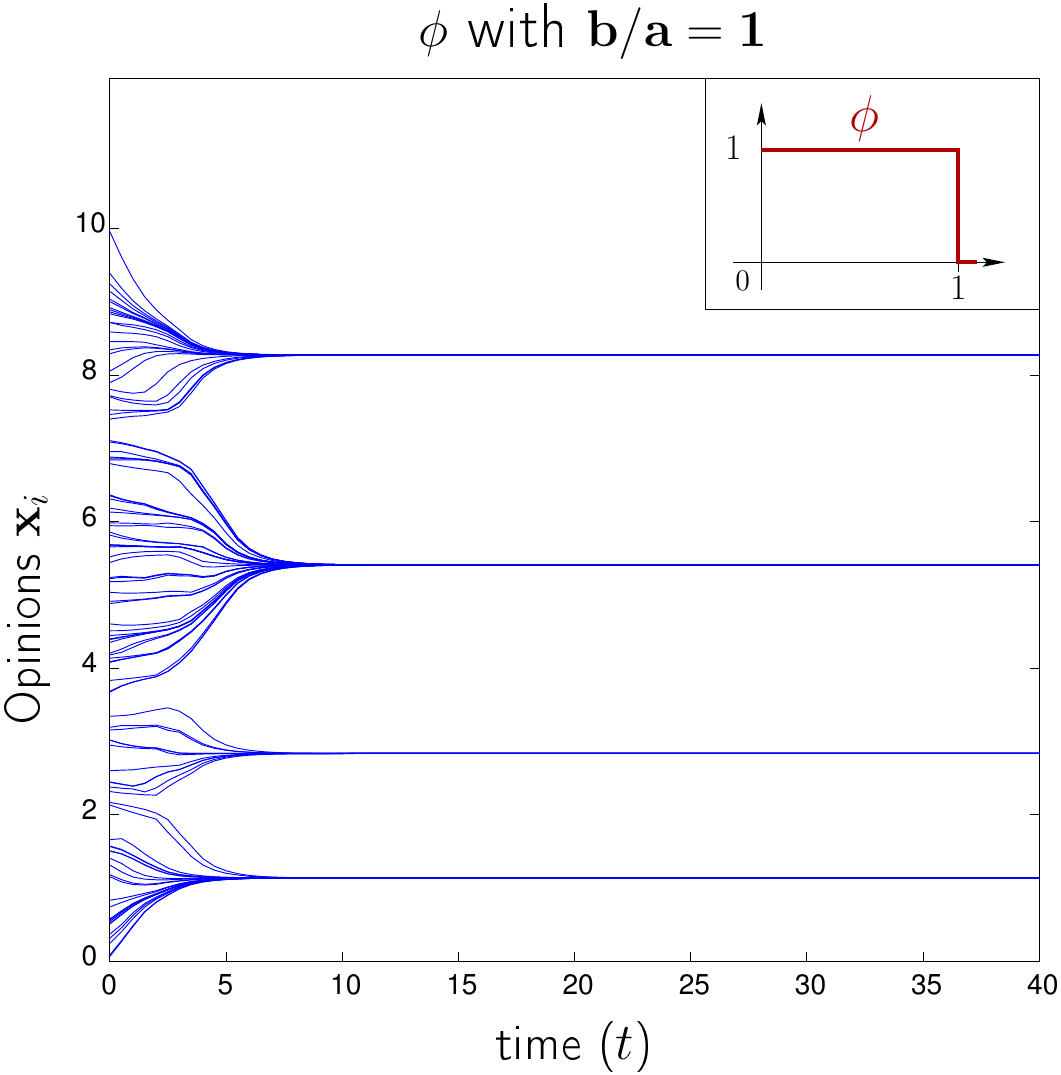}

  \bigskip
  \bigskip
  
  \includegraphics[scale=.35]{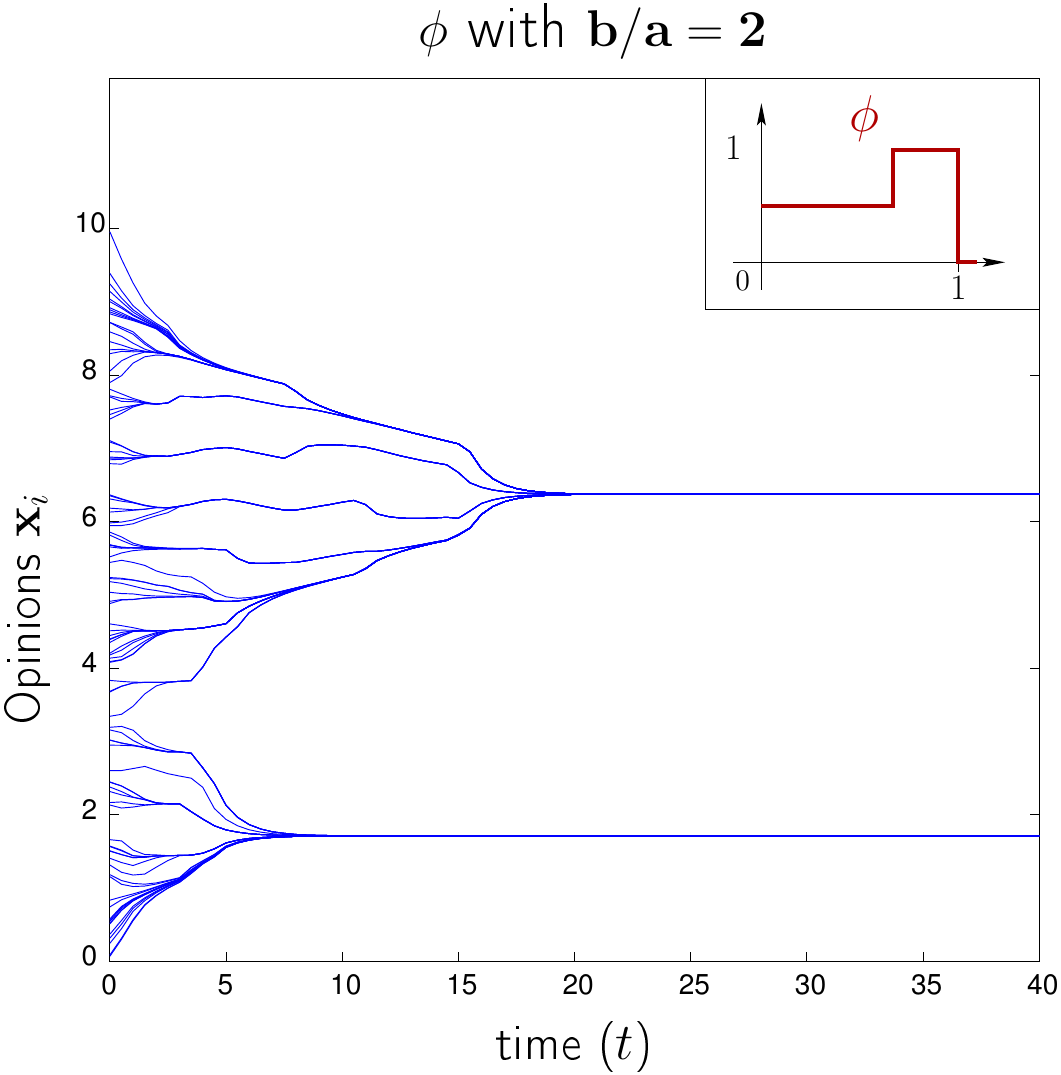} \quad
  \includegraphics[scale=.35]{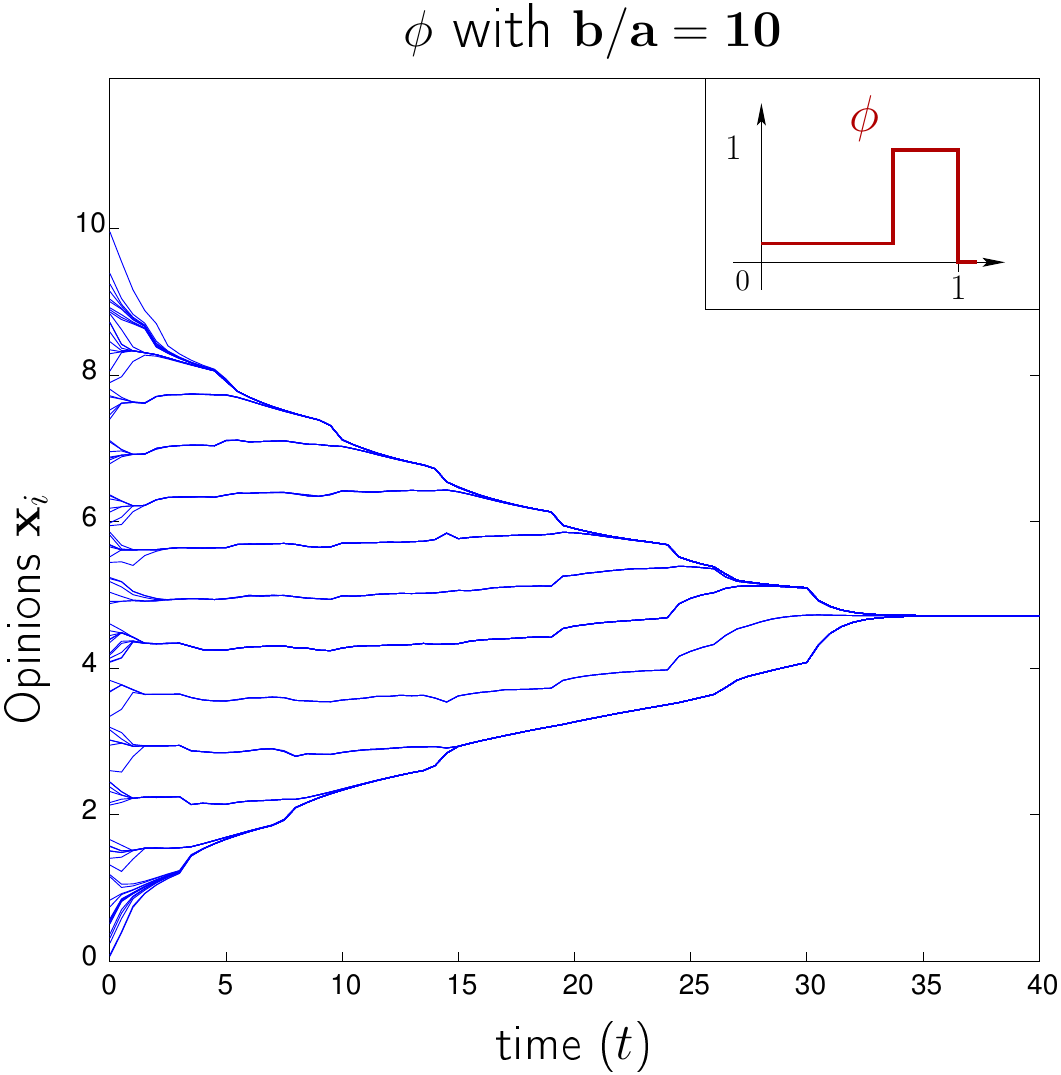}
  \caption{{\small Simulation of the opinion dynamics model with different interacting function $\phi$. When the influence of close neighbors is reduced (i.e. $b/a$ large), the number of cluster decreases. For $b/a=10$, the dynamics converges to a consensus.}}
  \label{fig:ratio011210}
\end{figure}

\begin{figure}[ht]
  \centering
  \includegraphics[scale=.35]{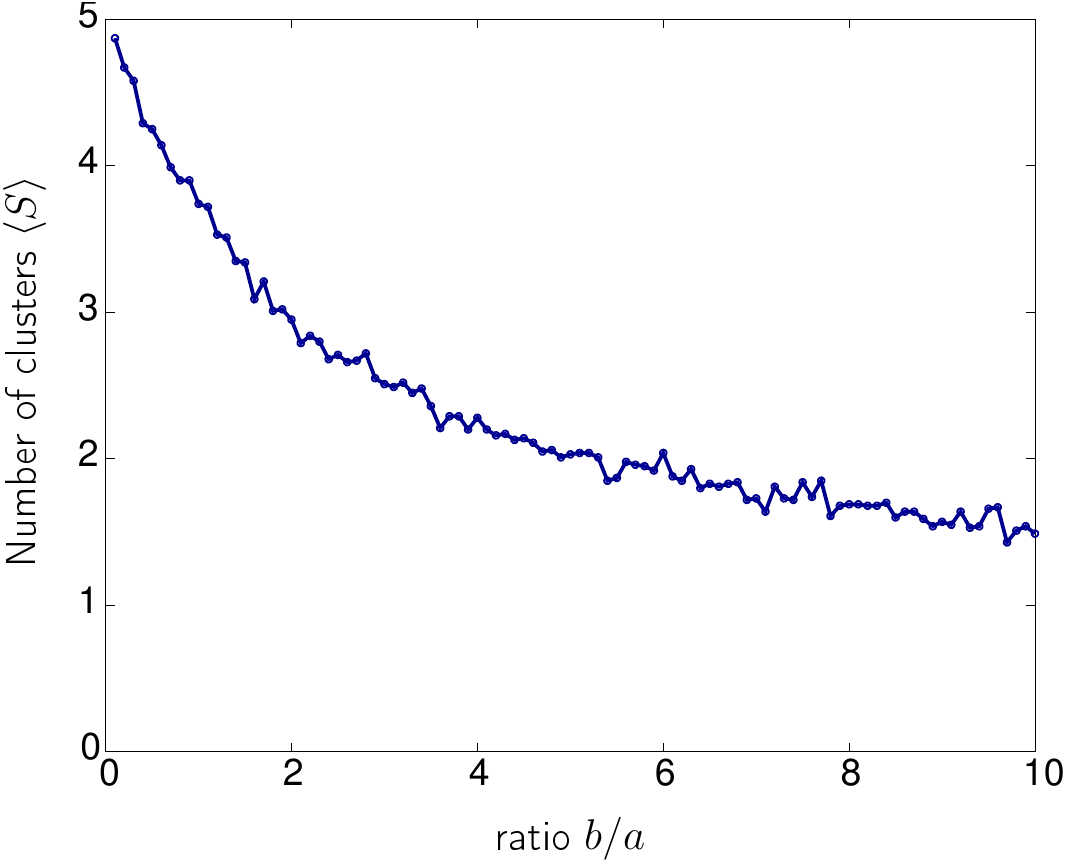} \quad
  \includegraphics[scale=.35]{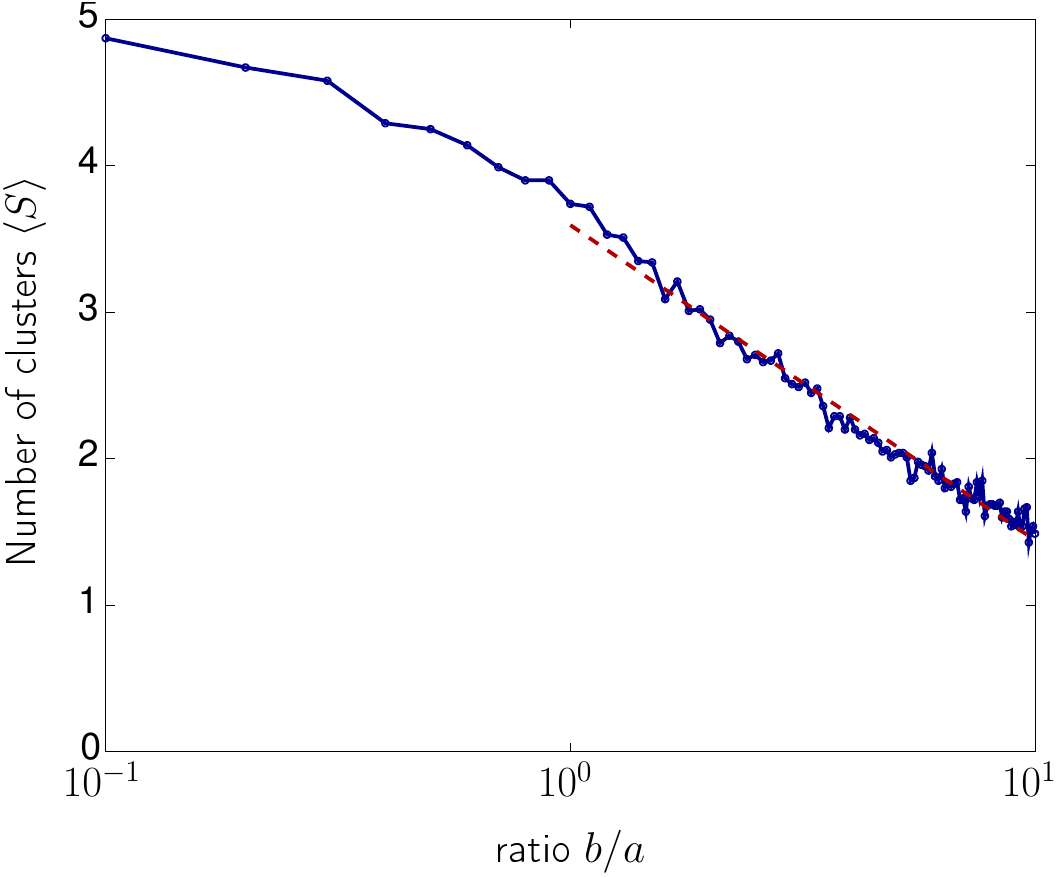}
  \caption{{\small Average number of clusters $\langle S\rangle$ depending on the ratio $b/a$ ({\bf Left} figure). The larger $b/a$ is, the fewer the number of clusters. The decay is logarithmic on $[1,10]$ ({\bf Right} figure). For each value of $b/a$, we run $100$ simulations to estimate the mean number of clusters $\langle S\rangle$. Simulations are run with $\Delta t=.05$ and a final time equals to $t=100$ unit time.}}
  \label{fig:S_depending_bDiva}
\end{figure}

\subsection{Clusters and branches}\label{sec:branches}

We revisit the opinion model (\ref{eq:opinion_formationb}) with an influence step function (\ref{eq:phi_step}).
As noted before, the increasing value of  $b/a$  increases  the probability to reach a consensus.
The simulations in figure \ref{fig:ratio011210} with  $b/a=2$ and with $b/a=10$, show  the apparition of \emph{branches}, where subgroups of  agents have converged to the same opinion yet, in contrast to clustering, these branches of opinions are still interacting with  outsiders, which are in distance which is strictly less than $\D=1$.  
 In particular, when  $b/a=10$, the distribution of opinions $\{{\bf x}_i(t)\}_i$, aggregate to form distinct branches seen in figure \ref{fig:ratio011210}: at $t \sim 5$, one can identify in figure  \ref{fig:histXratio10_t5}, the formation of $10$ branches which are separated by a  distance of approximately $.7$  spatial units. Since the distance between two such branches is always less than the diameter $\D=1$ of $\phi$, these branches are not qualified as isolated clusters, as they continue to be influenced by ``outsiders'' from the nearby branches.   Over time, these branches merge into each other  before they emerge into one final cluster, the consensus,  at $t\sim 33$.
Thus, the decisive factor in the consensus dynamics is not the number of branches but their large time connected components.   Indeed,  figure \ref{fig:ratio011210} with $b/a=10$, shows that the agents in the different branches remain in the same connected component at distance $\sim .7$,  corresponding  to the discontinuity of   $\phi(\cdot)$, which experience a  jump from $.1$ to $1$ at $1/\sqrt{2}\approx.7$. 

To illustrate the apparition of the distance $1/\sqrt{2}$ between two nearest branches, we repeat the simulations, this time with  special initial configurations where all the opinions are uniformly spaced with $|{\bx}_{i+1}-{\bx}_i|=d_*$ with $0<d_*<1$. As we observe in figure \ref{fig:periodic_formation}, the agents $\{{\bx}_i\}_i$ readjust their ``opinion'' such that the distance between nearest neighbors $|{\bx}_{i+1}-{\bx}_i|$ approaches  $1/\sqrt{2}$ as $t\gg 1$.

\begin{figure}[ht]
  \centering
  \includegraphics[scale=.30]{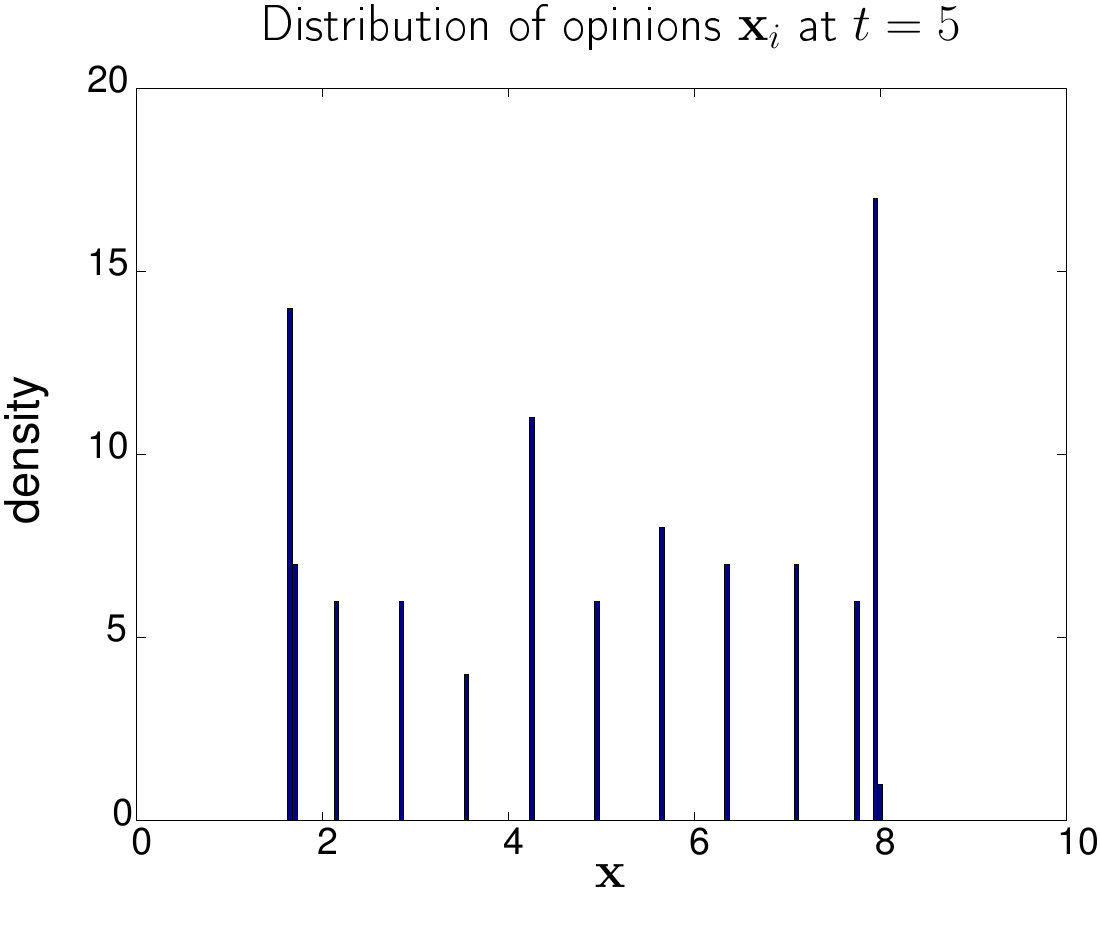}
  \caption{{\small The distribution of $\{{\bx}_i\}_i$ in the simulation of figure \ref{fig:ratio011210} with $b/a=10$ at time $t=5$. The distance between two picks of density is around $1/\sqrt{2}\approx.7$ space unit. This distance corresponds to the discontinuity of the function $\phi(r)$.}}
  \label{fig:histXratio10_t5}
\end{figure}

\begin{figure}[ht]
  \centering
  \includegraphics[scale=.30]{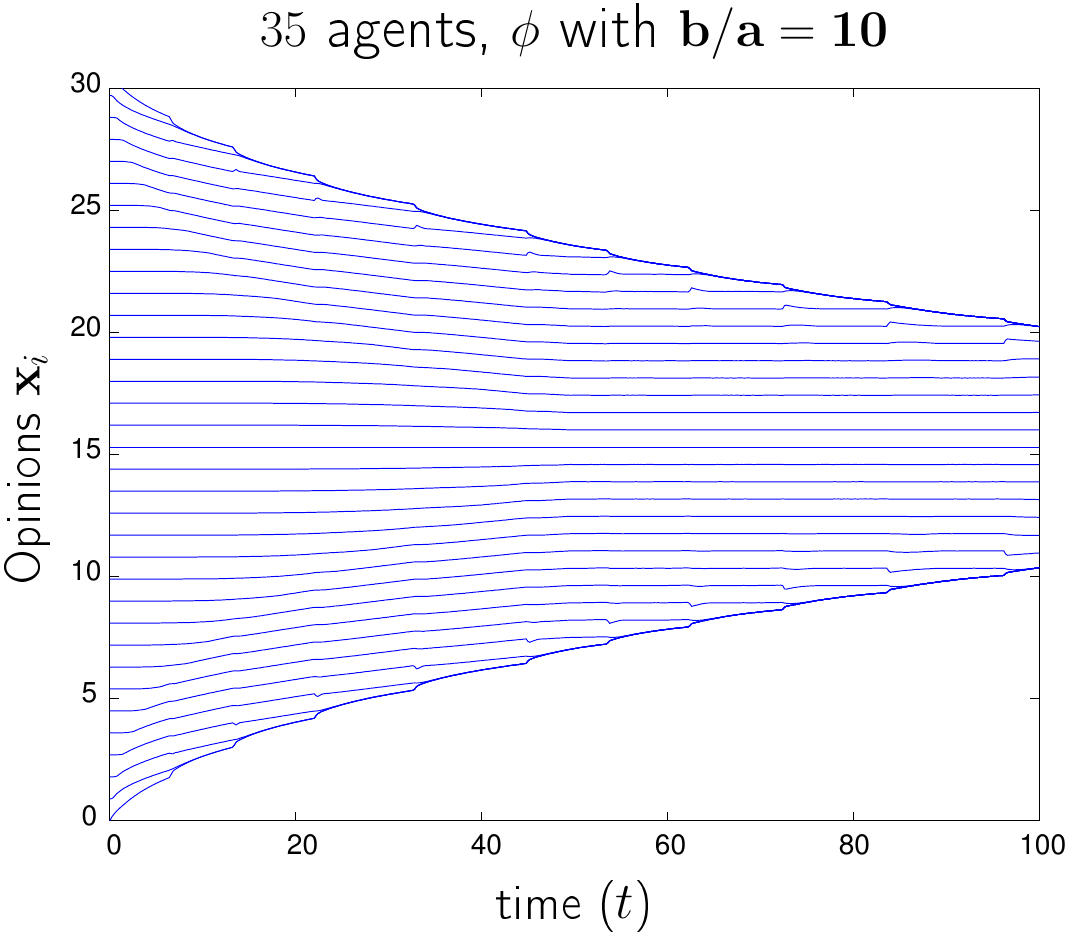}
  \includegraphics[scale=.30]{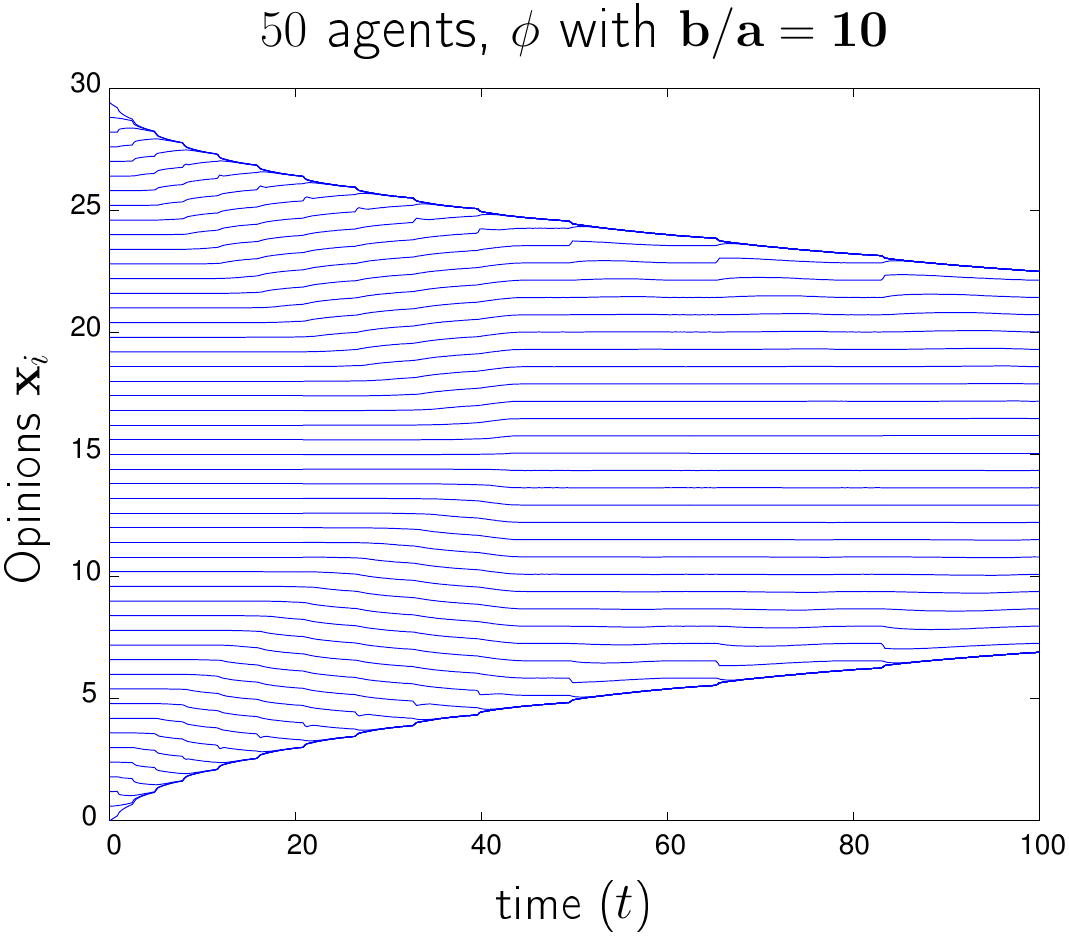}
  \caption{{\small Initial condition with an equi-repartition of $\{{\bf x}_i\}_i$: $|{\bf x}_{i+1}-{\bf x}_i|=.9$ ({\bf Left} figure) and $|{\bf x}_{i+1}-{\bf x}_i|=.6$ ({\bf Right} figure). Nearest neighbor readjust their distance to $1/\sqrt{2}\approx.7$ unit space, we observe a concentration of the trajectories in the left figure and a spread of the trajectories in the right figure.}}
  \label{fig:periodic_formation}
\end{figure}

\subsection{2D simulations}

We made several 2D simulations with different influence functions $\phi$. As a first illustration, we made a 2D simulation with the same initial configuration used in figure \ref{fig:simu_2D}, but this time we used the influence step function  $\phi$ in (\ref{eq:phi_step}) with $b/a=10$. In figure \ref{fig:simu_2D_ba10}, one can observe a concentration phenomenon (from $t=0$ to $t=2.5$) --- the opinions aggregate into $5$ final clusters, compared with the $17$ clusters  observed in figure \ref{fig:simu_2D} with the influence function $\phi=\chi_{[0,1]}$. Thus, as in the 1D case, a more heterophilious influence function increases the clustering effect.

We also estimate the average number of clusters $\langle S\rangle$ depending on the ratio $b/a$. As observed in figure \ref{fig:S_depending_bDiva_2D}, $\langle S\rangle$ is a decreasing function of $b/a$ and once again the decay of $\langle S\rangle$ as a function of $b/a \in [0,10]$ is \emph{logarithmic}.

\begin{figure}[p]
  \centering
  \includegraphics[scale=.45]{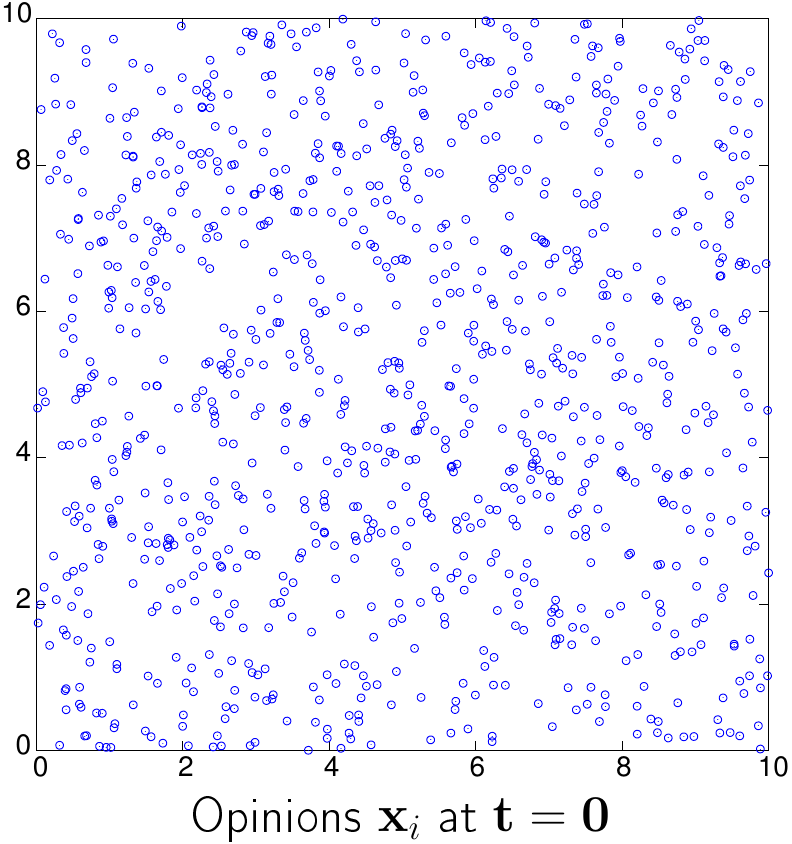} \quad
  \includegraphics[scale=.45]{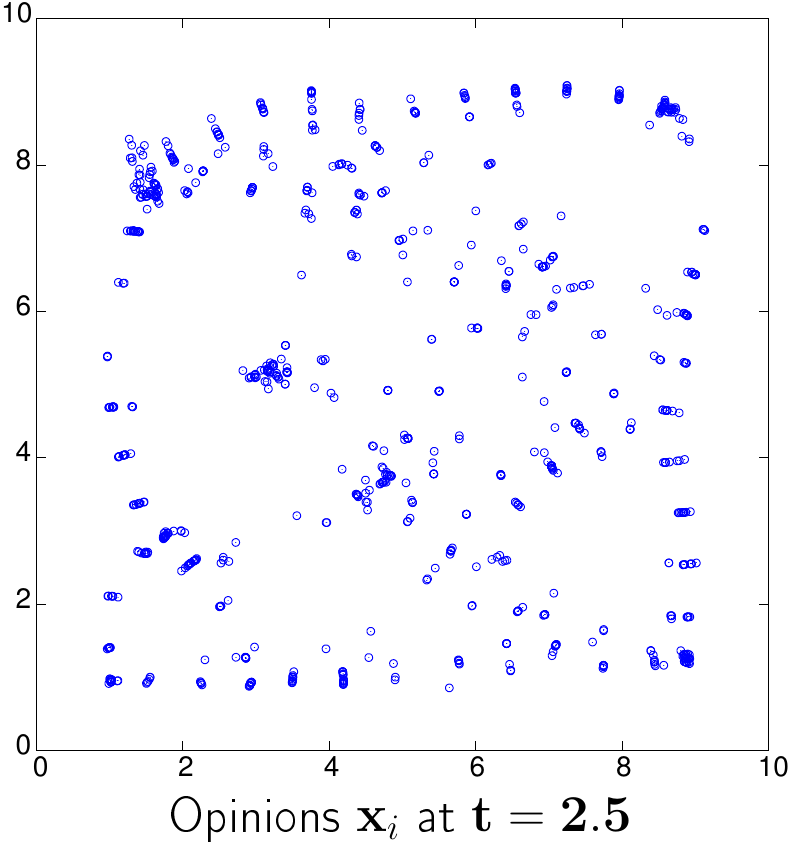}

  \bigskip

  \includegraphics[scale=.45]{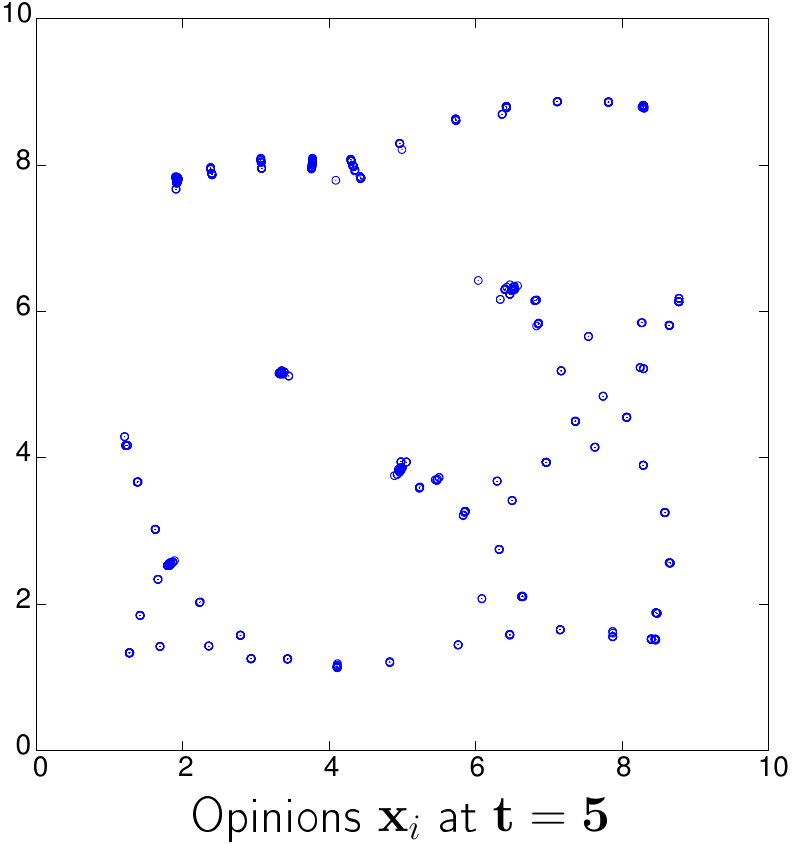} \quad
  \includegraphics[scale=.45]{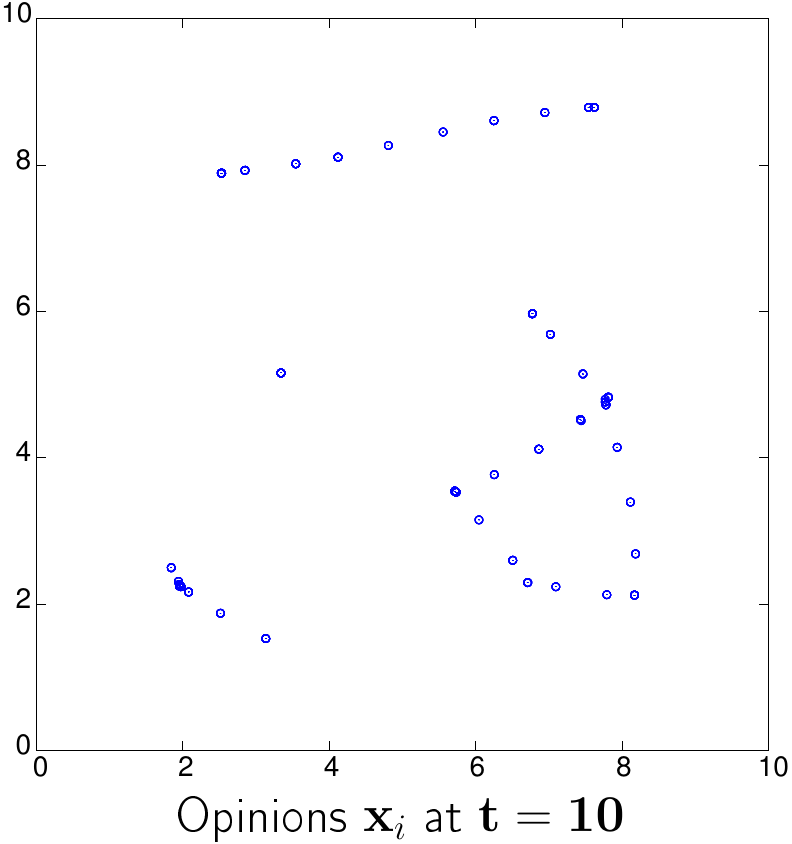}

  \bigskip
  
  \includegraphics[scale=.45]{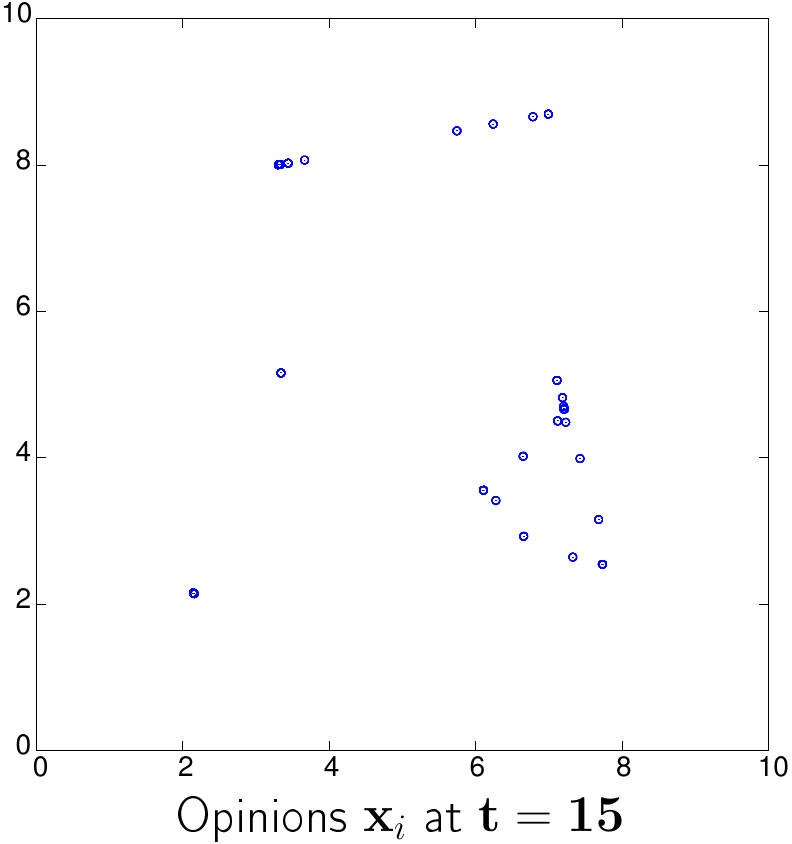} \quad
  \includegraphics[scale=.45]{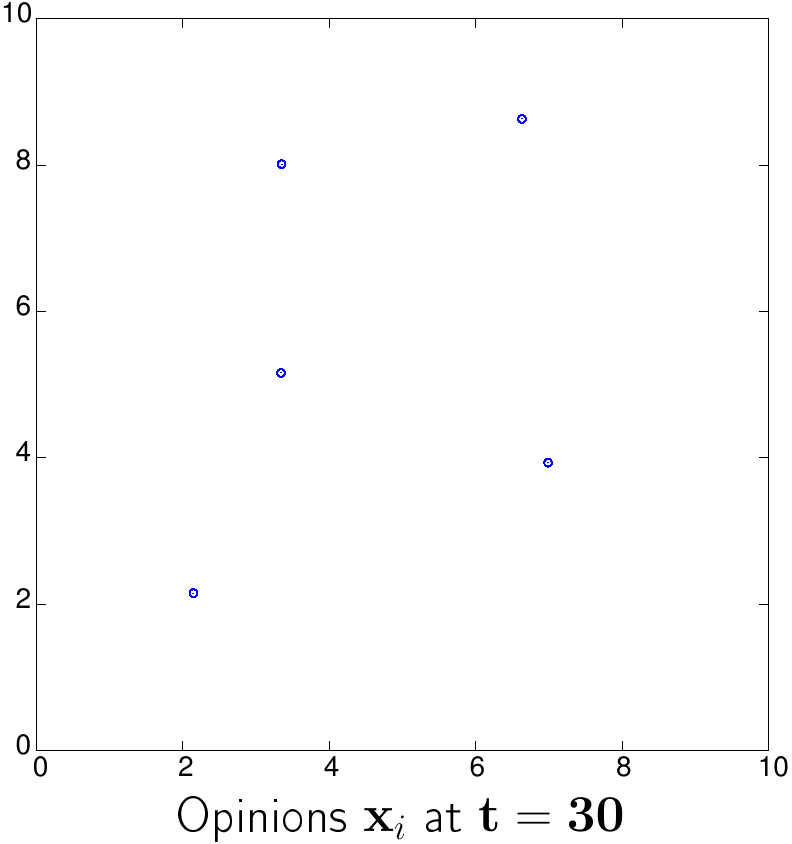}
  \caption{{\small The heterophilious effect: diminishing the influence of close neighbors relative to those further away, increases the clustering effect. 2D simulation of the opinion model (\ref{eq:opinion_formationb}) with $M=1000$ agents using a step influence function  $\phi=.1\,\chi_{[0,1/\sqrt{2}]} \,+\,\chi_{[1/\sqrt{2},1]}$ leads to $5$ clusters which remains at the end of the simulation. This should be  compared with  $17$ clusters with $\phi=\,\chi_{[0,1]}$ (see figure \ref{fig:simu_2D}).}}
  \label{fig:simu_2D_ba10}
\end{figure}

\begin{figure}[ht]
  \centering
  \includegraphics[scale=.30]{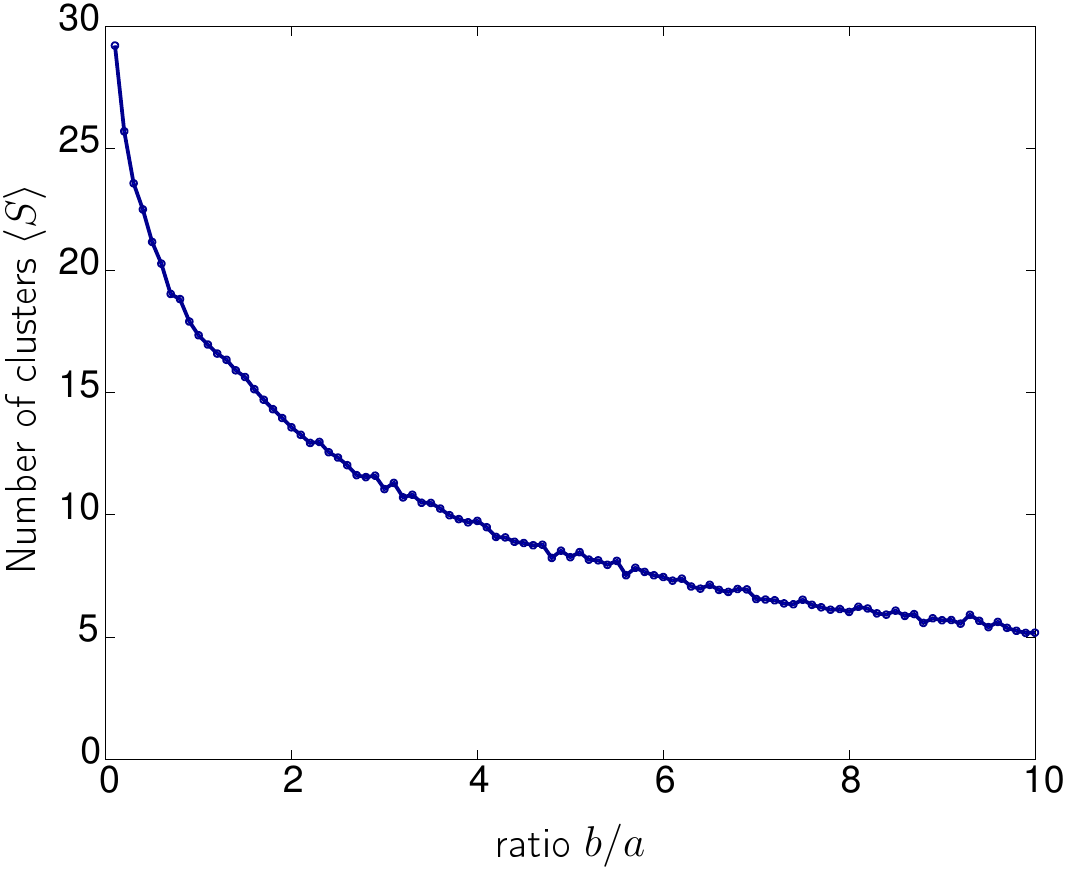} \quad
  \includegraphics[scale=.30]{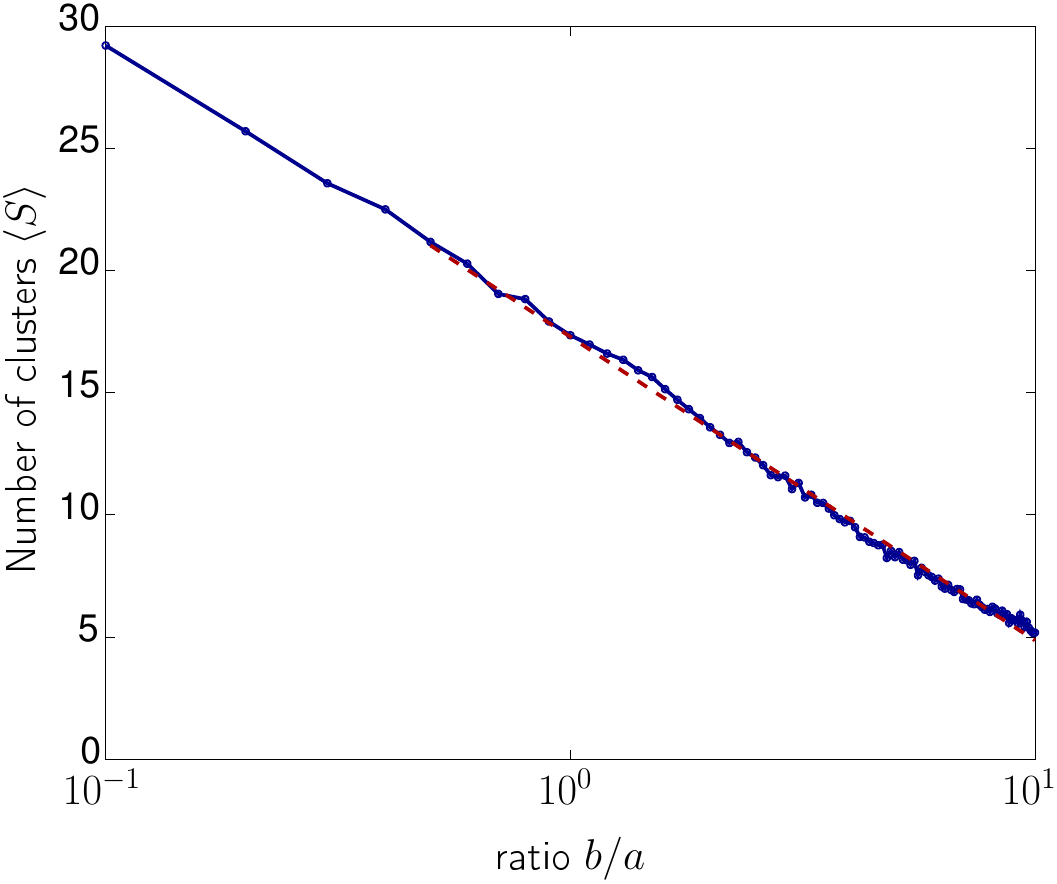}
  \caption{{\small Average number of clusters $\langle S\rangle$ depending on the ratio $b/a$ in 2D ({\bf Left} figure). As in 1D case, the larger $b/a$ is, the fewer the number of clusters, and the decay is logarithmic on $[1,10]$ ({\bf Right} figure). For each value of $b/a$, we made $100$ simulations to estimate the mean number of clusters $\langle S\rangle$. Simulations were made  with $\Delta t=.05$ and were recorded  at the final time  $t=100$.}}
  \label{fig:S_depending_bDiva_2D}
\end{figure}

\section{Heterophilious dynamics with  a fixed-number of neighbors}\label{sec:nearest_neighbor}
\setcounter{equation}{0}
\setcounter{figure}{0}

Careful observations of startling flocks led the Rome group \cite{cavagna_et_al_2008i, cavagna_et_al_2008ii, cavagna_et_al_2010} to the  fundamental conclusion that their dynamics is driven by local interaction with a \emph{fixed} number  of nearest neighbors.       
This motivates our study of nearest neighbor models for opinion dynamics which take the form 
\begin{subequations}\label{eqs:near}
\begin{equation}
  \label{eq:model_connectivity}
  \frac{d}{dt}{\bx}_i = \alpha\!\!\!\!\! \sum_{\{j:\, |j-i|\leq \nei\}}\frac{\phi_{ij}}{\deg_i}  ({\bx}_j-{\bx}_i), \qquad \bx_i \in {\mathbb R}^d,
\end{equation}
where the degree $\deg_i$ is given by one of two forms, depending on the symmetric and non-symmetric version of the opinion dynamics in (\ref{eqs:opinion})
\begin{equation}
\left\{\begin{array}{ll}\label{eq:neardeg}
\text{the symmetric case}: \ \ & \displaystyle \deg_i=\frac{1}{2\nei} \\ \\
\text{the nonsymmetric case}: \ \ & \displaystyle \deg_i=\sum_{\{j:|j-i|\leq \nei\}} \phi_{ij}.
\end{array}\right.
\end{equation}
\end{subequations} 
 
Thus,  each agent $i$  is assumed to interact only with its $2q$ agents $i-\nei,\ldots,i+\nei$. Typically, $\nei$ is small (the observation in \cite{cavagna_et_al_2008i, cavagna_et_al_2008ii, cavagna_et_al_2010} report on six to seven active nearest neighbors).
We analyze the connectivity of the particular case of  \emph{two nearest neighbors}, $q=1$. Here we   prove that such local models  preserve connectivity and hence converge to a consensus provided the influence function $\phi$ is increasing.  This result supports our findings in section \ref{sec:heterophily} that heterophilious  dynamics is an efficient strategy to reach a consensus. 

\subsection{A fixed-number of neighbors with global influence function}
We begin by noting that the different approaches for consensus of global models apply in the present framework  of local nearest neighbor models (\ref{eqs:near}). For example,  consider the non-symmetric nearest neighbor model
\begin{equation}\label{eq:nearnonsym}
\deg_i\frac{d}{dt}{\bx}_i = \alpha \sum_{\{j:|j-i|\leq \nei\}}{\phi_{ij}}  ({\bx}_j-{\bx}_i), \qquad \deg_i=\sum_{\{j:|j-i|\leq \nei\}} \phi_{ij}.
\end{equation}
It admits an energy functional,   
\begin{displaymath}
  {\mathcal E}(t) := \alpha\!\!\!\!\! \sum_{\{i,j: |i-j|\leq\nei \}} \Phi(|{\bx}_j(t)-{\bx}_i(t)|), \qquad \Phi(r):=\int_{s=0}^r s\phi(s)ds,
\end{displaymath}
which is decreasing in time, ${\mathcal E}(t) \leq {\mathcal E}(0)$ and we conclude 

\begin{theorem}\label{thm:nearg}$($Global connectivity$)$
 Consider the nearest neighbor model \eqref{eq:nearnonsym} with an influence function $\phi, \Supp=[0,\D)$ and assume $\alpha\Phi(\D)>{\mathcal E}(0)$.  Then  
$\min_{|i-j|\leq \nei}\phi_{ij}(t)>m_\infty$ where $m_\infty:=\min_{r< \D} \phi(r)$.
Hence, the nearest neighbor dynamics \eqref{eq:nearnonsym} remains connected and consensus follows.
\end{theorem}
\begin{proof}
  Since ${\mathcal E}$ is decreasing in time,
\[
    \alpha\Phi(|{\bx}_{i}(t)-{\bx}_j(t)|) < {\mathcal E}(0) \leq \alpha\Phi(\D) \qquad \text{ for any } |i-j|\leq \nei, 
\]
 and since $\Phi(r)=\int^r s\phi(s)ds$ is an increasing function, $|{\bx}_{i}(t)-{\bx}_j(t)| < \D$, hence $\phi_{ij}>m_\infty$ and consensus follows.
\end{proof}

We note, however, that since $m_\infty \leq \phi \leq 1$, then $\Phi(r)$ has a quadratic bounds, 
$m_\infty r^2 \leq 2\Phi(r) \leq r^2$, and hence the assumption made in theorem \ref{thm:nearg} implies
\[
\alpha\Phi(\D)>{\mathcal E}(0) \ \leadsto \ \D^2 > m_\infty \!\!\!\!\sum_{|i-j|\leq \nei} |\bx_i-\bx_j|^2.
\] 
Namely, the support of $\phi$ should be sufficiently large  to cover a globally connected path in phase space.

\subsection{Two-neighbor dynamics}
In this section we prove uniform connectivity and hence convergence to a consensus  of  a symmetric \emph{two} nearest neighbor model, (\ref{eqs:near}),
\begin{equation}
  \label{eq:model_connectivity_symmetric}
  \frac{d}{dt}{\bf x}_i = \frac{\alpha}{2}\Big(\kappa_{i+\hf} ({\bx}_{i+1}-{\bx}_i)  +  \kappa_{i-\hf} ({\bx}_{i-1}-{\bx}_i)\Big), \quad \kappa_{i+\hf}:= \left\{\begin{array}{ll} 0, & i=0,N\\ \phi(|\bx_{i+1}-\bx_i)|), & 1\leq i\leq N.\end{array}\right.
\end{equation}
We assume that the initial configuration of agents can be enumerated such that $\{\bx_i(0)\}_i$ is connected
\begin{equation}\label{eq:twoc}
\max_i |\bx_{i+1}(0)-\bx_i(0)| < \D, \qquad \Supp=[0,\D).
\end{equation}
The configuration of  such ``purely'' local interactions applies to the one-dimensional setup where each agent is initially connected to its left and right neighbors; we emphasize that these configurations are not necessarily restricted to the one dimensional setup. 

Forward differencing of (\ref{eq:model_connectivity_symmetric}) implies that $\Delta_{i+\hf}:= \Delta_{i+\hf}(t) := {\bx}_{i+1}(t)-{\bx}_i(t)$ satisfy 
\begin{eqnarray*}
  \frac{d}{dt}{\Delta}_{i+\hf} &=&  \frac{\alpha}{2}\Big(\kappa_{i+\frac{3}{2}} ({\bx}_{i+2}-{\bx}_{i+1})  +  \kappa_{i+\hf} ({\bx}_{i}-{\bx}_{i+1}) 
   - \kappa_{i+\hf} ({\bx}_{i+1}-{\bx}_i)  -  \kappa_{i-\hf} ({\bx}_{i-1}-{\bx}_i)\Big)\\
   &=& \frac{\alpha}{2}\Big(\kappa_{i+\frac{3}{2}} \Delta_{i+\frac{3}{2}} \;-\; 2 \kappa_{i+\hf} \Delta_{i+\hf} \;+\; \kappa_{i-\hf}\Delta_{i-\hf}\Big), \qquad i=1,2,\ldots, N-1.
\end{eqnarray*}
The missing $\Delta$'s for $i=\hf$ and $i=N+\hf$ are defined as $\Delta_{\hf}=\Delta_{N+\hf}=0$.
Let $\Delta_{p+\hf}$ denote the maximal difference, $|\Delta_{p+\hf}|=\max_i |\Delta_{i+\hf}|$ measured in the $\ell_2$-norm.  Then
\begin{eqnarray*}
\frac{1}{2}\frac{d}{dt}|{\Delta}_{p+\hf}|^2  &=&
 \frac{\alpha}{2}\Big(\kappa_{p+\frac{3}{2}} \langle \Delta_{p+\frac{3}{2}},\Delta_{p+\hf}\rangle \;-\; 2 \kappa_{p+\hf} |\Delta_{p+\hf}|^2 \;+\; \kappa_{p-\hf}\langle \Delta_{p-\hf},\Delta_{p+\hf}\rangle\Big) \\
 & \leq & \frac{\alpha}{2}\left(\kappa_{p+\frac{3}{2}}-2 \kappa_{p+\hf} +\kappa_{p-\hf}\right) |\Delta_{p+\hf}|^2.
\end{eqnarray*}
Now, if $\phi$ is   \emph{non-decreasing} influence function, then 
\[
|\Delta_{p+\hf}|\geq |\Delta_{i+\hf}| \ \ \leadsto \ \ 2\kappa_{p+\hf} = 2\phi(|\Delta_{p+\hf}|) \geq \phi(|\Delta_{p-\hf}|)+ \phi(|\Delta_{p+\frac{3}{2}}|),
\]
and hence $|{\Delta}_{p+\hf}(t)|=\max_i \phi(|\bx_{i+1}(t)-\bx_i(t)|) \leq \max_i \phi(|\bx_{i+1}(0)-\bx_i(0)|)$.
We deduce  the following theorem.
\begin{theorem}\label{thm:two_connect}
Consider the nearest neighbor dynamics \eqref{eq:model_connectivity_symmetric} subject to initial configuration, $\bx(0)$ which is connected,
\eqref{eq:twoc},
\[
\max_i |\bx_{i+1}(0)-\bx_i(0)| < \D, \quad \Supp=[0,\D).
\]
 Assume that the  influence function, $\phi$, is non-decreasing.  Then the dynamics \eqref{eq:model_connectivity_symmetric} remains connected and converges to a consensus, $\con{\bx}=\ave{\bx}(0)$,
  \begin{equation}    \label{eq:rate_cv}
    \sum_i|\bx_i(t) - \ave{\bx}(0)|^2 \lesssim \exp\left(-\frac{2\phi(0) t}{N}\right) \sum_i|\bx_i(0) - \ave{\bx}(0)|^2, \qquad 
\ave{\bx}:=\frac{1}{N} \sum_i \bx_i.
  \end{equation}
\end{theorem}

It is important to notice that theorem \ref{thm:two_connect} requires an non-decreasing influence function. Indeed, the steeper the increase of $\phi$ is, the better the connectivity is. This is concrete  ramification of our main statement that heterophilious dynamics enhances consensus. Note that a two-nearest neighbor dynamics driven by a decreasing $\phi$ will not guarantee consensus as illustrated by the following.

\noindent
{\bf Counterexample}. We revisit the counterexample in section \ref{sec:loc_symm_models}, of five agents symmetrically distributed around $\bx_3(t)\equiv 0$ with $\frac{1}{2} < \bx_4(t)<1 <\bx_5(t) <\frac{3}{2}$, governed by
\begin{eqnarray*}
   \dot{\bx_4} &= & -\phi(|\bx_4|) \bx_4 + \phi(|\bx_5-\bx_4|) (\bx_5-\bx_4), \\
    \dot{\bx}_5 &=& \phi(|\bx_5-\bx_5|) (\bx_4-\bx_5),
\end{eqnarray*}
with a compactly supported  influence function $\phi(r)= (1-r)^2(1+r)^2\chi_{[0,1]}$.
Observe that this configuration amounts to a two-nearest neighbor dynamics. Its concentration into three separate clusters $\{-1,0,1\}$ shown in figure \ref{fig:phase_portrait_cex2}, requires a rapidly decreasing  influence function (to be precise ---  $\phi(r)r \downarrow$ for $r\sim 1$), which is  not covered by the two-nearest neighbors' heterophilious dynamics sought in theorem \ref{thm:two_connect}.

\begin{proof}
 The  adjacency matrix  associated with (\ref{eq:model_connectivity_symmetric}), $\dot{\bx}=\alpha(A\bx-\bx)$ is   given by the tridiagonal matrix $A=\{a_{ij}\}$, given by
\[
a_{ij}=\left\{\begin{array}{ll} 
\displaystyle \hf\kappa_{\frac{i+j}{2}}, \ \  \kappa_{\frac{i+j}{2}}=\phi(|\bx_i-\bx_j|) & |i-j|=1, \\ \\
\displaystyle 1-\hf\kappa_i, \ \  \kappa_i:=\phi_{i,i+1}+\phi_{i,i-1} & i=j.
\end{array}\right.
\]
The corresponding Laplacian associated with $A$ is given by 
  \begin{equation}
    \label{eq:matrix_L}
    L_A = \frac{1}{2}\left[
      \begin{array}{ccccccccc}
        -\kappa_1 & \kappa_{\frac{3}{2}}& &&&&&& \\
        \kappa_{\frac{3}{2}}& -\kappa_2 & \kappa_{\frac{5}{2}} &&&&&& \\
        & \kappa_{\frac{5}{2}} & -\kappa_3 & \kappa_{\frac{7}{2}} &&&&& \\
        & & \ddots & \ddots & \ddots & & & & \\
        & & & \kappa_{N-\frac{3}{2}} & -\kappa_{N-1} & \kappa_{N-\hf} \\
        & & & & \kappa_{N-\hf} & -\kappa_N
      \end{array}
    \right].
  \end{equation}
  Since $|\bx_{i+1}(t)-\bx_i(t)|<\D$, the  off-diagonal entries  $\kappa_{i+\hf} = \phi(|\bx_{i+1}-\bx_i|)>0$ and hence, the graph 
${\mathcal G}_A=(\{\bx(t)\}, A(\bx(t)))$  remains connected with $\mu=\min_i\kappa_{i+\hf}$ and $diam({\mathcal G}_A)=N$. By (\ref{eq:fidlow}) we find 
\[
    \lambda_2(L_A) \geq \frac{\min_i\kappa_{i+\hf}}{N^2} \geq \frac{\phi(0)}{N^2}.
\]  
  Using theorem \ref{thm:symm} (see (\ref{eq:symop})) we end up with
\[
    \sum_i|\bx_i(t)-\ave{\bx}(0)|^2 \lesssim \expo^{\displaystyle -\frac{\alpha \phi(0) t}{N^2}} \sum_i|\bx_i(0)-\ave{\bx}(0)|^2,
\]
  which concludes the proof.
\end{proof}

\begin{myremark}
  The worst case scenario for the decaying of the $|\bx_i(t) - \ave{\bx}|$ is to have many opinions $\bx_i$ concentrate at two extreme values with just one path of opinion connecting the two extremes (see figure \ref{fig:worst_case_symmetricNN}).
\end{myremark}

\begin{figure}[ht]
  \centering
  \includegraphics[scale=0.7]{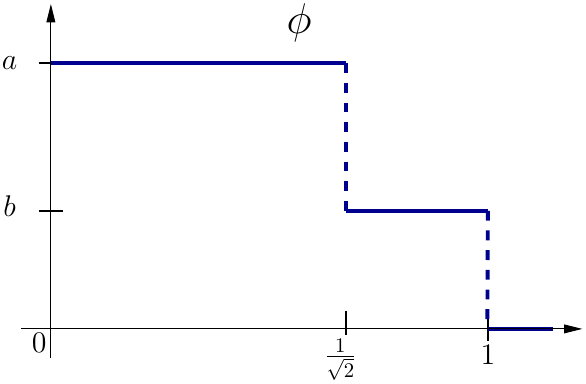}
  \caption{{\small The worst case scenario for the decaying of the norm of the vector ${\Delta}$: the formation is connected but there are two large groups with extreme values.}}
  \label{fig:worst_case_symmetricNN}
\end{figure}

\section{Self-alignment dynamics with discrete time steps}\label{sec:discrete_dynamics}
\setcounter{equation}{0}    
\setcounter{figure}{0}

Models for opinion  dynamics were originally introduced  as a discrete algorithms. In this section  we therefore extend our results on the semi-discrete continuous opinion dynamics (\ref{eq:opinion_formationb}) to the fully discrete case,
\begin{equation}
  \label{eq:consensus_Euler}
  \frac{{\bf x}_i(t+\Delta t) - {\bf x}_i(t)}{\Delta t} = \alpha \frac{\sum_j \phi_{ij}({\bf x}_j(t)-{\bf x}_i(t))}{\sum_j\phi_{ij}}.
\end{equation}
In particular, for $\alpha=1/\Delta t$ we find that  ${\bf x}_i^n= {\bf x}_i(n\Delta t)$ satisfies the Krause model \cite{blondel_convergence_2005,blondel_krauses_2009,krause_discrete_2000}
\begin{equation}
  \label{eq:consensus_discrete}
  {\bf x}_i^{n+1} = \frac{\sum_j \phi_{ij}{\bf x}_j^n}{\sum_j\phi_{ij}},\qquad \phi_{ij}
  = \phi(|{\bf x}_j^n-{\bf x}_i^n|).
\end{equation}
In the following, we study the properties of the discrete dynamics (\ref{eq:consensus_discrete}).

\subsection{Consensus with global interactions}

Many results of the continuous dynamics (\ref{eq:opinion_formationb}) can be translated to the discrete dynamics (\ref{eq:consensus_discrete}). For example, the convex hull of the opinions $\Omega$ (\ref{eq:Omega}) is still decreasing in time:
\begin{displaymath}
  \Omega(n+1) \subset \Omega(n).
\end{displaymath}

The discrete dynamics (\ref{eq:consensus_discrete}) will also converge to a consensus if initially all agents interact with each other. More precisely, arguing along the lines  of proposition  \ref{prop:CV_opinions} gives the following result.
\begin{theorem}
  \label{thm:expo_cv_discreet}
  Assume that $m=\min_{r\in[0,\dm{\bx(0)}]} \phi(r)>0$. Then, the diameter of the discrete dynamics (\ref{eq:consensus_discrete}) satisfies
    \begin{equation}
    \label{eq:d_cv_expo}
    \dm{\bx^n} \,\leq\, (1-m)^n\dm{\bx^0} \;\; \stackrel{n \rightarrow \infty}{\longrightarrow} \;\; 0.
  \end{equation}
and convergence to a consensus,  ${\bf x}_i^n \stackrel{n \rightarrow \infty}{\longrightarrow} \con{\bx} \in \Omega(0)$ follows.
\end{theorem}
\begin{proof}
  Using the contraction estimate  (\ref{eq:dpcontract}) followed by the bound $\ceta{_A} \geq \max_\theta \theta\cdot\lambda(\theta)$   yield
\[
    \dm{{\bx}^{n+1}}\leq (1-\ceta{_A})\dm{\bx^n} \leq (1- \theta\cdot \lambda(\theta,t^n))\dm{\bx^n}, \qquad  a_{ij} = \frac{\phi_{ij}}{\sum_\ell\phi_{\ell j}}.
\]
  Fix  $\theta={m}/{N}$ then $\Lambda(\theta)$ includes all agents, $\lambda(\theta,t^n)=N$, and we conclude
  \begin{displaymath}
    \dm{\bx^{n+1}} \leq \big(1-m\big)\, \dm{\bx^n},
  \end{displaymath}
  which proves (\ref{eq:d_cv_expo}).
\end{proof}

\subsection{Clustering with local interactions}

As in the continuous dynamics, we would like to investigate the behavior of the discrete dynamics (\ref{eq:consensus_discrete}) with  \emph{local} interactions; in particular, we are interested in the formation of clusters. Our aim is to reproduce the discrete analog of proposition \ref{ppo:limit_set_ct}.

\begin{proposition}\label{ppo:limit_set_discreet}
Let $\P^n=\{\bx^n_k\}_k$ be the solution of the discrete opinion dynamics  \eqref{eq:consensus_discrete} with compactly supported  influence function $\Supp=[0,\D)$.
Assume that it approaches  a steady state fast enough so that
\begin{equation}\label{eq:fast_end}
\sum_{n=m}^\infty \sum_i|\bx^{n+1}_i-\bx^n_i| \stackrel{m \rightarrow \infty}{\longrightarrow} 0.
\end{equation}
Then $\{\bx^n\}$ approaches a stationary state, $\con{\bx}$,
which  is partitioned into clusters, $\{{\mathcal C}_k\}_k$,  such that $\{1,2,\ldots,N\}=\cup_{k=1}^K \C_k$ and
  \begin{equation} 
    \label{eq:pq_cluster_d}
   \bx^n_i \longrightarrow \con{\bx}_{\C_k}, \qquad \text{for all} \ \ i\in \C_k.
  \end{equation}
\end{proposition}

\begin{proof}
By assumption (\ref{eq:fast_end})
\[
|\bx_i^{n_2}-\bx_i^{n_1}| \leq \sum_{n=n_1}^{n_2-1}|\bx_i^{n+1}-\bx_i^n| \ll 1, \ \ \text{for} \ \ n_2>n_1 \gg 1,
\] 
and hence $\bx^n$ approach a limit, $\bx_i^n \stackrel{n \rightarrow \infty}{\longrightarrow}  \con{\bx}_i$.  
The discrete dynamics (\ref{eq:consensus_discrete}) can be written in the following form:
\[
    \sum_j \phi_{ij} \left(\bx_i^{n+1}-{\bx}_i^n\right) = \sum_{j} \phi_{ij} ({\bx}_j^n-{\bx}_i^n).
\]
  Taking the scalar product  against ${\bx}_i^n$,  summing in $i$ and using the symmetry of $\phi_{ij}$ yields
\[
    \sum_{ij} \phi_{ij} \left\langle\left(\bx_i^{n+1}-{\bx}_i^n\right)\!, {\bx}_i^n\right\rangle = \sum_{ij} \phi_{ij} \langle{\bx}_j^n-{\bx}_i^n\,,\,{\bx}_i^n\rangle = -\frac{1}{2}\sum_{ij} \phi_{ij} |{\bx}_j^n-{\bx}_i^n|^2.
\]
    Since $\phi_{ij}$,  ${\bx}_i^n$ and by assumption, the tail $\sum_m^\infty |\bx_i^{n+1}-{\bx}_i^n|$ are bounded, we conclude that the sum on the right converges to zero
\[
\phi_{ij} |{\bx}_j^n-{\bx}_i^n|^2 \stackrel{n \rightarrow \infty}{\longrightarrow} \phi(\con{\bx}_j-\con{\bx}_i|) |\con{\bx}_j-\con{\bx}_i|^2 =0.
\]
Hence, either $\con{\bx}_j$ and $\con{\bx}_i$ are in separate clusters,  $|\con{\bx}_j -\con{\bx}_i|> \D$ or else, they are in the limiting point of the same cluster, say $i,j \in C_\ell$ so that $\con{\bx}_j =\con{\bx}_i$. 
\end{proof}

We now turn our attention to the convergence toward consensus for the discrete dynamics (\ref{eq:consensus_discrete}). As for the continuous dynamics (\ref{eqs:opinion}), there exists a Lyapunov functional energy for the dynamics under the additional assumption that the influence function $\phi$ is non-increasing. Consequently, we deduce the analog of theorem \ref{thm:con_non-symmetric} for the discrete dynamics.
\begin{theorem} \label{thm:E_decay_discrete}
  Let $\P^n=\{\bx^n_k\}_k$ be the solution of the discrete opinion dynamics  \eqref{eq:consensus_discrete} with non-increasing, compactly supported influence function $\Supp=[0,\D)$. If $\P^n$ remains uniformly connected for any $n$, then $\P^n$ converges to a consensus.
\end{theorem}
\begin{proof} First, we prove that the energy functional ${\mathcal E}^n$ is also a Lyapunov function for the discrete dynamics:
  \begin{equation}
    \label{eq:V_potential_discreet}
    {\mathcal E}^n := \sum_{ij} \Phi(|{\bf x}_j^n-{\bf x}_i^n|), \quad \Phi(r)=\int_0^r s\phi(s)ds.
  \end{equation}
  Introducing $\varphi(r^2)=\Phi(r)$, we have $\varphi(r) = \int_0^{\sqrt{r}} s\phi(s)ds=\frac{1}{2}\int_0^r \phi(\sqrt{y})dy$. By assumption $\phi$ is non-increasing, thus $\varphi$ is concave-down. Therefore,
  \begin{eqnarray*}
    {\mathcal E}^{n+1} - {\mathcal E}^n &=& \sum_{ij} \varphi(|{\bx}_j^{n+1}-{\bx}_i^{n+1}|^2) \,-\, \varphi(|{\bf x}_j^n-{\bf x}_i^n|^2) \\
    &\leq& \frac{1}{2} \sum_{ij} \phi(|{\bx}_j^n-{\bx}_i^n|) \big( |{\bx}_j^{n+1}-{\bf x}_i^{n+1}|^2 \,-\,|{\bx}_j^n-{\bx}_i^n|^2  \big).
  \end{eqnarray*}
  Using $|{\bf a}|^2-|{\bf b}|^2=\langle{\bf a}-{\bf b}\,,\,{\bf a}+{\bf b}\rangle$, we deduce:
  \begin{eqnarray*}
    {\mathcal E}^{n+1} - {\mathcal E}^n &\leq& \frac{1}{2} \sum_{ij} \phi_{ij} \langle\Delta_t{\bf x}_j^n-\Delta_t{\bx}_i^n\,,\, {\bx}_j^{n+1}-{\bx}_i^{n+1} +
    {\bx}_j^n-{\bx}_i^n\rangle\\
    &=&  \sum_{ij} \phi_{ij} \langle\Delta_t{\bx}_j^n\,,\, {\bx}_j^{n+1}-{\bx}_i^{n+1} + {\bf x}_j^n-{\bx}_i^n\rangle,
  \end{eqnarray*}
  since $\phi_{ij}=\phi_{ji}$. Writing ${\bx}_j^{n+1} =  {\bx}_j^n + \Delta_t{\bx}_j^n$, we obtain:
  \begin{eqnarray*}
    {\mathcal E}^{n+1} - {\mathcal E}^n &\leq& \sum_{ij} \phi_{ij} \langle\Delta_t{\bx}_j^n\,,\, 2({\bx}_j^n-{\bx}_i^n) + \Delta_t{\bx}_i^n - \Delta_t{\bx}_j^n\rangle.
  \end{eqnarray*}
  Combining with the equality:
  \begin{equation}
    \label{eq:Krause_equality}
    \sum_j \phi_{ij} \Delta_t{\bf x}_i^n = \sum_j \phi_{ij} ({\bf x}_j^n-{\bf x}_i^n).
  \end{equation}
  we conclude
  \begin{displaymath}
    {\mathcal E}^{n+1} - {\mathcal E}^n \leq \sum_{ij} \phi_{ij} \langle\Delta_t{\bx}_i^n\,,\, -\Delta_t{\bx}_i^n - \Delta_t{\bx}_j^n\rangle
   = -\sum_{ij} \phi_{ij} |\Delta_t{\bx}_i^n|^2,
  \end{displaymath}
  where we use once again the symmetry of the coefficients $\phi_{ij}$. Thus, ${\mathcal E}^n$ is decaying.

  Now, we would like to combine the decay of ${\mathcal E}^n$ and the strong connectivity of $\P^n$. Noting $\sigma_i = \sum_j \phi_{ij}$, the equality (\ref{eq:Krause_equality}) yields:
  \begin{displaymath}
    \frac{1}{2} \sum_{i,j} \phi_{ij} |\bx_j^n-\bx_i^n|^2 = \sum_i \sigma_i \langle \bx_i^n,\Delta_t \bx_i^n \rangle \leq {N \max_i |\bx_i^0|} \sqrt{\sum_i \sigma_i |\Delta_t \bx_i^n|^2},
  \end{displaymath}
  Thus,
  \begin{displaymath}
    {\mathcal E}^{n+1} - {\mathcal E}^n \leq - C_0^2 \left(\sum_{i,j} \phi_{ij}|\bx_j^n-\bx_i^n|^2\right)^2,\qquad C_0 = \frac{1}{2N \max_i |\bx_i(0)|}. 
  \end{displaymath}
  Summing in $n$, we deduce that the sum $\sum_{i,j} \phi_{ij}|\bx_j^n-\bx_i^n|^2$ becomes arbitrarily small. To conclude, we proceed as in the proof of theorem \ref{thm:con_non-symmetric}.
\end{proof}

\subsection{Numerical simulations of discrete dynamics}

In this section we illustrate the difference between the continuous opinion model (\ref{eq:opinion_formationb}) and its discrete version (\ref{eq:consensus_discrete}). To this end, we run in parallel numerical simulations of the discrete and continuous model subject to the same initial conditions.

First, we run a simulation with an influence function $\phi=\chi_{[0,1]}$ (figure \ref{fig:discrete_ratio1}). Discrete and continuous dynamics are very similar, except that there are three  \emph{branches} in the continuous dynamics which are not present in the discrete dynamics. For this reason, at the end of the simulation, we count $4$ clusters in the discrete dynamics and only $3$ in the continuous version.

Next we use the influence function (\ref{eq:phi_step}) $\phi=a\chi_{[0,1/\sqrt{2}]} + b\chi_{[1/\sqrt{2},1]}$ with $b/a=10$.
Here, the discrete and continuous dynamics give very different results shown in figure \ref{fig:discrete_ratio10}. As we have seen previously, the continuous dynamics converges to a distribution with uniformly spaced clusters and then reach a consensus. In contrast, the discrete dynamics does not stabilize. Order between the opinions $\{{\bf x}_i\}_i$ is no longer preserved, trajectories do cross. Even though the total number of clusters has been diminished with $b/a=10$ (from $4$ to $3$ clusters), the effect of the ratio $b/a$ on the clustering formation is less pronounced in the discrete dynamics.

\begin{figure}[ht]
  \centering
  \includegraphics[scale=.38]{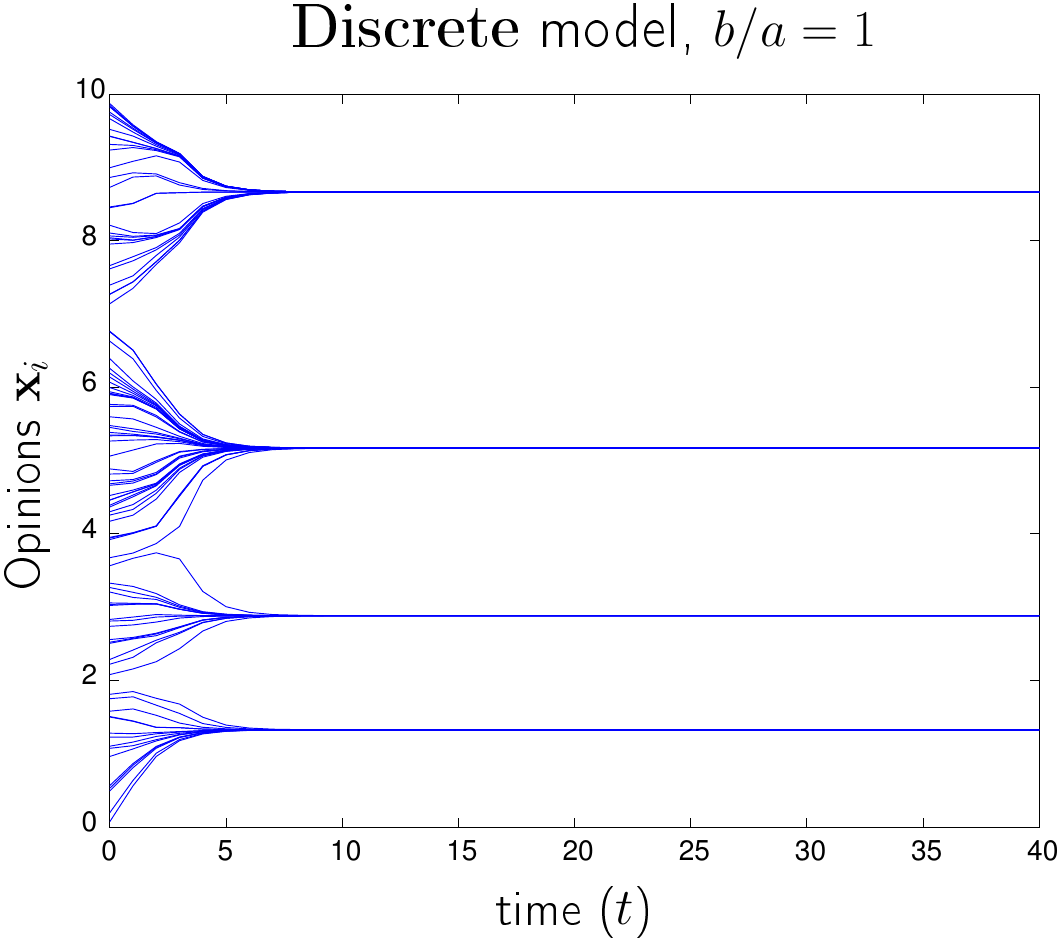}
  \includegraphics[scale=.38]{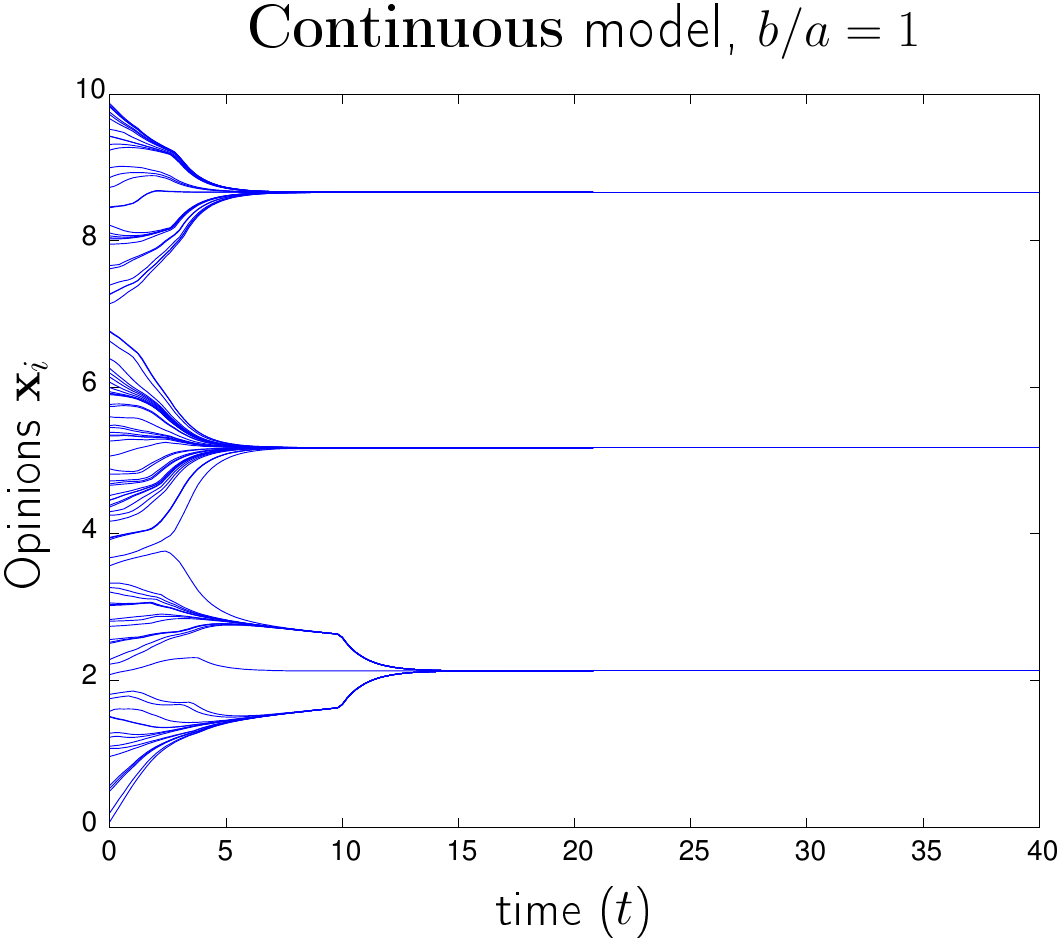}
  \caption{{\small Simulations of the discrete ({\bf Left} figure) and continuous dynamics ({\bf Right} figure) with $\phi=\chi_{[0,1]}$  starting with the same initial condition. Although the two simulations are very similar, the discrete dynamics yields $4$ clusters whereas the continuous dynamics gives $3$. We use a time discretization of $\Delta t=.05$ to simulate the continuous dynamics.}}
  \label{fig:discrete_ratio1}
\end{figure}

\begin{figure}[ht]
  \centering
  \includegraphics[scale=.38]{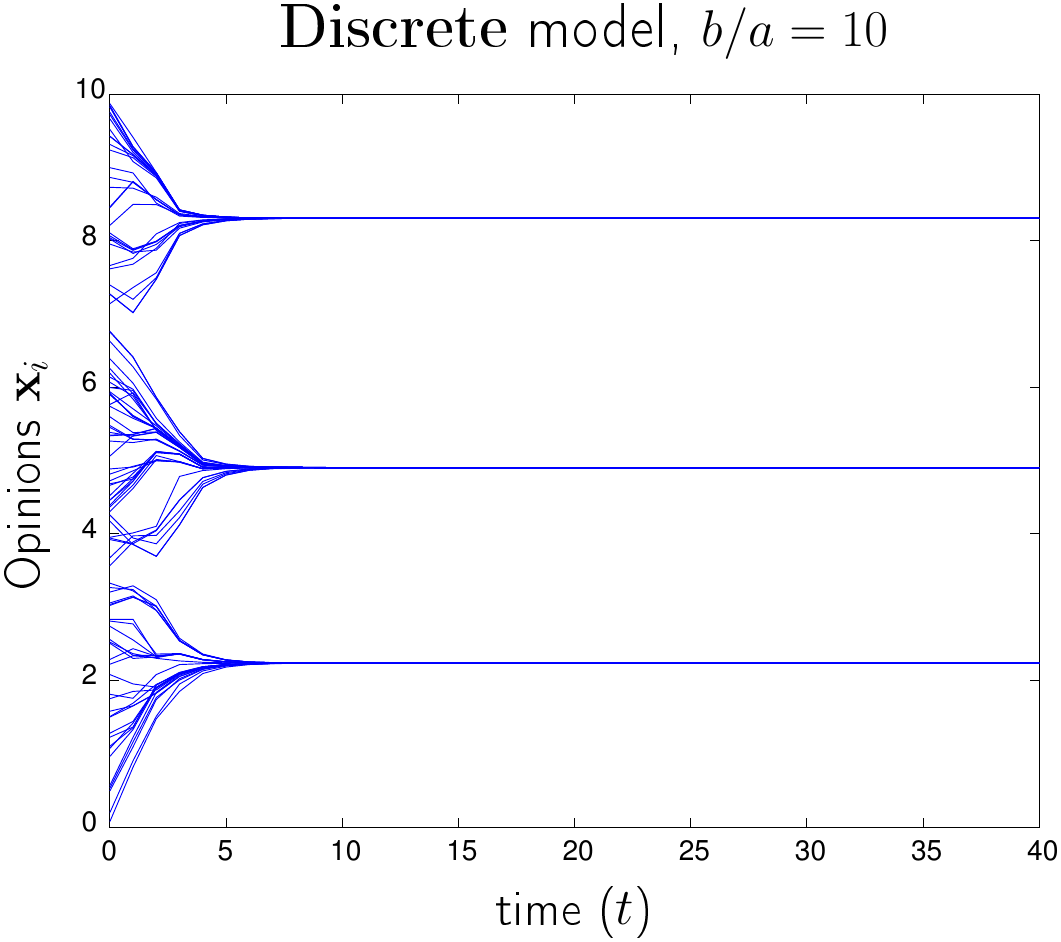}
  \includegraphics[scale=.38]{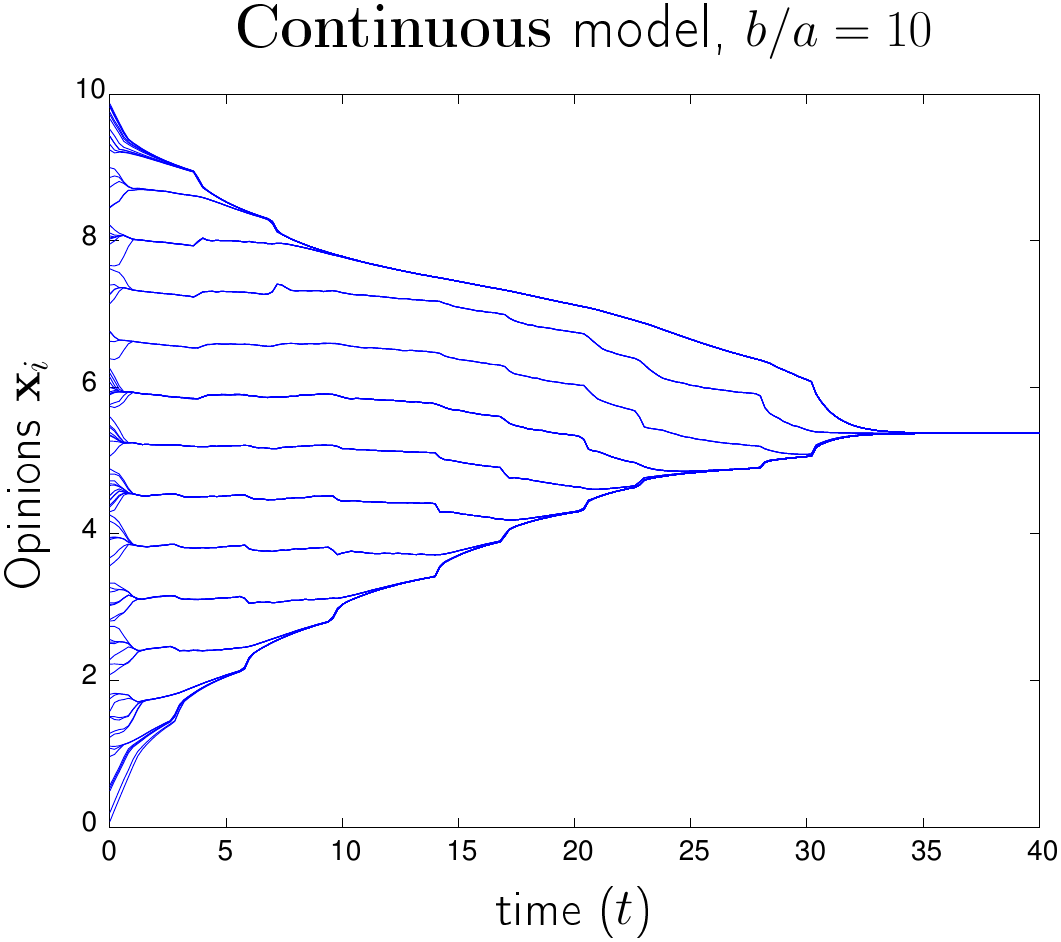}
  \caption{{\small Simulations of the discrete ({\bf Left} figure) and continuous dynamics ({\bf Right} figure) with $\phi=.1 \chi_{[0,1/\sqrt{2}]}+ \chi_{[1/\sqrt{2},1]}$  starting with the same initial condition. In contrast with figure \ref{fig:discrete_ratio1}, the two models produce very different output. There is no uniformly spaced formation in the discrete model, we only observe cluster formation.}}
  \label{fig:discrete_ratio10}
\end{figure}

\section{Mean-field limits: self-organized hydrodynamics}\label{sec:hydro}
\setcounter{equation}{0}
\setcounter{figure}{0}

When the number of agents $N$ is large, it is convenient to describe the evolution of the resulting large dynamical systems as mean-field equation. We limit ourselves to a few classic general  references on this topic \cite{CIP_1994,golse_2003,spohn_1991}, and a few recent references in the  context of opinion hydrodynamics \cite{canuto_krause_2012,toscani_2006}, and in flocking hydrodynamics
\cite{carrillo_asymptotic_2010a,carrillo_models_2010b,degond_2008,ha_particle_2008,KMT_2012,MOA_2010,motsch_tadmor_new_2011}.

\subsection{Opinion hydrodynamics} 
To derive the mean-field limit of the opinion dynamics model (\ref{eq:opinion_formationb}), we introduce the so-called empirical distribution $\rho(t,\bx)$:
\begin{displaymath}
  \rho(t,\bx): = \frac{1}{N} \sum_{j=1}^N \delta_{\bx_j(t)}(\bx),
\end{displaymath}
where $\delta$ is a Dirac mass and $\{\bx_j(t)\}_j$ is the solution of the consensus model (\ref{eq:opinion_formationb}). Expressed in terms of this empirical distribution,  the non-symmetric  model (\ref{eq:opinion_formationb}) (with $\alpha=1$) reads,
\begin{equation}
  \label{eq:consensus_empiricalRho}
  \dot{\bx}_i = \frac{\int_\by \phi(|\by-\bx_i|)(\by-\bx_i)\rho(t,\by)\,d\by}{\int_\by\phi(|\by-\bx_i|)\rho(t,\by)\,d\by} = \frac{(\phi(|\by|)\by * \rho) (\bx_i)}{(\phi(|\by|)* \rho)(\bx_i)}.
\end{equation}
This equation describes the characteristics of the density $\rho$. Indeed, integrating $\rho$ against a test function $\varphi$ yields\footnote{$(\cdot,\cdot)$ denotes the duality bracket between distributions and test functions}
\begin{displaymath}
  \frac{d}{dt} \big( \rho,\varphi\big) = \frac{d}{dt} \left( \frac{1}{N} \sum_{j=1}^N\varphi(\bx_j(t))\right)= \frac{1}{N} \sum_j^N \langle \dot{\bx}_j(t), \nabla_\bx\varphi(\bx_j(t))\rangle.
\end{displaymath}
Using the expression (\ref{eq:consensus_empiricalRho}), we deduce:
\begin{eqnarray*}
  \frac{d}{dt} (\rho,\varphi) &=& \frac{1}{N} \sum_{j=1}^N  \left\langle \frac{\phi(|\by|)\by * \rho (\bx_j)}{\phi(|\by|)* \rho (\bx_j)}, \nabla_\bx\varphi(\bx_j(t))\right\rangle = \left(\rho,\left\langle \frac{\phi(|\by|)\by * \rho }{\phi(|\by|)* \rho},\nabla_\bx\varphi\right\rangle\right) \\
  &=& \left(-\nabla_\bx \cdot\left(\frac{\phi(|\by|)\by * \rho }{\phi(|\by|)* \rho} \; \rho\right)\!,\varphi\right).
\end{eqnarray*}
Thus, $\rho=\rho(t,\bx)$ satisfies a continuum transport equation, 
\begin{subequations}\label{eqs:pdf_consensus}
\begin{equation}
  \label{eq:pdf_consensusa}
  \qquad \quad \partial_t \rho +  \nabla_\bx\cdot \big(\rho \bu\big) = 0 \quad \text{ with} \quad \bu(\bx) =  \frac{\int_\by \phi(|\by\!-\!\bx|)(\by\!-\!\bx)\,\rho(\by)\,d\by }{\int_\by \phi(|\by\!-\!\bx|)\,\rho(\by)\,d\by}.
\end{equation}
This is the hydrodynamic description of the agent-based opinion model (\ref{eq:opinion_formationb}). 
Similarly, the opinion hydrodynamics of the corresponding \emph{symmetric} model (\ref{eq:opinion_formationa}) (with $\alpha=1$) amounts to the aggregation model \cite{BCL2009,canuto_krause_2012}
\begin{equation}
  \label{eq:pdf_consensusb}
\qquad \quad \partial_t \rho +  \nabla_\bx\cdot \big(\rho \bu\big) = 0 \quad \text{ with} \quad \bu(\bx) =  \nabla \Phi*\rho
\end{equation}
\end{subequations}
We note that the transport equations (\ref{eqs:pdf_consensus}) are  non-linear due to the dependence of the velocity field $\bu=\bu(\rho)$. Main features of the  particle description for opinion dynamics (\ref{eqs:opinion}) carry over the  hydrodynamic model (\ref{eqs:pdf_consensus}). Thus, for example, the symmetric  model (\ref{eq:pdf_consensusb})  preserve the center of mass,  $\frac{d}{dt} \left(\int_\bx \bx\rho(t,\bx)\,d\bx \right) = 0 $ where the non-symmetric model (\ref{eq:pdf_consensusa}) does not.     
We distinguish between the two cases of global and local interactions.\newline
The existence of regular solutions of the \emph{symmetric}  aggregation model (\ref{eq:pdf_consensusb})
for bounded decreasing $\phi$'s such that $|\phi'(r)r|\lesssim \phi(r)$ was proved in \cite{BCL2009}.
This holds independently whether $\phi$ is global or not. 
Moreover, if the kernel $\phi$ is globally supported, then one can argue along the lines of the underlying  agent-based model   (\ref{eq:consensus_empiricalRho}), to  prove  convergence of the hydrodynamics toward a consensus, that is,  $\rho(t,\bx)$ converges to a single point asymptotically in time. 
If $\phi$ is compactly supported, however,  then the velocity field $\bu$ need not be continuous with respect to $\rho$  due to the singularity when $\int_\by \phi(|\by\!-\!\bx|)\,\rho(\by)\,d\by=0$. Then, existence and uniqueness of solution of the non-symmetric model (\ref{eq:pdf_consensusa}) cannot be obtained through a standard Picard's iteration argument. 
The large time  behavior of the dynamics in this local setup is completely open. As in the agent-based dynamics, the generic solution  $\rho(t,\bx)$ is expected to concentrate in a finitely many clusters, or ``islands''; in particular, under appropriate assumption on the persistence of connectivity among these islands, one may expect a consensus. Preliminary simulations show that cluster formation tends to persist for the hydrodynamic model, but analytical justification remains open. 

\subsection{Flocking hydrodynamics}  
We study the second-order flocking models (\ref{eqs:flocking}) in terms of  the  empirical distribution 
$f^N(t,{\bx},{\bv}) := \frac{1}{N} \sum_{j=1}^N \delta_{{\bx}_j(t)}(\bx) \otimes \delta_{{\bv}_j(t)}(\bv)$,
where $\delta_{\bx}\otimes\delta_{\bv}$ is the usual Dirac mass on the phase space $\mathbb{R}^d \times \mathbb{R}^d$.
Consider the non-symmetric particle model system for flocking  (\ref{eq:MT_model}): expressed in terms of $f^N$, it reads
\[
    \frac{d{\bx}_i}{dt} = {\bv}_i, \quad 
    \frac{d{\bv}_i}{dt} =\alpha F[f^N](\bx_i,\bv_i), \qquad 
F[f](\bx,\bv): = \alpha\frac{\int_{{\by},{\bw}} \phi(|{\by}\!-\!{\bx}|)\,({\bw}\!-\!{\bv}_i) \, f({\by},{\bw})\,d{\by} d{\bw}}{\int_{\by} \phi(|{\by}\!-\!{\bx}|)\,f({\by},{\bw})\,d{\by}d{\bw}},
\]
which leads to Liouville's equation,
  \begin{equation}
    \label{eq:evo_f}
    \partial_t f + {\bv}\cdot \nabla_{\bx} f + \nabla_{{\bv}}\cdot(F[f]\,f) = 0.
  \end{equation}
Integrating the empirical distribution $f^N$ in the velocity variable ${\bv}$ yields the hydrodynamic description of flocking, expressed in terms of the density and momentum distributions of particles, 
\begin{eqnarray*}
  \rho(t,{\bx}) = \int_{\bv} f(t,{\bx},{\bv})\,d{\bv} \qquad &\left(\text{corresponding to} \  \  \frac{1}{N} \sum_{j=1}^N \delta_{{\bx}_j(t)}({\bx})\right), \\
\rho(t,{\bx}){\bu}(t,{\bx})  =\int_{\bv} {\bv} f(t,{\bx},{\bv})\,d{\bv}  \qquad &\left(\text{corresponding to} \ \  \frac{1}{N} \sum_{j=1}^N \bv_j(t) \delta_{{\bx}_j(t)}({\bx})\right)
.
\end{eqnarray*}
Integrating the kinetic equation (\ref{eq:evo_f}) against the first moments $(1,{\bv})$ yields the system, cf.,  
 \cite{ha_particle_2008,carrillo_models_2010b,motsch_tadmor_new_2011},

\begin{subequations}\label{eqs:macro}
  \begin{eqnarray}
    \label{eq:macro_rho_tp}
    && \partial_t \rho + \nabla_{\bx} \cdot (\rho{\bu}) =0 \\
    \label{eq:macro_u_tp}
    && \partial_t (\rho {\bu}) + \nabla_{\bx}\cdot(\rho{\bu} \otimes {\bu} + {\bf P}) = \alpha \rho(\overline{{\bu}}-{\bu}).
  \end{eqnarray}
  The expression on the right reflects alignment: the tendency of agents with velocity ${\bu}$ to relax towards the local average velocity, $\overline{{\bu}}({\bx})$, dictated by the normalized influence function  $a({\bx},{\by})$,
  \begin{equation}
    \label{eq:over_u}
    \overline{{\bu}}({\bx}) := \int_{\by} a({\bx},{\by}) \rho({\by}) {\bu}({\by})\,d{\by}, \quad \int_{{\by}} a({\bx},{\by})\rho({\by})\,d{\by} = 1.
  \end{equation}
\end{subequations}
This  includes in particular, the hydrodynamic description of the symmetric and non-symmetric  flocking models, given respectively by 
\[
  a({\bx},{\by}) = \left\{\begin{array}{ll}
\phi(|{\by}-{\bx}|) & \text{C-S model} \ (\ref{eq:CS_model}),\\ \\
\displaystyle \frac{\phi(|{\by}-{\bx}|)}{\int_{{\by}}\phi(|{\by}-{\bx}|)\rho({\by})\,d{\by}} & \text{non-symmetric model} \ (\ref{eq:MT_model}).
\end{array}\right.
\]

The system (\ref{eqs:macro}) is not closed since the equation for $\rho{\bf u}$ (\ref{eq:macro_u_tp}) does depend on the third moment of $f$ which is encoded in the pressure term ${\bf P}:= \int_{\bv} ({\bv}-{\bu})\otimes({\bv}-{\bu})f(t,{\bx},{\bv})\,d{\bv}$. If we neglect the pressure (in other words,  assume a monophase distribution, $ f(t,{\bx},{\bv}) = \rho(t,{\bx})\, \delta_{{\bu}(t,{\bx})}({\bv})$ so that ${\bf P}\equiv 0$), then the flocking hydrodynamics  (\ref{eqs:macro}) is reduced to the closed system 
\begin{equation}\label{eq:macro_rho_noncons}
\left\{
\begin{array}{ll}
     \partial_t \rho + \nabla_{\bx} \cdot (\rho{\bu}) =0, \\ \\
        \partial_t {\bu} + ({\bu}\cdot\nabla_{\bx}){\bu} = \alpha (\overline{{\bu}}-{\bu}).
\end{array}
\right.
\end{equation}
The question of an emerging flock in (\ref{eq:macro_rho_noncons}) follows along the lines of our discussion on the underlying agent-based models (\ref{eqs:flocking}). The case of a \emph{global} influence function is rather well-understood:
in particular, regularity of the one-dimensional ``incompressible'' case, $\rho\equiv 1$, depends on initial critical threshold  \cite{LT_2001,ST_1992}. Flocking hydrodynamics governed by \emph{locally} supported influence function requires a more intricate analysis, due to the realistic presence of vacuum, \cite{tadmor_2014}. 
The hydrodynamic description of self-organized dynamics give rise to  systems like  (\ref{eq:macro_rho_noncons}) which involve nonlocal means. Questions of regularity and quantitative behavior of such systems provide a rich source for future studies.  

\section{Further reading on self-organized dynamics}\label{sec:further} 
In this paper we discussed  fundamental aspects which arise  in the context of 
flocking and opinion dynamics, as prototype models for self-organized dynamics. Specifically, we focused here on the emerging large-time behavior of self-alignment and we highlight  a few open questions aiming to attract further mathematical studies in this direction.  The  much broader subject of self-organized dynamics  lies at the crossroads of several fields. A comprehensive review of the subject is beyond the scope of this paper, in particular, as it continues to attract an increasing amount of attention reported in a rapidly growing literature. Instead, we  refer the interested reader to a selection of references outlined below.
As with all multidisciplinary fields, the work on  self-organized dynamics can be classified into several different categories. We shall mention five of them.

\noindent
{\bf Different disciplines}. A natural classification is offered by the underlying topic. Many models of self-organized dynamics are driven by examples from biology: these include aggregation of bacteria and amoeba \cite{Ben_2003,eftimie_2012,ha_levy_2009,KL_1993,Sh_2011}, dynamics of insects \cite{BCE_1991,couzin_franks_2003}, school of fish \cite{Aoki_1982, hemelrijk_self-organized_2008,youseff_2008}  flocking of birds \cite{Ballerini_2008,cavagna_et_al_2008i,cavagna_et_al_2008ii,cavagna_et_al_2010,cucker_emergent_2007,cucker_mathematics_2007,ha_particle_2008,reynolds_flocks_1987,T-T_1998,vicsek_novel_1995}, and related models in ecology \cite{grimm_individual-based_2005}.
Self-organized dynamics found its in many other areas,  from pedestrian and traffic dynamics \cite{He_2001,PT_2011}, social networks and economics \cite{DY_2000,He_2010,jackson_2010,merton_1954,McPherson_2001}, complex networks \cite{BHT_2013,DLBGL_2010,OSM_2004} and  opinion dynamics \cite{Ben_2005,CFL_2009,DeG_1974,duering_boltzmann_2009,EK_2001,hegselmann_opinion_2002,krause_discrete_2000,toscani_2006,WDA2005,Wei2006},  all the way to applications in marketing \cite{Ax_1984,Ax_1997}, production networks \cite{Ring_2012},  robotics \cite{CMB_2006,JE_2007,ZEP_2011} and materials \cite{Sill2000,MD2012}, and with somewhat more esoteric examples such as gossiping \cite{BGPS_2006}, collective motion at heavy metal concerts \cite{SBSC_2013} and self-organized phases in the Tour De France \cite{Tren_2013}.

\noindent
{\bf Different models}. Together with the different contexts, come different models of  self-organized dynamics.  We mention a few of the more notable ones: Krause model for opinion dynamics \cite{krause_discrete_2000} and the follow-up works in \cite{blondel_krauses_2009,canuto_krause_2012,hegselmann_opinion_2002,KurRam2011,lorentz_2007}, Axelrod models for marketing \cite{Ax_1984} and the  influential models for  ''flocking'' (at various ``levels'') of Aoki, Reynolds and Couzin \cite{Aoki_1982,couzin_franks_2003,couzin_2002,LX_2010,reynolds_flocks_1987,youseff_2008}, Vicsek et. al, \cite{vicsek_novel_1995} and the follow-up works in \cite{degond_2008,degond_2011,JLM_2003}, Cucker-Smale model \cite{cucker_emergent_2007,cucker_mathematics_2007} and  related works in \cite{Bialek_2013,carrillo_asymptotic_2010a,ha_simple_2009,ha_particle_2008,haskovec_2013,KMT_2012,motsch_tadmor_new_2011,Shen_2007}, and the StarFlag project \cite{Ballerini_2008,cavagna_et_al_2008i,cavagna_et_al_2008ii,cavagna_et_al_2010}.  

\noindent
{\bf Different scales}. Different models of self-organized dynamics are realized at different scales. As examples for agent-based models (also known as  Individual-Based Models (IBM)) we mention \cite{Ben_2005,couzin_franks_2003,grimm_individual-based_2005,krause_discrete_2000,Li_2008,OSFM_2007,reynolds_flocks_1987}.  Their mean-field limit  leads to a kinetic description
\cite{carrillo_milling_2009,carrillo_asymptotic_2010a,duering_boltzmann_2009,ha_particle_2008,toscani_2006} and macroscopic averaging then leads to hydrodynamic-scale description as in
\cite{BHW_2013,carrillo_models_2010b,degond_2008,degond_2011,eftimie_2012,KMT_2012,KL_1993,LRC_2000,MOA_2010,MD2012,tadmor_2014}.

\noindent
{\bf Different approaches}. In this paper, we focused our attention on mathematical aspects which explain the large time behavior of self-alignment models. The study of general models for self-organized dynamics includes several different approaches. Classified by the tools of the trade, we mention statistical mechanics \cite{Bialek_2013,CFL_2009,spohn_1991},  clustering and spectral theory of graphs \cite{BHW_2013,comaniciu_meir_2002,CMB_2006,JLM_2003,OSFM_2007}, optimization and control \cite{CS_2011,DLBG_2009,DLBGL_2010,JE_2007,JK_2010,OSM_2004,ZEP_2011}, game theory \cite{BHT_2013,Helbing_2009}, jump processes, \emph{nonlinear} Markov chains and  stochastic analysis  \cite{BHW_2013,FZ_2008,He_2010,vicsek_novel_1995}.

\noindent
{\bf Different patterns}. One of the most intriguing features of self-organized dynamics is the formation of different patterns.
In this paper, we limited ourselves to the simple pattern of ``consensus'' (or a ``flock'') but the format is much richer. We mention the example of swarming and mill-like vortices \cite{canizo_collective_2009,carrillo_milling_2009,carrillo_models_2010b,DCBC_2006,EK_2001,LRC_2000,Li_2008,LX_2010,Sh_2011,T-B_2004}, phase transition \cite{FL_2012,vicsek_novel_1995}, aggregation \cite{BHW_2013}, biotic colonies \cite{Ben_2003,KL_1993}, lattices \cite{OS_2006}, leaders \cite{couzin_effective_2005,Shen_2007}, shocks \cite{BRSW_2014,tadmor_2014} and related issues which arise in the context of control and stability {\cite{blondel_convergence_2005,DCBC_2006,JK_2010,Li_2008}.

\medskip\noindent 
Finally, we recommend on several reviews on self-organization \cite{camazine_self-organization_2001,CFL_2009,eftimie_2012,He_2010,Wei2006}, and  in particular,  the most recent comprehensive review of Vicsek and Zefeiris \cite{vicsek_2012}.



\begin{thebibliography}{10}

\bibitem{Aoki_1982}
I. Aoki,
\newblock A simulation study on the schooling mechanism in fish
\newblock {\em Bull. Japanese Society of Scientific Fisheries}, 48(8), 1081-1088, 1982.

\bibitem{Ax_1984}
R. Axelrod 
\newblock The Evolution of Cooperation,
\newblock  New York: Basic Books.

\bibitem{Ax_1997}
R. Axelrod
\newblock The Complexity of Cooperation: Agent-based models of competition and collaboration,
\newblock {\em Princeton University Press}, Princeton, NJ.    

\bibitem{Ballerini_2008}
M. Ballerini, N. Cabibbo, R. Candelier, A. Cavagna, E. Cisbani, I. Giardina, V. Lecomte,
A. Orlandi, G. Parisi, A. Procaccini, M. Viale and V. Zdravkovic,
\newblock Interaction ruling animal collective behavior depends on topological rather than metric distance
\newblock {\em PNAS},   105(4), 1232-1237, 2008.

\bibitem{BHT_2013}
N. Bellomo, M. Herrero and  A. Tosin
\newblock  On the dynamics of social conflicts: looking for the black swan,
\newblock {\em  Kinetic And Related Models} 6:459--479, 2013.

\bibitem{Ben_2003}
E. Ben-Jacob
\newblock Bacterial self-organization: co-enhancement of complexification and adaptability in a dynamic environment.
\newblock {\em  Phil. Trans. R. Soc. Lond. A.}, 361(1807):1283--1312, 2003.

\bibitem{Ben_2005}
E. Ben-Naim,
\newblock Opinion dynamics: rise and fall of political parties
\newblock {\em Europhys. Lett.}, 69(5):671--677,  2005. 

\bibitem{BCL2009}
A. Bertozzi, J. Carrillo and T. Laurent,
\newblock Blow-up in multidimensional aggregation equations with mildly singular interaction kernels
\newblock Nonlinearity 22 (2009) 683-710.

\bibitem{BRSW_2014}
A. Bertozzi, J. Rosado, M. Short and L. Wang,
\newblock Contagion shocks in one dimension
\newblock preprint

\bibitem{Bialek_2013}
W. Bialek, A. Cavagna, I. Giardina, T. Mora, O. Pohl, E. Silvestri, M. Viale, A. Walczak,
\newblock Social interactions dominate speed control in driving natural flocks toward criticality
\newblock ArXiv:1307.5563v1.

\bibitem{blondel_convergence_2005}
V.~Blondel, J.~M Hendricks, A.~Olshevsky, and J.~Tsitsiklis,
\newblock Convergence in multiagent coordination, consensus, and flocking.
\newblock In {\em {IEEE} Conference on Decision and Control}, volume~44, page
  2996, 2005.

\bibitem{blondel_krauses_2009}
V.~D. Blondel, J.~M. Hendricks, and J.~N. Tsitsiklis,
\newblock On Krause's multi-agent consensus model with state-dependent connectivity.
\newblock {\em Automatic Control, {IEEE} Transactions on}, 54(11):2586–2597,
  2009.

\bibitem{BGPS_2006}
S. Boyd, A. Ghosh, B. Prabhakar, and D.Shah,
\newblock Randomized gossip algorithms.
\newblock {\em IEEE Trans. Inform. Theory}, 52;2508-2530, 2006 

\bibitem{BCE_1991}
A. M. Bruckstein, N. Cohen, A. Efrat,
\newblock Ants, crickets and frogs in cyclic Pursuit
\newblock CIS report \#9105, Center for Intelligent Systems, Technion Israel Inst. of Tech., 1991.

\bibitem{BHW_2013}
M. Burger, J. Haskovec, M.-T. Wolfram,
\newblock Individual based and mean-field modelling of direct aggregation,
\newblock {\em Phys. D},  260: 145--158, 2013.

\bibitem{camazine_self-organization_2001}
S.~Camazine, J.~L. Deneubourg, N.~R Franks, J.~Sneyd, G.~Theraulaz, and
  E.~Bonabeau,
\newblock Self-organization in biological systems.
\newblock {\em Princeton University Press; Princeton, {NJ:} 2001}, 2001.

\bibitem{CS_2011}
G. de Campos and A. Seuret
\newblock Improved Consensus Algorithms using Memory Effects,
\newblock {\em Decision and Control} 2011 50th IEEE Conf. (CDC-ECC), IEEE, 982--987, 2011

\bibitem{canizo_collective_2009}
J.~A. Canizo, J.~A. Carrillo and J.~Rosado,
\newblock Collective behavior of animals: swarming and complex patterns.
\newblock 2009.

\bibitem{canuto_krause_2012}
C. Canuto, F. Fagnani, and P. Tilli,
\newblock An Eulerian approach to the analysis of Krause's consensus models
\newblock {\em SIAM J. Control Optim.}, 50(1), 243-265.  2012. 


\bibitem{CFSZ_2008}
R. Carli, F. Fagnani, A. Speranzon, and S. Zampieri,
\newblock Communication constraints in the average consensus problem.
\newblock {\em Automatica}, 44, 671-684, 2008.

\bibitem{carrillo_milling_2009}
J. ~A. Carrillo, M. D'Orsogna, V. Panferov,
\newblock Double milling in self-propelled swarms from kinetic theory.
\newblock {\em Kinet Relat Models},  2:363-378, 2009.

\bibitem{carrillo_asymptotic_2010a}
J.~A. Carrillo, M.~Fornasier, J.~Rosado, and G.~Toscani,
\newblock Asymptotic flocking dynamics for the kinetic {Cucker-Smale} model.
\newblock {\em {SIAM} J. Math. Anal.}, 42:218--236, 2010.

\bibitem{carrillo_models_2010b}
J.~A. Carrillo, M. Fornasier, G. Toscani, and F. Vecil,
\newblock Particle, kinetic, and hydrodynamic models of swarming.
\newblock in Naldi, G., Pareschi, L., Toscani, G. (eds.) Mathematical Modeling of Collective Behavior in Socio-Economic and Life Sciences, Series: Modelling and Simulation in Science and Technology, Birkhauser, (2010), 297-336. 


\bibitem{CFL_2009}
C. Castellano, S. Fortunato, V. Loreto,
\newblock Statistical physics of social dynamics,
\newblock {\em Rev. Modern Phys}, 81:591--646, 2009.

\bibitem{cavagna_et_al_2008i}
A. Cavagna,  I. Giardina, A. Orlandi, G. Parisi, A. Procac-cini, M. Viale and V. Zdravkovic,
\newblock  The starflag handbook on collective animal behaviour. 1: Empirical methods.
\newblock {\em Animal Behaviour} 76:217-236, 2008.

\bibitem{cavagna_et_al_2008ii}
A. Cavagna, I. Giardina, A. Orlandi, G. Parisi and A. Procaccini,
\newblock  The starflag handbook on collective animal behaviour. 2: Three-dimensional analysis.
\newblock {\em  Animal Behaviour} 76:237-248, 2008.

\bibitem{cavagna_et_al_2010}
A. Cavagna, A. Cimarelli, I. Giardina, G. Parisi, R. Santagati, F. Stefanini and M. Viale,
\newblock  Scale-free correlations in starling flocks. 
\newblock {\em Proc. Nat. Academy  Sci.,  U.S.A.} 107:11865-11870, 2010.

\bibitem{CIP_1994}
C.  Cercignani, R. Illner, M. Pulvirenti,
\newblock The mathematical theory of dilute gases.
\newblock Springer series in Applied Mathematical Sciences, 106, Springer-Verlag, 1994. 

\bibitem{chung_1997}
Fan R. K. Chung,
\newblock Spectral Graph Theory
\newblock Amer. Math. Soc. CBMS Regional Conference Series in Mathematics, No. 92), 1997

\bibitem{CDZ_1993} J.E. Cohen, Y. Derriennic and Gh. Zbaganu,
\newblock Majorization, Monotonicity of Relative Entropy, and Stochastic Matrices
\newblock {\em Contemp. Mathematics} 149, 251--259, 1993.

\bibitem{CIRRSZ_1993} J. E. Cohen, Y. Iwasa, Gh. Rautu, M. B. Ruskai, E. Seneta and 
Gh. Zbaganu,
\newblock Relative entropy under mappings by stochastic matrices
\newblock {\em Linear Algebra  Applications}, 179(15) 211--235, 1993.

\bibitem{comaniciu_meir_2002}
D. Comaniciu and P. Meer,
\newblock Mean shift: A robust approach toward feature space analysis.
\newblock  IEEE Transactions on Pattern Analysis and Machine Intelligence archive
Volume 24(5) Pages 603 - 619, 2002.

\bibitem{CMB_2006}
J. Cort\'{e}s, S. Martinez, and F. Bullo,
\newblock Robust rendezvous for mobile autonomous agents via proximity graphs in arbitrary dimensions,
\newblock {\em IEEE Trans. Automat. Control}, 51, 1289-1298, 2006.

\bibitem{couzin_franks_2003}
I. Couzin and N. Franks
\newblock Self-organized lane formation and optimized traffic
flow in army ants
\newblock {\em Proc. R. Soc. Lond. B}, 270:139--146, 2003.

\bibitem{couzin_effective_2005}
{I.D.} Couzin, J.~Krause, {N.R.} Franks, and {S.A.} Levin,
\newblock Effective leadership and decision-making in animal groups on the move.
\newblock {\em Nature}, 433(7025):513-516, 2005.

\bibitem{couzin_2002}
I. D. Couzin, J. Krause, R. James, G. D. Ruxton and N. R. Franks,
\newblock Collective memory and spatial sorting in animal groups
\newblock {\em J. of Theoretical Biology}, 218(1), 1-11, 2002.

\bibitem{cucker_emergent_2007}
F.~Cucker and S.~Smale.
\newblock Emergent behavior in flocks.
\newblock {\em {IEEE} Transactions on automatic control}, 52(5):852, 2007.

\bibitem{cucker_mathematics_2007}
F.~Cucker and S.~Smale,
\newblock On the mathematics of emergence.
\newblock {\em Japanese Journal of Mathematics}, 2(1):197-227, 2007.

\bibitem{CSZ_2006}
F. Cucker, S. Smale and D.X. Zhou,
\newblock Modeling language evolution.
\newblock {\em Found. Comput. Math.}, 4, 315-343, 2006.

\bibitem{degond_2008}
P.~ Degond and S.~ Motsch,
\newblock Continuum limit of self-driven particles with orientation interaction
\newblock {\em Math. Models Methods Appl. Sci.}, 18(1):1193-1215, 2008.

\bibitem{degond_2011}
P.~ Degond and S.~ Motsch,
\newblock A macroscopic model for a system of swarming agents using curvature control
\newblock {\em J. Stat. Physics}, 143(4):685--714,2011.

\bibitem{DeG_1974}
M. H. DeGroot,
\newblock Reaching a consensus
\newblock {\em J. Amer. Stat. Association}, 69(345): 118--121, 1974.

\bibitem{DLBG_2009}
P. DeLellis, M. diBernardo and F. Garofalo
\newblock Novel decentralized adaptive strategies for the synchronization of complex networks
\newblock {\em Automatica} 45(5):1312--1318, 2009.

\bibitem{DLBGL_2010}
P. DeLellis, M. diBernardo, F. Garofalo, D. Liuzza
\newblock Analysis and stability of consensus in networked control systems
\newblock{\em  Applied Mathematics and Computation} 217(3):988--1000, 2010.

\bibitem{Demmel_1996} J. Demmel.
\newblock Applications of Parallel Computers.\newline
\newblock Lecture notes {\sf http://www.cs.berkeley.edu/$\sim$\,demmel/cs267/lecture20/lecture20.html}.

\bibitem{Dob_1956} R.L. Dobrushin
Central limit theorem for nonstationary Markov chains.I
Theory Probab. Appl., 1(1), 65-80, 1956.

\bibitem{DCBC_2006}
D'Orsogna, M. R., Chuang, Y.-L., Bertozzi, A. L. and Chayes, L.,
\newblock Self-propelled particles with soft-core interactions. patterns, stability, and collapse.
\newblock {\em  Phys. Rev. Lett.} {96}:104--302, 2006.

\bibitem{DY_2000}
A. A. Dragulescu and V. M. Yakovenko
\newblock Statistical mechanics of money
\newblock {\em  The European Physical Journal B}, 17:723--729, 2000.
  
\bibitem{duering_boltzmann_2009}
B.~Duering, P.~Markowich, {J.F.} Pietschmann, and {M.T.} Wolfram,
\newblock Boltzmann and {Fokker-Planck} equations modelling opinion formation
  in the presence of strong leaders.
\newblock {\em Proceedings of the Royal Society A: Mathematical, Physical and
  Engineering Science}, 465(2112):3687, 2009.

\bibitem{EK_2001}
L. Edelstein-Keshet, 
\newblock Mathematical models of swarming and social aggregation, 
\newblock {\em International Symposium on Nonlinear Theory and its Applications}, (NOLTA 2001) Miyagi, Japan, 2001.

\bibitem{eftimie_2012}
R. Eftimie,
\newblock Hyperbolic and kinetic models for self-organized biological aggregations and movement: a brief review.
\newblock {\em J Math Biol.} 65(1):35-75, 2012.

\bibitem{FZ_2008}
F. Fagnani and S. Zampieri, 
\newblock Randomized consensus algorithms over large scale networks.
\newblock {\em IEEE J. Selected Areas of Communications} , 26,  634-649, 2008.

\bibitem{fiedler1973} 
M. Fiedler, 
\newblock Algebraic connectivity of graphs, 
\newblock {\em Czech. Math. J.} 23(98), 1973, pp. 298--305.

\bibitem{fiedler1989}
M. Fiedler,
\newblock Laplacian of graphs and algebraic connectivity.
\newblock {\em Combinatorics and Graph Theory}  25, 57--70, 1989.

\bibitem{FL_2012}
A. Frouvelle and J.-G. Liu
\newblock Dynamics in a kinetic model of oriented particles with phase transition,
\newblock {\em  SIAM J. Math Anal.},  44:791--826, 2012.

\bibitem{GR_2001}  C. Godsil and G. Royle, 
\newblock Algebraic Graph Theory
\newblock Graduate Texts in Mathematics ser. vol. 207. .Springer-Verlag, 2001.
 
\bibitem{golse_2003}
F.~ Golse,
\newblock The mean-field limit for the dynamics of large particle systems.
\newblock {\em Journ\'{e}es \'{E}quations aux d\'{e}riv\'{e}s partielles}, 9:1--47, 2003.

\bibitem{grimm_individual-based_2005}
V.~Grimm and S.~F Railsback,
\newblock {\em Individual-based modeling and ecology}.
\newblock Princeton Univ Pr, 2005.

\bibitem{ha_levy_2009}
S.~Y Ha and D. Levy
\newblock Particle, kinetic and fluid models for phototaxis,
\newblock {\em Discrete Cont. Dynamical Systems Ser. B.}, 12(1):77-108, 2009 

\bibitem{ha_simple_2009}
S.~Y Ha and J.~G Liu,
\newblock A simple proof of the {Cucker-Smale} flocking dynamics and mean-field
  limit.
\newblock {\em Communications in Mathematical Sciences}, 7(2):297--325, 2009.

\bibitem{ha_particle_2008}
S.~Y. Ha and E.~Tadmor,
\newblock From particle to kinetic and hydrodynamic descriptions of flocking.
\newblock {\em Kinetic and Related Models}, 1(3):415--435, 2008.

\bibitem{haskovec_2013}
J. Haskovec
\newblock Flocking dynamics and mean field limit of the Cucker-Smale-type model with topological interactions, 
\newblock {\em Phys. D}, 261: 42--51, 2013.

\bibitem{hegselmann_opinion_2002}
R.~Hegselmann and U.~Krause,
\newblock Opinion dynamics and bounded confidence: models, analysis and
  simulation.
\newblock {\em Journal of Artificial Societies and Social Simulation}, 5(3),
  2002.


\bibitem{He_2001} 
D. Helbing,
\newblock Traffic and related self-driven many particle systems.
\newblock {\em  Reviews of Modern Physics}, 73, 1067-1141, 2001.

\bibitem{Helbing_2009}
D.  Helbing,
\newblock Pattern formation, social forces, and diffusion instability in games with success-driven motion,
\newblock {\em Eur. Phys. J. B}, 67:345--356, 2009

\bibitem{He_2010}
D. Helbing
\newblock Quantitative Sociodynamics: Stochastic Methods and Models of Social Interaction Processes
\newblock Springer-Verlag,2010.

\bibitem{hemelrijk_self-organized_2008}
C.~K Hemelrijk and H.~Hildenbrandt,
\newblock Self-organized shape and frontal density of fish schools.
\newblock {\em Ethology}, 114(3):245--254, 2008.

\bibitem{Hubbard_West_1997} J. H. Hubbard and B. H. West,
\newblock Differential Equations: A Dynamical Systems Approach. 
\newblock Springer Verlag, 1997.

\bibitem{jackson_2010}
M. O. Jackson,
\newblock Social and Economic Networks.
\newblock Princeton University Press, 2010.

\bibitem{JLM_2003}
A. Jadbabaie, J. Lin, and A.S. Morse,
\newblock Coordination of groups of mobile autonomous agents using nearest neighbor rules.
\newblock {\em IEEE Trans. Automat. Control}, 48, 988-1001, 2003.

\bibitem{JE_2007} M. Ji and M. Egerstedt, 
\newblock Distributed coordination control of multi-agent systems while preserving
connectedness,
\newblock {\em IEEE Trans. Robot.}, vol. 23(4),  693--703, 2007.

\bibitem{JK_2010}
E.W. Justh and P.S. Krishnaprasad
\newblock Extremal Collective Behavior,
\newblock in {\em Proc. 49th IEEE Conf. Decision and Control}, 5432--5437, 2010.

\bibitem{KMT_2012}
T. Karper, A. Mellet and K. Trivisa,
\newblock Hydrodynamic limit of the kinetic Cucker-Smale flocking model
\newblock {\em ArXiv}:1205.6831, 2012.

\bibitem{KL_1993}
D. A. Kessler and H. Levine,
\newblock Pattern Formation in Dictyostelium via the 
Dynamics of Cooperative Biological Entities,
\newblock {\em  Phys. Rev. E} 48(6):4801--4804, 1993. 

\bibitem{krause_discrete_2000}
U.~Krause,
\newblock A discrete nonlinear and non-autonomous model of consensus formation.
\newblock {\em Communications in difference equations}, page 227--236, 2000.

\bibitem{KurRam2011}
S. Kurz and J. Rambau
\newblock On the Hegselmann-Krause conjecture in opinion dynamics
\newblock {\em J. Difference eqs. Appl}, 17(6):859--876, 2011.

\bibitem{merton_1954}
P. F. Lazarsfeld, R. K. Merton,
\newblock  Friendship as a social process: a substantive and methodological analysis.
\newblock In ``Freedom and Control in Modern Society'' (M. Berger, T. Abel, and C. H. Page, eds.) New York, Van Nostrand, 18-66, 1954.

\bibitem{LRC_2000}
H. Levine, W.-J. Rappel, I. Cohen,
\newblock  Self-organization in systems of self-propelled particles
\newblock {\em Phys. Rev. E} 63:017101, 2000.


\bibitem{Li_2008}
W. Li,
\newblock Stability analysis of swarms with general topology
\newblock {\em IEEE Trans. Systems, Man, Cyber.}, Part B, 38(4), 2008. 

\bibitem{LX_2010}
X. Li and J. Xiao,
\newblock Swarming in homogeneous environments: A social interaction based framework
\newblock {\em Journal Theor. Biology}, 264(3):747--759, 2010.

\bibitem{LT_2001}
H. Liu and E. Tadmor, 
\newblock Critical thresholds in convolution model for nonlinear conservation laws,
\newblock {\em SIAM  J. Math. Anal.} 33(4), 930-945, 2001.

\bibitem{lorentz_2007}
J.Lorenz, 
\newblock Continuous opinion dynamics of multidimensional allocation problems under bounded confidence.
A survey.
\newblock {\em Internat. J. Modern Phys. C}, 18, 1819-1838, 2007.

\bibitem{McPherson_2001}
M. McPherson, L. Smith-Lovin, and J. M Cook,
\newblock Birds of a Feather: Homophily in Social Networks.
\newblock {\em Annual Review of Sociology}, Vol. 27: 415-444 2001.

\bibitem{MOA_2010}
N. Mecholsky, E. Ott and T. M. Antonsen,
\newblock Obstacle and predator avoidance in a model for flocking
\newblock {Physica D}, 239, 988-996, 2010.

\bibitem{MD2012}
T. Mengesha and Q. Du,
\newblock Analysis of a scalar peridynamic model with a sign changing kernel
\newblock preprint

\bibitem{Mer_1994} R. Merris
\newblock Laplacian Matrices of Graphs: A Survey
\newblock {\em Linear Algebra Applications} 197,198, 143--176, 1994. 
 
\bibitem{mohar1991} B. Mohar, 
\newblock Eigenvalues, diameter, and mean distance in graphs, 
\newblock {\em Graph  and Combinatorics} 7:53--64, 1991.

\bibitem{motsch_tadmor_new_2011}
S.~Motsch and E.~Tadmor,
\newblock A new model for self-organized dynamics and its flocking behavior.
\newblock {\em Journal of Statistical Physics}, 144(5):923--947, August 2011.

\bibitem{OS_2006} R. Olfati-Saber,
\newblock Flocking for multi-agent dynamic systems: algorithms and theory.
\newblock {\em IEEE Trans. Auto. Control} 51(3), 401--420, 2006.

\bibitem{OSM_2004}
R. Olfati-Saber and R.M. Murray
\newblock Consensus problems in network
of agents with switching topology and time delays.
\newblock {\em IEEE Trans. on Automatic Control}  49(9):1520-1533,  2004.

\bibitem{OSFM_2007}
R. Olfati-Saber, J.A. Fax, and R.M. Murray,
\newblock Consensus and cooperation in networked multi-agent systems.
\newblock {\em Proc. IEEE}, 95, 215-233, 2007.

\bibitem{PT_2011}
B. Piccoli and A. Tosin,
\newblock Time-evolving measures and macroscopic modeling of pedestrian flow.
\newblock {Arch. Rat. Mech. Anal.}, 199, 707-738, 2011


\bibitem{reynolds_flocks_1987}
C.~W. Reynolds,
\newblock Flocks, herds and schools: A distributed behavioral model.
\newblock In {\em {ACM} {SIGGRAPH} Computer Graphics}, 21, 25-34, 1987.

\bibitem{Ring_2012}
C. Ringhofer
\newblock Traffic flow models and service rules for complex production systems,
\newblock {\em Decision Policies for Production Networks},  (K. Kempf, D. Armbruster eds), pp.209-233, Springer 2012

\bibitem{Scha_2007} S. E. Schaeffer
\newblock Survey graph clustering
\newblock {\em Computer Science Review}  1, 27--64, 2007.

\bibitem{ST_1992}
{S. Schochet and E. Tadmor}, 
\newblock  The regularized Chapman-Enskog expansion for scalar conservation laws,
\newblock {\em  Arch. Rational Mech. Anal.} 119, 95-107, 1992.

\bibitem{Shen_2007}
J. Shen
\newblock Cucker-Smale flocking under hierarchical leadership
\newblock {\em SIAM J. Appl. Math.}, 68(3):694--719, 2007. 


\bibitem{Sh_2011}
A. Shklarsh, G. Ariel,, E. Schneidman and E. Ben-Jacob
\newblock Smart swarms of bacteria-inspired agents with
performance adaptable interactions,
\newblock {\em PLoS Computational Biology} 7(9):1--11, 2011.


\bibitem{Sill2000}
S. A. Silling,
\newblock Reformulation of Elasticity Theory for Discontinuities and Long-Range Forces. 
\newblock {\em J. Mech. Physics of Solids}, 48,  175-209, 2000.

\bibitem{SBSC_2013}
J. L. Silverberg, M. Bierbaum, J. Sethna, and I. Cohen,
\newblock Collective motion of humans in mosh and circle pits at heavy metal concerts,
\newblock {\em Phys. Rev. Lett.}, 110: 228701, 2013.

\bibitem{spohn_1991}
H.~ Spohn,
\newblock Large Scale Dynamics of Interacting Particles,
\newblock {\em Texts and Monographs in Physics}, Springer 1991.
  
\bibitem{tadmor_2014}
E. Tadmor and C. Tan,
\newblock Critical thresholds in flocking  hydrodynamics with
  nonlocal alignment,
\newblock ArXiv1403.0991v1.

\bibitem{T-T_1998} 
J. Toner  and Y. Tu, 
\newblock Flocks, herds, and schools. A quantitative theory of flocking,
\newblock {\em Physical Review E}. 58, 4828-4858 (1998).

\bibitem{T-B_2004} 
C. M. Topaz and A. L. Bertozzi, 
\newblock Swarming patterns in a two-dimensional kinematic model for biological groups.
\newblock {\em SIAM J. Appl. Math.} 65, 152-174 (2004).

\bibitem{toscani_2006}
G. Toscani,
\newblock Kinetic models of opinion formation.
\newblock {\em Comm. Math. Sci.}, 4,481-496, 2006.

\bibitem{Tren_2013}
H. Trenchard
\newblock Peloton phase oscillations
\newblock {\em Chaos,Solitons \& Fractals}, 56: 194--201, 2013.

\bibitem{vicsek_novel_1995}
T.~Vicsek, A.~Czir\'o{k}, E.~{Ben-Jacob}, I.~Cohen, and O.~Shochet,
\newblock Novel type of phase transition in a system of self-driven particles,
\newblock {\em Physical Review Letters}, 75(6), 1226--1229, 1995.

\bibitem{vicsek_2012}
T.~Vicsek and A. Zefeiris,
\newblock Collective motion.
\newblock {\em Physics Reprints}, 517:71-140(2012).


\bibitem{WDA2005}
G. Weisbuch, G. Deffuant and F. Amblard,
\newblock Persuasion dynamics
\newblock {\em Phys. A}, 353:555--575, 2005.

\bibitem{Wei2006}
G. Weisbuch
\newblock Social opinion dynamics
\newblock in ``{\em Econophysics and Sociophysics: Trends and Perspectives}'' (B. K. Chakrabarti, A. Chakrabarti and A. Chatterjee, eds), Wiley, 2006, 67-94.


\bibitem{youseff_2008}
L. Youseff, A. Barbaro, P. Trethewey, B. Birnir, and J. Gilbet,
\newblock Parallel modeling of fish interaction,
\newblock {\em  Computational Science and Engineering}, 11th IEEE International Conference, 234--241, 2008.

\bibitem{ZP_2007} M. Zavlanos G.Pappas
\newblock Potential Fields for Maintaining Connectivity of Mobile Networks.
\newblock {\em EEE Transactions on Robotics} 23(4), 812--816, 2007.

\bibitem{ZEP_2011} M. Zavlanos, M. Egerstedt, and G. J. Pappas,
\newblock Graph theoretic connectivity control of mobile robot networks.  
\newblock {\em Proceedings of the IEEE}, 99(9):1525--1540, 2011.
\end{thebibliography}
\end{document}